\PassOptionsToPackage{unicode}{hyperref}
\PassOptionsToPackage{hyphens}{url}
\PassOptionsToPackage{usenames,dvipsnames,svgnames,x11names}{xcolor}
\documentclass[12pt]{article}

\usepackage{authblk}
\usepackage{amsmath,amssymb,amsthm,setspace}
\usepackage{iftex}
\ifPDFTeX
  \usepackage[T1]{fontenc}
  \usepackage[utf8]{inputenc}
  \usepackage{textcomp} 
\else 
  \usepackage{unicode-math}
  \defaultfontfeatures{Scale=MatchLowercase}
  \defaultfontfeatures[\rmfamily]{Ligatures=TeX,Scale=1}
\fi
\usepackage{lmodern}
\ifPDFTeX\else  
\fi
\IfFileExists{upquote.sty}{\usepackage{upquote}}{}
\IfFileExists{microtype.sty}{
  \usepackage[]{microtype}
  \UseMicrotypeSet[protrusion]{basicmath} 
}{}

\usepackage{xcolor}
\makeatletter
\ifx\paragraph\undefined\else
  \let\oldparagraph\paragraph
  \renewcommand{\paragraph}{
    \@ifstar
      \xxxParagraphStar
      \xxxParagraphNoStar
  }
  \newcommand{\xxxParagraphStar}[1]{\oldparagraph*{#1}\mbox{}}
  \newcommand{\xxxParagraphNoStar}[1]{\oldparagraph{#1}\mbox{}}
\fi
\ifx\subparagraph\undefined\else
  \let\oldsubparagraph\subparagraph
  \renewcommand{\subparagraph}{
    \@ifstar
      \xxxSubParagraphStar
      \xxxSubParagraphNoStar
  }
  \newcommand{\xxxSubParagraphStar}[1]{\oldsubparagraph*{#1}\mbox{}}
  \newcommand{\xxxSubParagraphNoStar}[1]{\oldsubparagraph{#1}\mbox{}}
\fi
\makeatother

\usepackage{longtable,booktabs,array}
\usepackage{calc} 
\usepackage{etoolbox}
\makeatletter
\patchcmd\longtable{\par}{\if@noskipsec\mbox{}\fi\par}{}{}
\makeatother
\IfFileExists{footnotehyper.sty}{\usepackage{footnotehyper}}{\usepackage{footnote}}
\makesavenoteenv{longtable}
\usepackage{graphicx}
\makeatletter
\def\maxwidth{\ifdim\Gin@nat@width>\linewidth\linewidth\else\Gin@nat@width\fi}
\def\maxheight{\ifdim\Gin@nat@height>\textheight\textheight\else\Gin@nat@height\fi}
\makeatother
\setkeys{Gin}{width=\maxwidth,height=\maxheight,keepaspectratio}
\makeatletter
\def\fps@figure{htbp}
\makeatother

\makeatletter
\@ifpackageloaded{caption}{}{\usepackage{caption}}
\AtBeginDocument{%
\ifdefined\contentsname
  \renewcommand*\contentsname{Table of contents}
\else
  \newcommand\contentsname{Table of contents}
\fi
\ifdefined\listfigurename
  \renewcommand*\listfigurename{List of Figures}
\else
  \newcommand\listfigurename{List of Figures}
\fi
\ifdefined\listtablename
  \renewcommand*\listtablename{List of Tables}
\else
  \newcommand\listtablename{List of Tables}
\fi
\ifdefined\figurename
  \renewcommand*\figurename{Figure}
\else
  \newcommand\figurename{Figure}
\fi
\ifdefined\tablename
  \renewcommand*\tablename{Table}
\else
  \newcommand\tablename{Table}
\fi
}
\@ifpackageloaded{float}{}{\usepackage{float}}
\floatstyle{ruled}
\@ifundefined{c@chapter}{\newfloat{codelisting}{h}{lop}}{\newfloat{codelisting}{h}{lop}[chapter]}
\floatname{codelisting}{Listing}

\makeatother
\makeatletter
\@ifpackageloaded{caption}{}{\usepackage{caption}}
\@ifpackageloaded{subcaption}{}{\usepackage{subcaption}}
\makeatother

\usepackage[margin=1in]{geometry}

\ifLuaTeX
  \usepackage{selnolig}  
\fi
\usepackage[]{natbib}
\usepackage{bookmark}
\usepackage{multirow}
\usepackage{cleveref}

\IfFileExists{xurl.sty}{\usepackage{xurl}}{} 
\hypersetup{
  pdftitle={Nonparametric regression of spatio-temporal data using infinite-dimensional covariates},
  pdfauthor={Subhrajyoty Roy; Soudeep Deb; Sayar Karmakar; Rishideep Roy},
  pdfkeywords={3 to 6 keywords, that do not appear in the title},
  colorlinks=true,
  linkcolor={blue},
  filecolor={Maroon},
  citecolor={Blue},
  urlcolor={Blue},
  pdfcreator={LaTeX via pandoc}}

\newcommand{\argmin}{\arg\,\min}

\newcommand{\R}{\mathbb{R}}
\newcommand{\Z}{\mathbb{Z}}
\newcommand{\N}{\mathbb{N}}
\newcommand{\E}{\mathbb{E}}
\newcommand{\prob}{\mathbb{P}}
\newcommand{\Tcal}{\mathcal{T}}
\newcommand{\Scal}{\mathcal{S}}

\newcommand{\Dcal}{\mathcal{D}}

\newcommand{\Hcal}{\mathcal{H}}
\newcommand{\Kcal}{\mathcal{K}}
\newcommand{\Fcal}{\mathcal{F}}
\newcommand{\Ocal}{\mathcal{O}}
\newcommand{\ppm}{\mathrm{PM}2.5}

\newcommand{\inner}[1]{\left\langle #1\right\rangle}
\newcommand{\bb}[1]{\boldsymbol{#1}}
\newcommand{\var}{\text{Var}}
\newcommand{\cov}{\text{Cov}}
\newcommand{\ind}[1]{\bb{1}_{#1}}
\newcommand{\indd}[1]{\bb{1}_{\{#1\}}}
\newcommand{\tr}{^\intercal}
\newcommand{\abs}[1]{\left\lvert #1 \right\rvert}
\newcommand{\norm}[1]{\left\| #1 \right\|}

\graphicspath{{../..}}
\usepackage{enumitem}
\usepackage[linesnumbered,ruled,vlined]{algorithm2e}
\SetKwInput{KwInput}{Input}                
\SetKwInput{KwOutput}{Output}              

\newtheorem{thm}{Theorem}
\newtheorem{lem}{Lemma}
\newtheorem{corr}{Corollary}
\newtheorem{prop}{Proposition}

\newcommand{\anon}{1}

\begin{document}

\def\spacingset#1{\renewcommand{\baselinestretch}%
{#1}\small\normalsize} \spacingset{1}


\if1\anon
{
  \title{\bf Nonparametric regression of spatio-temporal data using infinite-dimensional covariates}
  \author[1]{Subhrajyoty Roy}
  \author[2]{Soudeep Deb}
  \author[3]{Sayar Karmakar}
  \author[4]{Rishideep Roy}
  
  \affil[1]{{\small Department of Statistics and Data Science, Washington University at St. Louis, USA}}
  \affil[2]{{\small Decision Sciences Area, Indian Institute of Management Bangalore, India}}
  \affil[3]{{\small Department of Statistics, University of Florida, USA}}
  \affil[4]{{\small School of Mathematics, Statistics and Actuarial Science, University of Essex, UK}}
  
  \maketitle
} \fi

\if0\anon
{
  \bigskip
  \bigskip
  \bigskip
  \begin{center}
    {\LARGE\bf Nonparametric regression of spatio-temporal data using infinite-dimensional covariates}
\end{center}
  \medskip
} \fi

\bigskip
\begin{abstract}
In spatio-temporal analysis, we often record data at specific time intervals but with varying spatial locations between these timepoints. We propose a conditional model to analyze such spatio-temporal data that accommodates the dependencies alongside second-order stationary explanatory variables, which may be infinite-dimensional and accommodate spatio-temporal covariates. Because of the absence of a mixing-type dependence condition in this case, which is typically required by the existing studies, we consider a weaker polynomially decaying moment contraction (PMC) condition on the covariates. In this paper, we obtain nonparametric point estimates of the mean and covariate functions of such a regression model, which we then show to be statistically consistent. We also obtain a simultaneous confidence interval of the mean function using the central limit theorem for the proposed estimator. Such simultaneous inference tools can be used to test for certain specifications of the mean function. Some simulation studies and two real-data analyses have been illustrated to corroborate the findings.
\end{abstract}

\noindent%
\vfill

\section{Introduction}\label{sec:intro}

Spatiotemporal datasets increasingly arise from monitoring systems whose sampling designs evolve over time. At each timepoint, measurements are recorded at a collection of spatial sites that can change as sensors fail, new sensors are deployed, or access constraints vary. At the same time, analysts often observe explanatory information aggregated at the temporal level (e.g., meteorological fields, land-use surfaces, remote-sensing products, GPS tracking, or historical records), whose natural representation is functional or otherwise high or possibly infinite-dimensional. These features, namely irregular time-stamps, time-varying spatial sampling locations, and infinite-dimensional temporal-level covariates, create a setting in which classical spatiotemporal methods and existing functional regression approaches do not directly deliver nonparametric inference on the conditional mean and volatility surfaces. 

To better understand this framework, as the first practical use-case, we showcase our inference strategy on an air pollution dataset from New Delhi, the capital of India. This dataset exhibits irregularly spaced observations of $\ppm$ (airborne particles with a diameter smaller than 2.5 micrometers), with missing data across time due to inactivity of sensors: see \Cref{fig:complete_cases_pm25} which shows this missingness. Naturally, it demands an infinite-dimensional covariate process due to the growing nature of historical observations and factors like measurements of other pollutants, weather patterns observed in a different set but nearby locations, etc., thus highlighting the need for a model that allows for infinite-dimensional covariates.

\begin{figure}[!ht]
    \centering
    \includegraphics[width=0.7\textwidth,keepaspectratio]{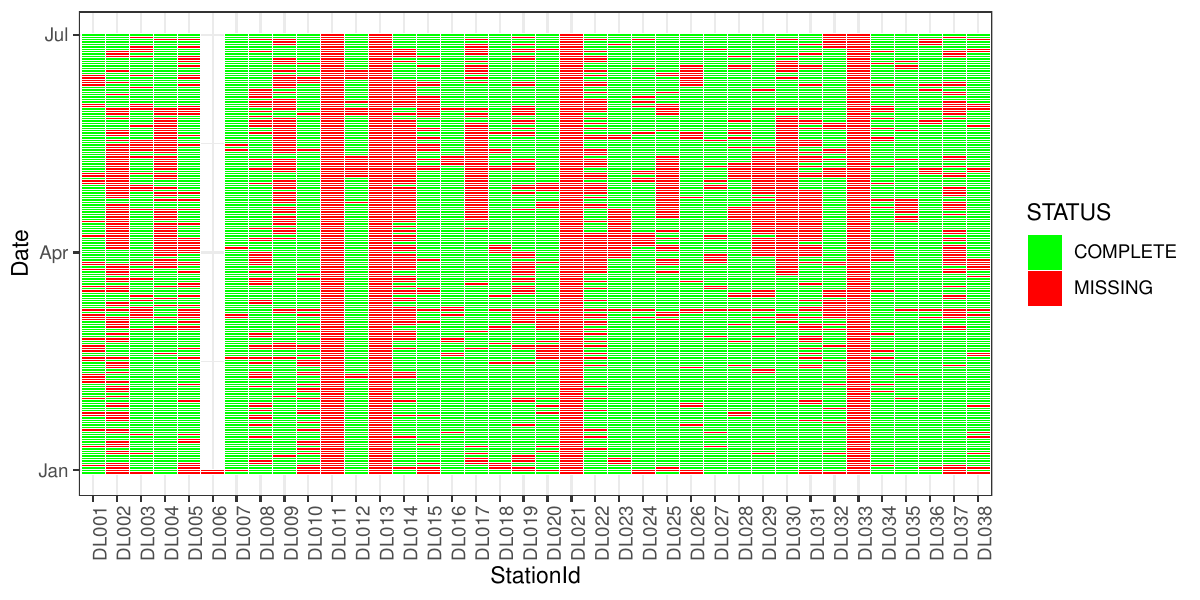}
    \caption{A heatplot indicating the missingness of the $\ppm$  observations across locations (columns) and timepoints (rows).}
    \label{fig:complete_cases_pm25}
\end{figure}

For the second motivating example, we turn our attention to a very different context that stems from the analysis of shots in soccer. In \Cref{fig:football-eda}, the locations of all shots taken in the matches against a specific team (Chelsea) by three different teams (Arsenal, Liverpool, and Manchester City) are illustrated. It is clear that the shots are irregularly distributed over the entire playing field, with no fixed spatial design. One may aim to analyze the quality of a shot, commonly measured through expected goals or xG \citep{mead2023expected}, based on various contextual variables over the season. These covariates are naturally indexed over time, and accumulate as the season progresses, yielding a high (or infinite) dimensional structure. Together, the irregular spatial sampling of shot locations and the complex temporal evolution of contextual information pose challenges for standard spatial or functional regression methods.

\begin{figure}[!ht]
    \centering
\includegraphics[width=0.8\textwidth,keepaspectratio]{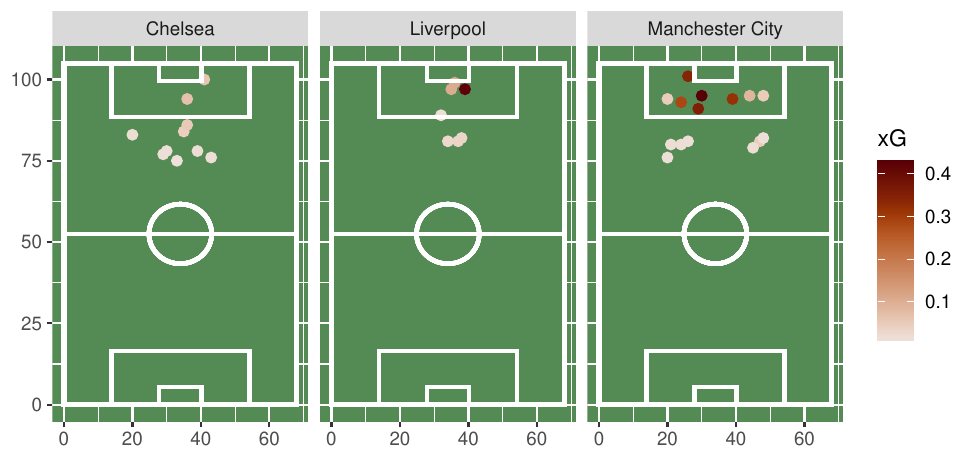}
    \caption{Position and expected goals (xG) of shots taken against Chelsea during the league 2014-15 by the players of Arsenal, Liverpool, and Manchester City.}
    \label{fig:football-eda}
\end{figure}

To define the mathematical framework for similar applications, let $Y_t(s) \in \R$ denote the observed response at time $t$ and location $s \in \Scal$, and let $\bb{X}_t$ be a second-order stationary covariate taking values in an abstract space $\chi \subseteq \R^\infty$. While we only keep the time index in the subscript to focus on the temporal evolution, we emphasize that $\bb{X}_t$, being infinite-dimensional, can accommodate spatio-temporal covariates and lagged responses by including explanatory variables across all locations. We formulate the conditional regression model 
\begin{equation}
    Y_t(s) = \mu(\bb{X}_t, s) + \sigma(\bb{X}_t, s) \, \epsilon_t(s).
    \label{eqn:model}
\end{equation}
where the timepoints $t$ index the temporal horizon $\Tcal$ and locations $s$ index the spatial horizon $\Scal$ which is considered to be a compact subset of $\R^d$. In real-life problems, $\Scal$ is often a compact subset of $\R^2$ or $\R^3$, depending on the dataset at hand. Without loss of generality, following the principles of in-fill asymptotics, we can assume $\Tcal = [0, 1]$ and $\Scal = [0, 1]^d$ for a fixed $d$ after suitably scaling the temporal and spatial domains. The mean function $\mu(\cdot, \cdot)$ and the covariance function $\sigma(\cdot, \cdot)$ are both defined on $\chi \times \Scal$. They are assumed to be measurable and sufficiently smooth. The error component $\epsilon_t(s)$ in \eqref{eqn:model} captures idiosyncratic variation with a nontrivial spatial correlation. Specifically, as temporal dependence is entirely captured by the covariates $\bb{X}_t$, we consider $\epsilon_t(s)$ to be temporally independent and identically distributed (iid) random variables satisfying $\E(\epsilon_t(s)) = 0$. They are also independent of the past covariates $\bb{X}_{t'}$ for any $t' \leqslant t$ and for all $s \in \Scal$. These will be discussed in detail in \Cref{sec:math-framework}.

Although the framework \eqref{eqn:model} models the response variable $Y_t(s)$ for     all timepoints $t \in \Tcal$ and spatial points $s \in \Scal$, in practice, the continuous curves are not typically observed. We assume that the data are collected at some irregularly spaced timepoints and at some specific spatial locations that may differ from one timepoint to another. Keeping that in view, we shall represent the dataset with observations from a total of $n$ timepoints as
\begin{equation}
    \Dcal = \left\{ \left(t_i, s_{t_i j}, \bb{X}_{t_i}, Y_{t_i}(s_{t_i j}) \right): 0 = t_0 < t_1 < \dots < t_{n-1} < t_n = 1, j \in \{ 1, 2, \dots, n_i \} \right\},
    \label{eq:dataset}
\end{equation}
where the $i^{th}$ timepoint $t_i$ observes data from $n_i$ locations. The goal is to draw nonparametric inference about $\mu(\cdot, \cdot)$ and $\sigma(\cdot,\cdot)$ given $\Dcal$, and consequently, use the model \eqref{eqn:model} to make predictions of the response function at unknown locations or at a future timepoint.

Our model framework can be regarded as a significant generalization of the conventional stochastic time series regression model \citep[see, e.g.,][]{zhao2008confidence} in two ways. First, we incorporate the infinite-dimensional covariate $\bb{X}_t$ into the model. Second, we extend the time series framework to a spatiotemporal framework, allowing a non-trivial spatial dependency structure. Of course, it provides an extremely general setting, and includes the popular autoregressive setups (e.g. AR, ARCH) by letting $\bb{X}_t$ to be a spatially aggregated function of the lagged response values. It is easy to note that allowing $\bb{X}_t$ to be infinite-dimensional automatically includes the finite-dimensional case of the covariates $\bb{X}_t \in \R^d$ for some fixed $d > 0$, by choosing the coordinates of $\bb{X}_t$ to be equal to $0$ beyond the fixed dimension $d$.

It is imperative to note that the existing literature addresses only fragments of our setting. Functional regression in an iid setting may accommodate infinite-dimensional covariates but not temporal dependence; functional time series methods accommodate dependence but typically assume regular time grids and do not treat spatially indexed responses observed on time-varying location sets; and nonparametric spatiotemporal regression focuses on finite-dimensional covariates and regular grids. We bridge this research gap and develop frequentist nonparametric estimation and inference for the mean and volatility surfaces in \eqref{eqn:model} under the irregular design \eqref{eq:dataset} with possibly infinite-dimensional temporal-level covariates. Our key contribution in this regard is discussed in more detail next, in light of some existing literature.

\subsection{Existing literature and our contribution}\label{sec:literature}

In functional data analysis, models similar to \eqref{eqn:model} have been studied to some extent, where an infinite-dimensional covariate process is used to model the response variable. The readers are referred to \cite{ferraty2003curves,ferraty2004nonparametric} for a general treatment of univariate responses, and \cite{xiang2013nonparametric,omar2019nonparametric} for a multivariate response setup. However, all these approaches consider each datapoint to be iid without any regard to the temporal dependence. Recently, \cite{li2023statistical} addressed this gap by considering a temporally dependent functional data setup, but the observations are presented at regular intervals of time and space. We relax that by allowing the dataset \eqref{eq:dataset} to have irregularly spaced points in time and varying locations in space.

Nonlinear spatiotemporal regression model is far less studied compared to the extensive literature on the nonlinear analysis of time series or the literature on spatial models \citep{cressie2011statistics}. Much success of spatiotemporal modeling has instead come from a Bayesian point of view \citep[see][for a detailed review]{haining2020modelling}. Among the frequentist viewpoint through nonparametric estimation, \cite{wang2009estimation} solved a related problem of mean (or trend) estimation using a local linear estimation technique. A completely nonparametric estimation was attempted by \cite{yang2018spatiotemporal} for modeling data on incidence rates of diseases. While both studies considered regularly spaced timepoints, the case for irregularly spaced timepoints was investigated by \cite{alsulami2017estimation}, who proposed a semi-parametric approach. However, these studies did not allow the covariate process to be infinite-dimensional. On that note, perhaps the closest article to our approach is the work by \cite{hong2020nonparametric}, who considered infinite-dimensional covariates but only in a time series setup. 

While the existing works on nonparametric modeling of spatiotemporal data and functional data analysis with infinite-order regression have grown into independent sets of literature, our proposed model \eqref{eqn:model} aims to bring a novel perspective by combining these two ideas. By bridging these, it allows us to efficiently estimate the trend function using even irregularly spaced observations in space and time, which had only limited attention in the past. Furthermore, we completely avoid any mixing-type assumption, which is a standard practice to model data dependence across space and time. In fact, \cite{hong2020nonparametric} illustrated that, since $\bb{X}_t$ is infinite-dimensional, it is quite possible that it contains infinite order lags of some exogenous variable, say $\bb{Z}_t$ such that $\bb{X}_t = \sum_{j=0}^\infty a_j \bb{Z}_{t-j}$ for some scalars $a_j \in \R$. This is not $\alpha$-mixing in general. To avoid such scenarios, they considered a near-epoch dependence (NED) condition on the covariates $\bb{X}_t$, but still require an $\alpha$-mixing condition on $(Y_t, \bb{Z}_t)$ jointly. However, we show that it is possible to completely avoid any mixing-type conditions and obtain valid inference under a polynomially-decaying moment contraction (PMC) assumption. It is much weaker than the standard geometric moment contraction (GMC) condition, and encompasses various processes which are not $\alpha$ or $\beta$-mixing \cite[see][for a counter-example]{chen2016selfnormalized}. Our requirement is more general than NED conditions, and can be easier to work with due to lesser restrictions on the moments \citep[Remark 5.1,][]{shao2007asymptotic}. As a result, in addition to allowing stronger than geometric dependence structure, we are able to derive asymptotic validity results of our proposed estimator under only a second-order stationarity condition on the covariate $\bb{X}_t$. This improves upon the requirement of the existence of $(2+\delta)$-order moment ($\delta > 0$) for $Y_t$ and fourth order moment condition for $\bb{X}_t$ required by the analysis of~\cite{hong2020nonparametric}. In summary, the major contributions of our work are threefold: 
\begin{itemize}
    \item[(a)] We propose a first-of-its-kind nonlinear regression model with infinite-dimensional covariates for spatiotemporally dependent data and its nonparametric estimation from irregularly sampled observations; 
    
    \item[(b)] We replace the mixing-type dependence condition with a functional dependence assumption and allow for significantly slow decay of dependence through PMC assumption.  
    
    \item[(c)] We reduce the higher-order moment assumptions, which, in turn, allows for a slightly heavier-tailed distribution of covariates. 
\end{itemize}

The methodology and theoretical results are presented in \Cref{sec:methods} and \Cref{sec:theoretical} below. In \Cref{sec:realdata}, we demonstrate the use of the proposed methodology in the real life dataset of air pollution from India. The paper concludes with some important remarks in \Cref{sec:conclusion}. Additional theoretical discussions, proofs, another real life example with soccer data, and a simulation study are deferred to the supplementary material in the interest of space.

\subsection{Notations}\label{sec:notations}

For better understanding of the technical discussions, we collect some notations that will be used throughout the paper. The sets of real numbers, natural numbers, and integers are denoted respectively by $\R,\N,\Z$. For two sequences $a_n$ and $b_n$, $a_n = \Ocal(b_n)$ denotes that for all sufficiently large $n$, $a_n$ is bounded between $cb_n$ and $Cb_n$ for two positive real constants $c < C < \infty$. The notation $a_n = o(b_n)$ indicates that $\limsup\limits_{n\rightarrow \infty} (a_n/b_n) = 0$. Correspondingly, $\Ocal_{\prob}(\cdot)$ and $o_{\prob}(\cdot)$ denote the probabilistic versions, where the analytical convergence and boundedness are respectively replaced by convergence in probability and stochastic boundedness. 

\section{Nonparametric estimation}\label{sec:methods}

\subsection{Mathematical Framework}\label{sec:math-framework}

Before we discuss the proposed estimator, it is critical to define all components and conditions of model \eqref{eqn:model}. We emphasize, and shall explicate properly as required later, that these conditions do not affect the generality of the spatiotemporal model to a large extent.

As the covariate process $\bb{X}_t$ is allowed to be infinite-dimensional, the abstract space of its range $\chi$ must be well-structured to allow necessary mathematical operations to be performed. We assume that $\chi$ is an affine subspace of a weighted $L^2$ Banach space, namely, for any $\bb{x}, \bb{x}' \in \chi$, the difference $(\bb{x} - \bb{x}') \in \chi$, and, $\chi \subseteq \bb{D}\mathcal{L}^2$, where 
\begin{equation}
    \bb{D}\mathcal{L}^2 = \left\{ (z_1, z_2, \dots) \in \R^\infty : \sum_{i=1}^\infty \zeta_i^{-2}z_i^2 < \infty \right\}
    \label{eqn:scaled-l2-space}
\end{equation}
for some fixed choice of sequence $\{ \zeta_i \}_{i=1}^\infty$. Let us use $\bb{D}$ to denote an infinite-dimensional linear operator that multiplies the $i^{th}$ coordinate by $\zeta_i$ for all $i$. Then, the above restriction allows one to formally define scaled norms for the infinite-dimensional covariates in $\R^\infty$, and induce a scaled $L^2$ metric $\norm{\bb{D}^{-1}(\bb{x} - \bb{x}')}$ for any $\bb{x}, \bb{x}' \in \chi$. In addition, we consider a PMC condition to model the dependence structure of the covariate process $\{ \bb{X}_t\}_{t \in \Tcal}$. Specifically, there exists iid random variables $\{Z_i\}_{i \in \Z}$ and measurable functions $G_t : \R^\infty \rightarrow \chi$ satisfying
\begin{equation*}
    \bb{X}_t = G_t\left( \dots, Z_{t-1}, Z_t \right).
\end{equation*}
For an independent and iid copy $Z_i^\ast$ for $i \in \Z$, assume that the quantity
\begin{equation}
    \theta_2(m) = \sup_{t} \left\Vert \bb{D}^{-1} \left( \bb{X}_t - G_t\left( \dots, Z_{t - m-1}, Z^\ast_{t-m}, Z_{t-m+1}, \dots, Z_t \right)\right) \right\Vert_{2} = \Ocal(m^{-\tau}),
    \label{assum:gmc-Xt}
\end{equation}
for some $\tau > 2$. Equation~\eqref{assum:gmc-Xt} controls the temporal dependence structure through the widely used notion of functional dependence measures \citep[see][]{wu2005nonlinear, wu2011asymptotic, wu2011gaussian, karmakar2020optimal}. The PMC condition also implies that $\Delta_2(m) = \Ocal(m^{-\tau+1})$, where 
\begin{equation*}
    \Delta_2(m) = \sup_{t} \left\Vert \bb{D}^{-1} \left(\bb{X}_t - G_t(\dots, Z^\ast_{t-m-2}, Z^\ast_{t-m-1}, Z^\ast_{t-m}, Z_{t-m+1}, \dots ,Z_t ) \right) \right\Vert.
\end{equation*}
This is a generalization of $\theta_2(m)$ that considers the changes in $\bb{X}_t$ due to the change in all historical values up to a time lag of $m$. 

Next, assume that the error terms satisfy a temporal homogeneity condition 
\begin{equation}
    \cov(\epsilon_{t_1}(s_1), \epsilon_{t_2}(s_2)) = \rho(s_1, s_2) = \Ocal\left( \Vert s_1 - s_2\Vert^{-(d+\delta)} \right),
    \label{eqn:cov-y-rho}
\end{equation}
for all $s_1, s_2 \in \Scal$, where $\rho(\cdot,\cdot)$ is a symmetric function satisfying $\rho(s, s) = 1$ for any $s \in \Scal$ and $d$ is the dimension of the spatial horizon $\Scal$. $\delta > 0$ is some small positive constant. Typically, $\rho(s_1, s_2) = \rho^\ast(\norm{s_1 - s_2})$ with $\rho^\ast(0) = 1$, and the decay rate in~\eqref{eqn:cov-y-rho} simply restricts the $\rho^\ast$ to be a polynomially decaying function of its argument. This connection immediately simplifies the conditional variance structure of the data $Y_t(s)$ by separating it into two parts: one involving temporal correlation governed by the covariate $\bb{X}_t$ and another involving spatial correlation governed by the $\rho(\cdot, \cdot)$ function, i.e., we can write $\cov(Y_{t_1}(s_1), Y_{t_2}(s_2)) = \sigma(\bb{X}_{t_1}, s_1) \sigma(\bb{X}_{t_2}, s_2) \rho(s_1, s_2).$

We have also indicated that in model \eqref{eqn:model}, the mean and covariance functions are sufficiently smooth. To add some mathematical formalism, we assume that for each $\bb{x} \in \chi$, both $\mu(\bb{x}, \cdot)$ and $\sigma(\bb{x}, \cdot)$ are elements of a Hilbert space $\Hcal$ of functions from $\Scal$ to $\R$, equipped with the inner product $\inner{f, g} = \int_{\Scal} f(s)g(s)ds$. The space $\Hcal$ can be represented as
\begin{equation*}
    \Hcal =  \left\{ \sum_{i=1}^\infty a_i b_i(s): a_i \in \R, \text{ and } \sum_{i=1}^\infty \vert a_i\vert < \infty \right\},
\end{equation*}
where $\{(b_i: \Scal \rightarrow \R)\}_{i \in \N}$ is a set of countable orthonormal basis functions of $\Hcal$. This suggests a representation of the mean and covariance functions in the form
\begin{equation}
    \mu(\bb{x}, s) = \sum_{k=1}^\infty \mu_k(\bb{x})b_k(s), \;
        \sigma(\bb{x}, s) = \sum_{k=1}^\infty \sigma_k(\bb{x})b_k(s),
        \label{eqn:mean-sigma-basis}
\end{equation}
such that $\sum_{k=1}^\infty \vert \mu_k(\bb{x})\vert$ and $\sum_{k=1}^\infty \vert \sigma_k(\bb{x})\vert$ both exist and are finite for each $\bb{x} \in \chi$. Such representation has appeared before for spatiotemporal modeling \citep{wikle1999dimension}, and it is standard in the functional PCA literature due to Karhunen-Lo\'{e}ve expansion \citep{karhunen1946,loeve1946}, where the basis functions are chosen to correspond with the eigenfunctions of covariance-operator~\citep{yao2005functional}. The basis functions in \eqref{eqn:mean-sigma-basis} have also been extensively used in dependent data modeling, in the form of bi-square functions \citep{cressie2008fixed}, splines \citep{goldsmith2012longitudinal}, wavelets \citep{chen2015optimal}, etc. We refer to \cite{cressie2022basis} and the references therein for a comprehensive exposition. The particular choice of the basis function often has to be paired with the choice of the hyperparameter $K_\epsilon$ as in Theorem~\ref{thm:pointwise-full}, to maintain a consistent level of accuracy. In a practical setting, one may use any orthogonal set of basis functions for computational convenience, and the corresponding $K_\epsilon$ is picked out via a cross-validation approach; see Section~\ref{sec:realdata} for our specific choice.

While this takes care of a smoothness restriction over the spatial locations $\Scal$, to ensure further smoothness over the set $\chi$, we assume that the component functions $\mu_k(\bb{x})$ and $\sigma_k(\bb{x})$, for any $k \in \N$, are continuous on $\chi$ with respect to the scaled distance metric introduced in  \eqref{eqn:scaled-l2-space}. Moreover, these conditions imposed on $\mu(\cdot,\cdot)$ and $\sigma(\cdot,\cdot)$ in \eqref{eqn:mean-sigma-basis} are not usually very restrictive. Since $\Scal = [0, 1]^d$ is a Hausdorff topological space, and is compact, there exist choices of universal kernel functions such that their linear span is able to uniformly approximate any continuous function \citep{micchelli2006universal}. Hence, by choosing the basis of the corresponding Hilbert space spanned by those universal kernels, it follows that the assumption of existence of such representations given in \eqref{eqn:mean-sigma-basis} becomes simply a rephrasing of a continuity-type smoothness condition for $\mu(\bb{x}, \cdot)$ and $\sigma(\bb{x},\cdot)$.

\subsection{Estimation methodology}\label{sec:mean-estimation}

We begin with the nonparametric estimation of the mean component for our nonlinear regression model: $Y_t(s) = \mu(\bb{X}_t, s) + \sigma(\bb{X}_t, s) \, \epsilon_t(s)$. Because the mean function has a decomposition as shown in \eqref{eqn:mean-sigma-basis}, it is enough to estimate the component functions $\mu_k(\bb{x})$ for each $k \in \N$.  In order to proceed further, we shall need two assumptions about the boundedness of the expected covariate and the error components, given as follows.

\begin{enumerate}[label = (A\arabic*), ref = (A\arabic*)]
    \item\label{assum:mean-sigma-bound} For each $s \in \Scal$, define $\tilde{\mu}(s) =  \E(\vert \mu(\bb{X}_0, s)\vert)$ and $\tilde{\sigma}^2(s) =  \E(\sigma^2(\bb{X}_0, s))$ where the expectations are taken over the marginal distribution of the covariate $\bb{X}_0$. Assume that both $\tilde{\mu}(s)$ and $\tilde{\sigma}^2(s)$ are continuous and are integrable on $\Scal$.
    \item\label{assum:bounded-errors} The first and second order moments of the errors $\epsilon_t(s)$ are uniformly bounded, i.e., $\sup_{t \in [0, 1], s \in \Scal} \E(\epsilon_t^2(s)) < \infty$.
\end{enumerate}

Assumption \ref{assum:bounded-errors} is a technical assumption that controls the behaviour of the errors and allows us to perform interchangeability of integral and expectation operators later on. It is usually very weak: for example, if the second order moment of $\epsilon_t(s)$ exists and $\E(\epsilon_t^2(s))$ is a continuous function of $t$ and $s$, then by compactness of $\Tcal = [0, 1]$ and $\Scal$, the assumption~\ref{assum:bounded-errors} follows naturally. As a consequence of Assumption~\ref{assum:bounded-errors}, we also have $\sup_{t \in [0, 1], s \in \Scal} \E(\vert \epsilon_t(s)\vert) < \infty$. Additionally, note that since $\tilde{\mu}(s)$ and $\tilde{\sigma}^2(s)$ as in Assumption~\ref{assum:mean-sigma-bound} are continuous on the compact set $\Scal$, they are also uniformly bounded. For notational convenience, let us indicate $M$ as a generic constant which serves as the uniform bound for all these quantities. Therefore, it follows by an application of Cauchy-Schwarz inequality that
\begin{equation*}
    \int_{\Scal} \E \left[\abs{\sigma(\bb{X}_t, s) \epsilon_t(s) b_k(s)}\right] ds
    \leqslant \int_{\Scal} M \E \left[\abs{\sigma(\bb{X}_t, s) b_k(s)}\right] ds \leqslant M \left(\int_{\Scal} b_k^2(s) ds \right) \left( \int_{\Scal} \E(\sigma^2(\bb{X}_t, s)) ds \right).
\end{equation*}
The first term $\int_{\Scal} b_k^2(s)ds$ is equal to $1$ as $b_k(s)$ is an orthonormal basis function of $\Hcal$. For the second term, note that due to the second-order stationarity of $\bb{X}_t$, it is equal to $\int_{\Scal} \E(\sigma^2(\bb{X}_0, s))ds$, which is also bounded due to Assumption~\ref{assum:mean-sigma-bound}.

Now, starting with the decomposition of the mean function as in \eqref{eqn:mean-sigma-basis}, an application of Fubini's theorem yields that
\begin{equation*}
     \int_{\Scal} Y_t(s) b_k(s) \, ds 
    = \int_{\Scal} \sum_{i=1}^\infty \mu_i(\bb{X}_t) b_i(s) b_k(s) \, ds + \int_{\Scal} \sigma(\bb{X}_t, s) \epsilon_t(s)b_k(s) \, ds.
\end{equation*}
Following the orthonormality of the basis functions, the above can be simplified as
\begin{equation}
    \int_{\Scal} Y_t(s) b_k(s)ds = \mu_k(\bb{X}_t) + \eta_{tk},
    \label{eqn:mean-reduce-pre}
\end{equation}
where $\eta_{tk} = \int_{\Scal} \sigma(\bb{X}_t, s) \epsilon_t(s)b_k(s)\, ds$. Subsequently, we may write
\begin{equation*}
    \E(\eta_{tk}) = \E\left( \int_{\Scal} \sigma(\bb{X}_t, s) \epsilon_t(s)b_k(s) \, ds \right)
    = \int_{\Scal} b_k(s) \E(\sigma(\bb{X}_t, s))\E(\epsilon_t(s)) \, ds = 0,
\end{equation*}
where we use the independence of $\epsilon_t(s)$ and $\bb{X}_t$. Also, by another application of Fubini's theorem (valid due to Assumption~\ref{assum:bounded-errors}), we get
\begin{align*}
    \cov(\eta_{t_1k}, \eta_{t_2k})
     = \E(\eta_{t_1 k} \eta_{t_2 k}) 
    & = \E\left[ \int_{\Scal}\int_{\Scal} \sigma(\bb{X}_{t_1}, s) \sigma(\bb{X}_{t_2}, s') \epsilon_{t_1}(s)\epsilon_{t_2}(s')b_k(s)b_k(s') \, ds \, ds' \right]\\
    & = \int_{\Scal^2} \E\left[\sigma(\bb{X}_{t_1}, s) \sigma(\bb{X}_{t_2}, s')\right] \rho(s, s') b_k(s) b_k(s') \, ds \, ds'.
\end{align*}

\noindent Hence, for every $k \in \N$, we get a reduced temporal model
\begin{equation}
    Y^\ast_{tk} =  \int_{\Scal} Y_t(s)b_k(s) \, ds = \mu_k(\bb{X}_t) + \eta_{tk},
    \label{eqn:model-reduced}
\end{equation}
where $\eta_{tk}$ is a zero-mean second-order error process. Now, fix any $k \in \N$. The reduced model shown in \eqref{eqn:model-reduced} is a nonlinear regression model for dependent data, but with infinite-dimensional covariates $\bb{X}_t$. Hence, ideally, one may estimate the mean component functions $\mu_k(\bb{x})$ using local constant estimators, by suitably replacing the aggregated response $Y^\ast_{tk}$ (which is unobserved) by an estimate based on the dataset $\Dcal$ as detailed below.

It is worthwhile to note that, although the errors $\epsilon_t(s)$ are independent of the covariates $\bb{X}_t$, that is not the case for the aggregated error $\eta_{tk}$, as its variance depends on $\bb{X}_t$ through the $\sigma(\cdot, \cdot)$ function. However, it still follows that 
\begin{equation*}
    \E\left( \eta_{tk} \mid \{ \bb{X}_{t'}: t' \leqslant t \} \right)
    = \int_{\Scal} b_k(s)\sigma(\bb{X}_t, s) \E\left( \epsilon_t(s) \mid \{ \bb{X}_{t'}: t' \leqslant t \} \right) \, ds
    = 0,
\end{equation*}
for any $t \in \Tcal$. As we shall show below, this conditional mean-zero property is usually enough to ensure desirable theoretical properties of the proposed nonparametric estimators. We next focus on estimating each coefficient $\mu_k(\cdot)$ in \eqref{eqn:model-reduced} using standard nonparametric methods. For that, a kernel function of the following form is introduced:
\begin{equation*}
    \Kcal^\ast \left(\bb{X}_t, \bb{x}\right) = \Kcal \left(\bb{H}_n^{-1}(\bb{X}_t - \bb{x})\right) = K\left( \Vert{\bb{V}^{-1} (\bb{X}_t - \bb{x})}\Vert/h_n\right),
\end{equation*}
where $\bb{H}_n$ is a bandwidth matrix of the form $h_n \bb{D}$, $\bb{D}$ being a diagonal matrix in line with the linear operator defined in \eqref{eqn:scaled-l2-space}. Here, we also make use the scaled distance function as in \eqref{eqn:scaled-l2-space} to ensure that the above expressions are well-defined. Therefore, it is possible to use a univariate kernel function $K(\cdot)$ that acts upon the norm of the scaled distance vector scaled by an appropriate choice of the bandwidth $h_n$. The specific conditions on the choice of the kernel function $K(\cdot)$ are discussed in \Cref{sec:kernel-choice}.

We can now provide an estimator of $\mu_k(x)$ as 
\begin{equation}
    \widehat{\mu}_k(x) 
    =  \dfrac{\sum_{i=1}^n \Kcal(\bb{H}_n^{-1}(\bb{x} - \bb{X}_{t_i}) ) \widehat{Y}_{t_ik}^\ast }{\sum_{i=1}^n \Kcal(\bb{H}_n^{-1}(\bb{x} - \bb{X}_{t_i}) ) }
    = \dfrac{\sum_{i=1}^n K\left( \norm{\bb{H}_n^{-1}(\bb{x} - \bb{X}_{t_i})} \right) \widehat{Y}_{t_i k}^\ast }{\sum_{i=1}^n K\left(\norm{\bb{H}_n^{-1}(\bb{x} - \bb{X}_{t_i})} \right) }
    \label{eqn:mean-estimator}
\end{equation}
where $\widehat{Y}_{t_i k}^\ast$ is an estimate of $Y_{t_i k}^\ast$ based on dataset $\Dcal$. Although the usual Monte Carlo method can be used to estimate the aggregated response as $Y_{t_i k}^\ast \approx \frac{1}{n_i}\sum_{j=1}^{n_i}Y_{t_i}(s_{t_i j}) b_k(s_{t_i j})$, it has one serious drawback. If we assume that the choice of the spatial points $\{ s_{t_i1}, \dots, s_{t_i n_i} \}$ are deterministic, then it is possible that they may concentrate on one region of space for any particular time period. The simple Monte Carlo procedure in this case will provide a biased estimate. On the other hand, if we assume that the choices of spatial points are uniformly at random, then the covariance of the $Y$-values in these locations does not decay sufficiently fast with the distance between two spatial points. Due to the presence of this covariance, the Monte Carlo estimate may end up having infinite variance, thus leading to non-consistency of the estimator. Gathering insights from the functional data analysis, we subdivide the spatial horizon $\Scal$ into small grids, and choose a representative point on each hypercube of the grid. Then, the Monte-Carlo procedure is applied for this reduced set of representative points instead of all the points. This ensures a uniform distribution of the considered spatial points and also allow us to control the covariance between representative points of two distant grids. The final algorithm to obtain $\widehat{Y}_{t_i k}^\ast$ is illustrated in Algorithm~\ref{alg:modified-montecarlo}. A correctness guarantee that this procedure produces a consistent estimate of $Y_{t_i k}^\ast$ is given in Proposition~\ref{lem:Yk-convergence}.

\begin{algorithm}[!ht]
    \DontPrintSemicolon
    \KwInput{$t_i \in \Tcal$, $\{ s_{t_i j} \}_{j=1}^{n_i} \subseteq \Scal$, Basis index $k \in \Z^+$ and response $\{ Y_{t_i}(s_{t_i j})\}_{j=1}^{n_i}$}
    Let $\epsilon^\ast_{t_i} \leftarrow \inf\left\{ \epsilon > 0, \text{ and } \Scal \subseteq \cup_{j=1}^{n_i} B(s_{t_i j}, \epsilon) \right\}$, the ``effective spatial resolution''.\;
    Let $H_1, H_2, \dots, H_{r_i}$ be hypercubes from the grid of $[0, 1]^d$ with diameters equal to $\epsilon_{t_i}^\ast$\;
    Let $s^c_l \leftarrow \text{Center}(H_l)$ for each $l = 1, \dots, r_i = (\epsilon^\ast_{t_i} / \sqrt{d})^{-d}$.\;
    Let $s_{t_i,j_l} \leftarrow \argmin_{j} \Vert s^c_l - s_{t_i j}\Vert$ for each $l = 1, 2, \dots, r_i$.\;
    Obtain the final estimate as
    \vspace{-0.2in}\begin{equation}
        \widehat{Y}_{t_i k}^\ast = \dfrac{1}{r_i} \sum_{l=1}^{r_i} Y_{t_i}(s_{t_i j_l})b_k(s_{t_i j_l})
        \label{eqn:Yk-estimate}
    \end{equation}
    \vspace{-0.2 in}
    \caption{Modified Monte-Carlo procedure for estimating $Y_{t_i,k}^\ast$}
    \label{alg:modified-montecarlo}
\end{algorithm}

\vspace{-0.3 in}
\subsection{Choice of the Kernel and the Bandwidth}\label{sec:kernel-choice}

Before diving into technical results on the aforementioned estimation procedure, it is meaningful to state the underlying assumptions concerning the choice of the kernel and the bandwidth parameters. We restrict our choice to the type-I kernels as motivated by~\cite{hong2020nonparametric}. The kernel function $K: \R^+ \rightarrow \R^+$ is bounded and has bounded support, i.e.,  
\begin{equation}
    \int K(u)du = 1, \ C_1\bb{1}_{[0, \lambda]}(u) \leqslant K(u) \leqslant C_2\bb{1}_{[0, \lambda]}(u), 
    \label{eqn:kernel-bounds}
\end{equation}
for some positive finite constants $C_1, C_2$ and $\lambda$. The uniform kernel is one such example of a type-I kernel. Due to this, we can now represent the small ball probability using
\begin{equation}
    \phi_{\bb{x}}(h_n u) =  \prob\left(\Vert { \bb{H}_n^{-1}(\bb{X}_t - \bb{x})}\Vert \leqslant u \right) = P(\Vert{ \bb{D}^{-1}(\bb{X}_t - \bb{x})}\Vert \leqslant h_n u) = P(\bb{X}_t \in B(\bb{x}, h_n u)),
    \label{eqn:small-ball-prob}
\end{equation}
where $B(\bb{x}, h_n u)$ is the infinite-dimensional hyper ellipsoid centered at $\bb{x} \in \chi \subseteq \R^\infty$ with different axis lengths $2h_n\eta_j u$ at the $j^{th}$ direction. Note that, due to the stationarity of $\bb{X}_t$, the quantity $\phi_{\bb{x}}(h_n \lambda)$ is free of the time index $t$. Furthermore, regarding the choice of the bandwidth, we make the following assumption.

\begin{enumerate}[label = (A\arabic*), ref=(A\arabic*)]
    \setcounter{enumi}{2}    
    \item\label{assum:bandwidth-decay} The sequence of bandwidth parameter $h_n$ satisfies $h_n \rightarrow 0$, $\phi_{\bb{x}}(h_n u) \rightarrow 0$, and $nh_n^2\phi_{\bb{x}}^2(h_n u) \rightarrow \infty$ as $n \rightarrow \infty$ for any fixed $u > 0$ where $\phi_{\bb{x}}(h_n u)$ is the small-ball probability defined in \eqref{eqn:small-ball-prob}. Additionally, the joint small-ball probability defined below in \eqref{eqn:joint-ball-prob} satisfies $\phi_{\bb{x},\vert t-t'\vert}(h_n u, h_n u) = \Ocal(\phi_{\bb{x}}^2(h_nu))$ for any $t \neq s$.
\end{enumerate}

Assumption~\ref{assum:bandwidth-decay} ensures that while the size of the neighborhood and the small-ball probability $\phi_{\bb{x}}(h_n u)$ diminishes asymptotically, the number of observations in this neighborhood $B(\bb{x},h_n u)$ increases at a rate more than $n^{1/2}h_n^{-1}$. In the second part of Assumption~\ref{assum:bandwidth-decay}, we consider the joint small-ball probability for any $t \neq t'$,
\begin{equation}
    \phi_{\bb{x}, \vert t-t'\vert}(h_n u,h_n v) =  
    \prob\left( \Vert{\bb{H}_n^{-1}(\bb{X}_t - \bb{x}) }\Vert \leqslant u, \Vert{\bb{H}_n^{-1}(\bb{X}_{t'} - \bb{x}) }\Vert \leqslant v \right),
    \label{eqn:joint-ball-prob}
\end{equation}
which again, by the second-order stationarity of $\bb{X}_t$, depends only on the difference in time indices $\vert t- t'\vert$. By bounding this joint small-ball probability, Assumption~\ref{assum:bandwidth-decay} ensures only a localized dependence in the joint distribution of $(\bb{X}_t, \bb{X}_{t'})$ for $t \neq t'$. This assumption is the same as Assumption B.5 of \cite{hong2020nonparametric}, which, as we show below, helps to control the cross-product moments of the kernel functions. An obvious consequence of Assumption~\ref{assum:bandwidth-decay} is that as $n \to \infty$,
\begin{equation}
    n \phi_{\bb{x}}(h_n u) = \sqrt{nh^2_n \phi^2_{\bb{x}}(h_n u)} \times \dfrac{\sqrt{n}}{h_n} \rightarrow \infty.
    \label{eqn:bandwidth-density-inf}
\end{equation}

The condition~\ref{assum:bandwidth-decay} along with type-I kernel is key to obtaining various insights on the asymptotic nature of the kernel function. Starting with the expected value of the kernel function, note that $E( \Kcal( \bb{H}_n^{-1}(\bb{X}_t - \bb{x}))) = \E( K(\Vert  \bb{H}_n^{-1}(\bb{X}_t - \bb{x}) \Vert)) = \int K(u) \, d\phi_{\bb{x}}(h_n u)$, which we denote as $\psi_{\bb{x},1}(h_n)$ below. It is also easy to see that $\E\left( \Kcal^2( \bb{H}_n^{-1}(\bb{X}_t - \bb{x})) \right)  = \int K^2(u) \, d\phi_{\bb{x}}(h_n u)$, denoted as $\psi_{\bb{x}, 2}(h_n, 0).$ Due to the boundedness of the kernel function, the quantity $\psi_{\bb{x},1}(h_n)$ decays exactly at the order of $\phi_{\bb{x}}^j(h_n\lambda)$ for $j = 1, 2$. Therefore, we have,
\begin{equation}
    \dfrac{1}{\phi_{\bb{x}}(h_n\lambda)}\E\left[ \Kcal^j( \bb{H}_n^{-1}(\bb{X}_t - \bb{x})) \right] \rightarrow \xi_j,\label{eqn:kernel-convergence}
\end{equation}
as $n \rightarrow \infty$ for some $\xi_1 > 0, \xi_2 > 0$ (cf. Corollary 1 of~\cite{hong2020nonparametric}).

Moving on to the joint distributions, we consider the cross moment of kernel functions
\begin{align}
    &\E\left( \Kcal( \bb{H}_n^{-1}(\bb{X}_t - \bb{x})) \Kcal( \bb{H}_n^{-1}(\bb{X}_{t'} - \bb{x})) \right)
    =  \E\left( K\left(\norm{ \bb{H}_n^{-1}(\bb{X}_t - \bb{x}) }\right) K\left(\norm{ \bb{H}_n^{-1}(\bb{X}_{t'} - \bb{x}) }\right) \right)\nonumber\\
    & \hspace{0.5in} = \int K(u)K(v) d\phi_{\bb{x}, |t-t'|}(h_n u, h_n v) =  \psi_{\bb{x}, 2}(h_n, |t-t'|) \leqslant C_2^2 C \phi_{\bb{x}}^2(h_n \lambda) 
    \rightarrow 0, \label{eqn:kernel-cross-bound}
\end{align}
as $n \rightarrow \infty$ due to Assumption~\ref{assum:bandwidth-decay}. Here, $C > 0$ is some arbitrary constant free of $\bb{x}$. While the above derivation provides a uniform rate of decay for the cross-order moments of the kernel functions, we expect the decay rate to be $\xi_1^2\phi_{\bb{x}}^2(h_n\lambda) + g(\vert t - t'\vert)$ for some function $g(\cdot)$ that controls the covariance between the kernel function evaluated at $(\bb{X}_t - x)$ and $(\bb{X}_{t'} - \bb{x})$ as a function of the lag $\vert t-t'\vert$. In view of the PMC dependence structure as in~\eqref{assum:gmc-Xt}, it turns out that a quantitative upper bound to a Martingale difference of any measurable Lipschitz function of $\bb{X}_t$ is possible, as illustrated in the following result.

\begin{prop}\label{lem:f-Xk-bound}
    Let $f : \chi \rightarrow \R$ be a function such that it is measurable and Lipschitz (under the scaled norm given in \eqref{eqn:scaled-l2-space}) with a constant $C$, i.e., $\abs{f(\bb{x}) - f(\bb{y})} \leqslant C\Vert{\bb{D}^{-1}(\bb{x} - \bb{y})}\Vert$ for any $\bb{x}, \bb{y} \in \chi$. Then, with $\Fcal_{t}$ denoting the $\sigma$-field generated by $\{ \bb{X}_{t'}: t' \leqslant t \}$, 
    \begin{equation}
        \left\vert \E\left( f(\bb{X}_{t_i}) \mid \Fcal_{t_i-m} \right) - \E(f(\bb{X}_{t_i})) \right\vert \leqslant C \Delta_2(m)/\sqrt{2}.
        \label{eqn:f-Xk-bound}
    \end{equation}
\end{prop}

We note that, to make use of Proposition~\ref{lem:f-Xk-bound}, we need to apply Lipschitz-type smoothness restrictions on the kernel function. This is formalized below.

\begin{enumerate}[label = (A\arabic*), ref=(A\arabic*)]
    \setcounter{enumi}{3}
    \item\label{assum:kernel-mu-lipschitz} For each $\bb{x} \in \chi$, the function $f_{\bb{x}} : \chi \rightarrow \R$ given by $f_{\bb{x}}(\bb{y}) = \Kcal( \bb{H}_n^{-1}(\bb{x}-\bb{y}) )\mu_k(\bb{y})$     is Lipschitz over $\chi$ in the scaled $L^2$-metric given in \eqref{eqn:scaled-l2-space}, where $\mu_k$ is the component functions of $\mu$ as given in \eqref{eqn:mean-sigma-basis}. The Lipschitz constant is also uniform over $\bb{x} \in \chi$. The same holds for the kernel-weighted variance components, i.e., for $g_{\bb{x}}(\bb{y}) =  \Kcal( \bb{H}_n^{-1}(\bb{x}-\bb{y}) )\sigma_k(\bb{y})$ for any fixed $k \in \Z^+$ and uniformly over any $\bb{x} \in \chi$.
\end{enumerate}

This results in an improved bound on the autocovariance between the kernel functions, as illustrated in the following corollary, whose proof is provided in the supplementary material.

\begin{corr}\label{corr:kernel-covariance-bound}
    For sufficiently large $n$, for some constant $C > 0$ and $C_2$ as in~\eqref{eqn:kernel-bounds}, we have
    \begin{multline}
        \cov\left( \Kcal(\bb{H}_n^{-1}(\bb{X}_t - \bb{x})), \Kcal(\bb{H}_n^{-1}(\bb{X}_{t'} - \bb{x})) \right)
        \leqslant \dfrac{\sqrt{2} C C_2^2}{9h_n^2\lambda^2} \Delta_2(\vert t - t'\vert) (1 - \E\left( \phi_{\bb{X}_0}(3h_n\lambda) \right))\\ 
        + C_2^2 (1 + C) \phi_{\bb{X}}^2(h_n \lambda) \E\left( \phi_{\bb{X}_0}(3h_n\lambda) \right).
        \label{eqn:Kt-cov-bound}
    \end{multline}
\end{corr}
\noindent A related inequality can be found in inequality (34) of~\cite{xiao2011asymptotic}. We shall make use of \eqref{eqn:kernel-convergence} and \eqref{eqn:Kt-cov-bound} repeatedly while deriving the asymptotic properties of the nonparametric estimators of $\mu(\cdot, \cdot)$ and $\sigma(\cdot, \cdot)$.

\section{Asymptotic properties}\label{sec:theoretical}

\subsection{Consistency of the estimator of mean}\label{sec:mean-consistency}

The first step in establishing various asymptotic properties of the proposed estimator would be to ensure that it is consistent. In this subsection, we derive two key theorems related to the consistency of the mean estimator. However, before proceeding with them, we present a short result that provides correctness guarantees for the Modified Monte-Carlo procedure illustrated in Algorithm~\ref{alg:modified-montecarlo} and provides a quantitative bound on the error.

\begin{prop}\label{lem:Yk-convergence}
    Fix any $t_i \in \Tcal, k \in \Z^+$. Under assumptions~\ref{assum:mean-sigma-bound}, we have
    \begin{equation}
        \E\abs{\widehat{Y}^\ast_{t_i k} - Y^\ast_{t_i k}} = \Ocal\left( \max\left\{ (\epsilon^\ast_{t_i})^d, \sup_{\Vert s - s'\Vert \leqslant \epsilon^\ast_{t_i}} \vert b_k(s) -b_k(s')\vert,  \sup_{\Vert s - s'\Vert \leqslant \epsilon^\ast_{t_i}} \vert \widetilde{\mu}(s) - \widetilde{\mu}(s') \vert \right\}  \right),
        \label{eqn:Yk-convergence-errorbound}
    \end{equation}
    where $\epsilon_{t_i}^\ast$ and $\widehat{Y}^\ast_{t_i k}$ is the estimated spatially-aggregated response $Y^\ast_{t_i k}$ using the modified Monte-Carlo procedure (see \eqref{eqn:Yk-estimate}), and $\widetilde{\mu}(s)$ is as defined in Assumption~\ref{assum:mean-sigma-bound}. Let us denote the right-hand side of \eqref{eqn:Yk-convergence-errorbound} as $\delta_k(\epsilon_{t_i}^\ast)$.
\end{prop}

It is clear that, in addition to Assumption~\ref{assum:mean-sigma-bound}, if we consider $\epsilon_{t_i}^\ast \rightarrow 0$ (which is a proxy for the average gap of data-points in the spatial horizon at timepoint $t_i$) in the asymptotic regime, then Proposition~\ref{lem:Yk-convergence} establishes $L^1$-type convergence for the modified Monte-Carlo estimate $\widehat{Y}^\ast_{t_i k}$, from which the consistency of the Algorithm~\ref{alg:modified-montecarlo} readily follows.

Moving on, the first key result presented in this section, establishes the consistency of the estimate $\widehat{\mu}_k(\bb{x})$ of each component of the mean function to the corresponding component $\mu_k(\bb{x})$ as in \eqref{eqn:mean-sigma-basis}, pointwise at each $\bb{x} \in \chi$. 

\begin{thm}\label{thm:pointwise-component}
    Along with Assumptions~\ref{assum:mean-sigma-bound}-\ref{assum:kernel-mu-lipschitz}, assume that the effective spatial resolution $\epsilon_{t_i}^\ast$ as defined in Algorithm~\ref{alg:modified-montecarlo} decays to $0$ uniformly over $t_i \in \Tcal$. Then, the estimate of the mean components, i.e., $\widehat{\mu}_k(x)$ as in \eqref{eqn:mean-estimator} is pointwise consistent for $\mu_k(x)$ as $n \rightarrow \infty$. Mathematically, it means that for any fixed $k \in \Z^+$ and any $\bb{x} \in \chi$, $    \widehat{\mu}_k(\bb{x}) - \mu_k(\bb{x}) = o_{\prob}(1)$.
\end{thm}

While Theorem~\ref{thm:pointwise-component} establishes the pointwise consistency of $\widehat{\mu}_k(\bb{x})$ for the component functions $\mu_k(\bb{x})$, it has limited applicability in practice. Since the unknown mean function $\mu(\bb{x}, s)$ may consist of infinitely many component functions as illustrated in \eqref{eqn:mean-sigma-basis}, with a finite amount of computational power, we can estimate it by a truncated series consisting of a large, but finitely many terms. This motivates the following theorem.

\begin{thm}\label{thm:pointwise-full}
    Suppose that the assumptions of Theorem~\ref{thm:pointwise-component} hold. In addition, assume that there exists a function $b_\infty : \Scal \rightarrow (0,\infty)$ such that $\vert b_k(s)\vert \leqslant b_\infty(s)$ for all $s \in \Scal$. Then for any fixed $x \in \chi$, $s \in \Scal$ and $\epsilon > 0$, there exists sufficiently large $K_{\epsilon}$ such that
    \begin{equation*}
        \left(\sum_{k=1}^{K_{\epsilon}} \widehat{\mu}_k(\bb{x})b_k(s) - \mu(\bb{x}, s)\right) = \Ocal_{\mathbb{P}}(\epsilon),
    \end{equation*}
    as $n \rightarrow \infty$. Moreover, if $b_\infty(s)$ is continuous, then as the sample size $n \rightarrow \infty$, 
    \begin{equation*}
        \sup_{s \in \Scal}\left\vert \sum_{k=1}^{K_{\epsilon}} \widehat{\mu}_k(\bb{x})b_k(s) - \mu(\bb{x},s) \right\vert = \Ocal_{\mathbb{P}}(\epsilon).
    \end{equation*}
\end{thm}

\subsection{Asymptotic distribution of the estimate of mean}\label{sec:mean-asymp}

A natural next step is to establish a non-degenerate asymptotic limit of the estimator, after proper normalization. As before, fix $k \in \Z^+$. Then, for any $\bb{x} \in \chi$, we aim to understand the asymptotics for $v_{nk}(\bb{x})^{-1/2} (\widehat{\mu}_k(\bb{x}) - \mu_k(\bb{x}) - \tilde{b}_{nk}(\bb{x}) )$, where $\tilde{b}_{nk}(\bb{x})$ is a bias component and $v_{nk}(\bb{x})$ is the suitable normalization factor to obtain a non-degenerate limit, which turns out to be equal to $(n \phi_{\bb{x}}(h_n \lambda))^{-1}$. However, before describing the result in detail, we present an additional assumption that strengthens the bandwidth restrictions of Assumption~\ref{assum:bandwidth-decay} by connecting it with both the spatial and temporal dependency patterns.

\begin{enumerate}[label = (A\arabic*), ref=(A\arabic*)]
    \setcounter{enumi}{4}
    \item\label{assum:bandwidth-zero} Let $\tau$ be the exponent in PMC condition~\eqref{assum:gmc-Xt} and $\delta_k(\epsilon^\ast_{t_i})$ be as given in Proposition~\ref{lem:Yk-convergence}. Assume that there exists $0 < (\tau - 1)^{-1} < \beta < \alpha < 1$  such that the sequence of bandwidth parameter $h_n$ satisfies
    $n^{1 - 2\alpha + 2\beta}\phi_{\bb{x}}(h_n\lambda) \to 0$ and $\sqrt{n}\phi_{\bb{x}}^{-1/2}(h_n \lambda) ( \sup_{t_i} \delta_k(\epsilon^\ast_{t_i})) \to 0$.
\end{enumerate}

Assumption~\ref{assum:bandwidth-zero} can be thought of as an extension of Assumption~\ref{assum:bandwidth-decay} that determines the precise rate at which the small-ball probability needs to decay. For example, if $\chi = \R^d$ and $h_n = n^{-d/2}$, then the small-ball probability decays at the rate of $\Ocal(h_n^d) = \Ocal(n^{-1/2})$. If $\tau = 3$, one may choose $\beta = x + 1/2$ and $\alpha = x + 5/6$ for some $x \in (0, 1/6)$. In this case, $\phi_{\bb{x}}(h_n\lambda)n^{1-2(\alpha - \beta)} = \Ocal(n^{-1/6})$ which decays to $0$, satisfying the first part of Assumption~\ref{assum:bandwidth-zero}. As $n \to \infty$ and one has dense samples over the unit square, $\epsilon_{t_i}^\ast \to 0$, which along with continuity of mean and basis functions also imply $\delta_k(\epsilon^\ast_{t_i}) \to 0$. With a precise interplay between the spatial resolution of the sampling and the small-ball probability, the second part of Assumption~\ref{assum:bandwidth-zero} can be satisfied. The final result on the asymptotic normality of the suitably centered and scaled estimator is obtained below.

\begin{thm}\label{thm:asymptotic-normality}
    Let $\xi_1$ and $\xi_2$ be as defined in \eqref{eqn:kernel-convergence}. Under Assumptions~\ref{assum:mean-sigma-bound}-\ref{assum:bandwidth-zero}, as $n \rightarrow \infty$,
    \begin{equation}
        \dfrac{\sqrt{n\phi_{\bb{x}}(h_n\lambda)}\xi_1}{\sqrt{\xi_2 \sigma_{kk}(\bb{x})}} \left( \widehat{\mu}_k(\bb{x}) - \mu_k(\bb{x}) - \tilde{b}_{nk}(\bb{x}) \right),
        \label{eqn:normalized-estimate}
    \end{equation}
    is asymptotically distributed as a standard normal random variable, where 
    \begin{equation}
        \sigma_{kl}(\bb{x}) = \sum_{u=1}^\infty\sum_{v=1}^\infty c_{u,v}(k,l) \sigma_u(\bb{x})\sigma_v(\bb{x}), \ \text{ for any } k, l = 1, 2, \dots,
        \label{eqn:sigma-k-l-defn}
    \end{equation}
    \begin{equation}
        c_{u,v}(k,l) =  \int_{\Scal^2} \rho(s,s')b_k(s)b_u(s)b_l(s')b_v(s')ds ds'.
        \label{eqn:c-kl-defn}
    \end{equation}
    and $\sigma_{u}(\bb{x})$s are the component functions of $\sigma(\bb{x}, s)$ as in \eqref{eqn:mean-sigma-basis}. The bias $\tilde{b}_{nk}(\bb{x})$ is given by
    \begin{equation}
        \tilde{b}_{nk}(\bb{x}) = \dfrac{1}{n\E(K_0)} \sum_{i=1}^n \E(K_{t_i} \mu_k(\bb{X}_{t_i})) - \mu_k(\bb{x}).
        \label{eqn:bias-term}
    \end{equation}
\end{thm}

It is clear from Assumption~\ref{assum:kernel-mu-lipschitz} that $\tilde{b}_{nk}(\bb{x}) = O(h_n)$, where $h_n$ is the bandwidth. While $h_n \rightarrow 0$ asymptotically, justifying the pointwise consistency of the mean estimator, in many practical purposes, we may need a bias-correction procedure to ensure a lower-order bias term. For this, one may consider a jackknife-type procedure and obtain a revised estimator as
\begin{equation}
    \widehat{\mu}_{k}^{\ast, (h_n)}(\bb{x}) = 2\widehat{\mu}_{k}^{(h_n)}(\bb{x}) - \widehat{\mu}_{k}^{(2h_n)}(\bb{x}),
    \label{eqn:bias-corrected-estimate}
\end{equation}
where $\widehat{\mu}_{k}^{(h_n)}(\bb{x})$ is the estimator of the $k^{th}$ mean component with bandwidth $h_n$. It is now easy to see that the bias of the corrected estimate $\widehat{\mu}_{k}^{\ast, (h_n)}(\bb{x})$ is $O(h_n^2)$. Another point to note is that the quantity $\phi_{\bb{x}}(h_n\lambda)$ appearing in normalizing constant is usually unknown, but one may estimate the entire normalizing constant $n^{1/2}\phi^{1/2}_{\bb{x}}(h_n\lambda)$ by a plug-in estimate of the square root of the number of points lying in the support of the kernel function $K(\cdot, \bb{x})$ centered at $\bb{x} \in \chi$. Finally, if we were to assume a Karhunen-Lo\'{e}ve-type expansion for the $Y$-value at a location $s$, when covariate is $\bb{x}$, as
$
    Y(\bb{x}) = \sum_{k=1}^\infty a_k(\bb{x}) b_k(s),
$
then the quantity $\sigma_{kl}(\bb{x})$ can be viewed as a covariance operator acting on these random coefficients $a_k(\bb{x})$.

The proof of Theorem~\ref{thm:asymptotic-normality} follows by considering a decomposition of the normalized estimated given in \eqref{eqn:normalized-estimate} as in~\cite{masry2005nonparametric} and then follows a big-block small-block decomposition to establish the result. While \cite{masry2005nonparametric} makes use of the $\alpha$-mixing property to establish asymptotic independence between the big blocks to apply central limit theorems, we achieve this through applying the PMC condition on $\bb{X}_t$ as in Assumption~\ref{assum:gmc-Xt}. The technical details of the proof are outlined in the supplementary material. One may also establish a more general result that considers the estimation of multiple components of the mean function simultaneously. We present this as the following corollary.

\begin{corr}
    Under same assumptions as in Theorem~\ref{thm:asymptotic-normality}, for any finite set of distinct indices $k_1, k_2, \dots, k_m \in \Z^+$ and any fixed $\bb{x} \in \chi$, as $n \to \infty$,
    \begin{equation}
        \sqrt{n\phi_{\bb{x}}(h_n\lambda)}\begin{bmatrix}
            \widehat{\mu}_{k_1}(\bb{x}) - \mu_{k_1}(\bb{x}) - \tilde{b}_{n k_1}(\bb{x})\\
            \vdots \\
            \widehat{\mu}_{k_m}(\bb{x}) - \mu_{k_m}(\bb{x}) - \tilde{b}_{n k_m}(\bb{x})
        \end{bmatrix}
        \xrightarrow{d}
        N\left( \bb{0}, \dfrac{\xi_2}{\xi_1^2} \begin{bmatrix}
            \sigma_{k_1 k_1}(\bb{x}) & \dots & \sigma_{k_1 k_m}(\bb{x})\\
            \vdots & \vdots & \ddots & \vdots \\
            \sigma_{k_1 k_m}(\bb{x}) &  \dots & \sigma_{k_m k_m}(\bb{x})\\
        \end{bmatrix} \right).
        \label{eqn:asymp-normal-joint}
    \end{equation}
    \noindent Additionally, under the setting of Theorem~\ref{thm:pointwise-full}, for any $\epsilon > 0$ there exists a sufficiently large $K_\epsilon$ such that, for a standard normal random variable $Z$,
    \begin{equation}
        \sqrt{n\phi_{\bb{x}}(h_n\lambda)} \dfrac{ \sum_{k=1}^{K_\epsilon} \left( \widehat{\mu}_k(\bb{x}) - \tilde{b}_{nk}(\bb{x}) \right)b_k(s) - \mu(\bb{x}, s) }{\sqrt{\frac{\xi_2}{\xi_1^2}\sum_{k=1}^{K_\epsilon}\sum_{l=1}^{K_\epsilon} \sigma_{kl}(\bb{x})b_k(s)b_l(s)}} = Z + \Ocal_{\prob}(\epsilon).
        \label{eqn:mean-total-normal}
    \end{equation}
\end{corr}
\noindent The proof of this is analogous to the proof of Theorem~\ref{thm:asymptotic-normality} and hence is omitted for brevity.

\subsection{Simultaneous Confidence Interval Estimation}\label{sec:interval-est}

Extending on the point estimators for the mean component, we can obtain an interval estimator to quantify the uncertainty of the estimation. From a direct usage of the asymptotic normality present in \eqref{eqn:mean-total-normal}, one can obtain an asymptotic confidence interval for $\mu_k(\bb{x})$, provided a valid consistent estimator is available for $\sigma_{kk}(\bb{x})$. To this end, consider the covariance between $\eta_{t_1 k}$ and $\eta_{t_2 l}$ for any $k, l \in \Z^+$ as given in \eqref{eqn:model-reduced}. 
\begin{align*}
    \cov(\eta_{t_1 k}, \eta_{t_2 l})
    & = \int_{\Scal^2} \E\left[ \sigma(\bb{X}_{t_1}, s)\sigma(\bb{X}_{t_2}, s') \right]\rho(s, s')b_k(s) b_l(s')ds ds'\\
    & = \sum_{u, v} \E\left[ \sigma_u(\bb{X}_{t_1})\sigma_v(\bb{X}_{t_2}) \right] \int_{\Scal^2} \rho(s,s')b_k(s)b_u(s)b_l(s')b_v(s')ds ds'\\
    & = \sum_{u,v} c_{u,v}(k,l) \E\left[ \sigma_u(\bb{X}_0)\sigma_v(\bb{X}_{\vert t_1 - t_2\vert}) \right],
\end{align*}
where the last line follows from stationarity of $\bb{X}_t$ and $c_{u,v}(k,l)$ is as given in \eqref{eqn:c-kl-defn}. Here, we also apply Fubini's theorem to exchange the expectation and the integral operator, which can be easily validated due to the uniform boundedness of the basis functions and $\tilde{\sigma}(s)$, as given in Assumption~\ref{assum:mean-sigma-bound}. By choosing $t_1 = t_2$, we get $\cov(\eta_{t_1 k},\eta_{t_1 l}) = \sum_{u,v} c_{u,v}(k,l) \E\left[ \sigma_u(\bb{X}_0)\sigma_v(\bb{X}_0) \right]$, or more precisely, the conditional covariance will be
\begin{equation}
    \cov(\eta_{t_i k}, \eta_{t_i l} \mid \Fcal_{t_i}) = \sum_{u,v} c_{u,v}(k,l) \sigma_u(\bb{X}_{t_i})\sigma_v(\bb{X}_{t_i}).
    \label{eqn:model-reduced-variance-cond}
\end{equation}
Connecting the above expression, namely \eqref{eqn:model-reduced-variance-cond}, with its population counterpart $\sigma_{k,l}(\bb{x})$ as in \eqref{eqn:sigma-k-l-defn}, a natural nonparametric estimate of $\sigma_{kl}(\bb{x})$ is given by
\begin{equation}
    \widehat{\sigma}_{kl}(\bb{x}) = \dfrac{\sum_{i=1}^n K( \Vert{\bb{H}_n^{-1}(\bb{x} - \bb{X}_{t_i})}\Vert ) (\widehat{Y}^\ast_{t_i k} - \widehat{\mu}_k(\bb{x}))(\widehat{Y}^\ast_{t_i l} - \widehat{\mu}_l(\bb{x})) }{\sum_{i=1}^n K( \Vert{\bb{H}_n^{-1}(\bb{x} - \bb{X}_{t_i})}\Vert ) }.
    \label{eqn:mode-reduced-variance-estimate}
\end{equation}
As established in the following theorem, it provides a reasonably valid estimator for $\sigma_{kl}(\bb{x})$. 

\begin{thm}\label{thm:sigma-pointwise}
    Under the same set of assumptions as in Theorem~\ref{thm:pointwise-component}, for any fixed $k, l \in \Z^+$ and a fixed $\bb{x} \in \chi$, we have $\widehat{\sigma}_{kl}(\bb{x}) - \sigma_{kl}(\bb{x}) = o_{\mathbb{P}}(1)$ as $n \rightarrow \infty$.
\end{thm}

Guided by Theorems~\ref{thm:asymptotic-normality} and ~\ref{thm:sigma-pointwise}, for any fixed $k \in \Z^+$ and $\bb{x} \in \chi$, we can now construct a Wald-style $100(1-\alpha)\%$ asymptotic confidence interval for $\mu_k(\bb{x})$ as illustrated in the following Corollary.
\begin{corr}
    Under the assumptions of Theorems~\ref{thm:asymptotic-normality} and~\ref{thm:sigma-pointwise}, for any fixed $k \in \Z^+$ and $\bb{x} \in \chi$, 
    \begin{equation*}
        \lim_{n \to \infty}\prob\left( \left\vert \widehat{\mu}_k(\bb{x}) - \tilde{b}_{nk}(\bb{x}) - \mu_k(\bb{x}) \right\vert \leq   \dfrac{z_{1-\alpha/2} \sqrt{\xi_2 \widehat{\sigma}_{k,k}(\bb{x}) }}{\xi_1 \sqrt{n\phi_{\bb{x}}(h_n\lambda)}} \right) = (1-\alpha),
    \end{equation*}
    \noindent $z_{1-\alpha/2}$ is the $(1-\alpha/2)^{th}$ quantile of the standard normal distribution, and the estimates $\widehat{\mu}_k(\bb{x})$ and $\widehat{\sigma}_{k k}(\bb{x})$ are as given by \eqref{eqn:mean-estimator} and \eqref{eqn:mode-reduced-variance-estimate} respectively. The same conclusion holds if one considers a bias-corrected version $\widehat{\mu}_{k}^{\ast, (h_n)}(\bb{x})$ as in \eqref{eqn:bias-corrected-estimate}.
\end{corr}

However, to produce a confidence interval for the entire mean function $\mu(\bb{x}, s)$, not just the individual components, we need to additionally consider the existence of a dominating function $b_\infty(s)$ on the basis functions, as in Theorem~\ref{thm:pointwise-full}. Then, starting with any $\epsilon > 0$, we obtain a sufficiently large $K_\epsilon \in \Z^+$ as guaranteed by Theorem~\ref{thm:pointwise-full}, which by means of \eqref{eqn:mean-total-normal} produces a $100(1-\alpha)\%$ confidence interval for $\mu(\bb{x}, s)$ with the endpoints given by
\begin{equation*}
     \widehat{\mu}^\ast_{1:K_\epsilon}(\bb{x},s) \pm \left( \dfrac{z_{1-\alpha/2} \sqrt{\xi_2 Q_{K_\epsilon}(\bb{x},s) }}{\xi_1 \sqrt{n\phi_{\bb{x}}(h_n\lambda)}} +\epsilon b_\infty(s) \right),
\end{equation*}
where $b_\infty(\cdot)$ is as given in Theorem~\ref{thm:pointwise-full} and
\begin{align*}
    \widehat{\mu}_{1:K}^\ast(\bb{x},s) & =  \sum_{k=1}^{K}\widehat{\mu}_{k}^{\ast, (h_n)}(\bb{x})b_k(s), \ 
    \bb{b}_{1:K}(s) =  (b_1(s), \dots, b_{K})\tr, \\
    \widehat{\bb{\Sigma}}_{1:K, 1:K}(\bb{x}) & =  ((\widehat{\sigma}_{kl}(\bb{x}) ))_{k,l=1}^{K}, \ Q_K(\bb{x},s) = \bb{b}_{1:K}\tr(s)\widehat{\bb{\Sigma}}_{1:K, 1:K}(\bb{x}) \bb{b}_{1:K}(s).
\end{align*}

It is possible to extend the above pointwise confidence interval for each fixed $\bb{x} \in \chi$ to a simultaneous confidence interval over a dense subset $\chi_n$ of $\chi$ given by $\left\{ \bb{x} : \bb{x} \in \chi, \text{ and } B(\bb{x}, 2h_n\lambda) \cap \chi = \phi \right\}$, where $B(\bb{x}, 2h_n\lambda)$ denotes the infinite-dimensional ball centered at $\bb{x}$ with radius $2h_n\lambda$ with respect to the scaled distance metric. This means, for any $\bb{x} \neq \bb{x}' \in \chi_n$, the scaled norm $\Vert{\bb{D}^{-1}(\bb{x} - \bb{x}')}\Vert_2 > 2h_n\lambda$.  The following theorem establishes this uniform probabilistic bound over all $x \in \chi_n$ but for a fixed $s \in \Scal$. 

\begin{thm}\label{thm:simult-conf-1}
    Fix any $s \in \Scal$. Suppose that the Assumptions~\ref{assum:mean-sigma-bound}-\ref{assum:bandwidth-zero} hold. Let $\chi_n$ be a countable, bounded subset of $\chi$ such that for any $\bb{x}_0 \in \chi_n$, there is a local neighborhood $B(\bb{x}_0, \delta)$ for which the infinite series $\sum_{k=1}^\infty \mu_k(\bb{x})$ converges uniformly over all $\bb{x} \in B(\bb{x}_0, \delta)$ with $\delta$ independent of $n$. In addition, assume that there is a dominating function $b_\infty(s)$ such that $\vert b_k(s)\vert \leqslant b_\infty(s)$ for this $s \in \Scal$. Then for any $\epsilon > 0$, there exists $K_\epsilon \in \Z^+$ such that
    \begin{equation}
        \lim_{n \rightarrow \infty}\mathbb{P}\left( \sup_{\bb{x} \in \chi_n} \dfrac{\xi_1 \sqrt{n\phi_{\bb{x}}(h_n\lambda)}}{\sqrt{\xi_2 Q_{K_\epsilon }(\bb{x},s) }} \; \biggl\vert \abs{\widehat{\mu}^\ast_{1:K_\epsilon}(\bb{x},s) - \mu(\bb{x},s)} - \epsilon b_\infty(s) \biggr\vert < B_{m_n}(z) \right) \geqslant e^{-2e^{-z}}
        \label{eqn:conf-simult-1}
    \end{equation}
    for any fixed $z > 0$ and $s \in \Scal$. Here, 
    \begin{equation*}
        B_{m_n}(z) =  \sqrt{2\log(m_n)} - \dfrac{1}{\sqrt{2\log(m_n)}}\left[ \dfrac{1}{2}\log(\log(m_n)) + \log(2\sqrt{\pi}) \right] + \dfrac{z}{\sqrt{2\log(m_n)}}, \ m_n \geqslant 2,
    \end{equation*}
    and $m_n = \vert \chi_n\vert$, the number of elements in the set $\chi_n$.
\end{thm}

The size of the confidence set is controlled by $m_n$, which depends on the properties of the covariate domain $\chi$. For example, if $\chi$ is a subset of Euclidean vector space, then $m_n$ can be approximated by typical covering number metrics, while if $\chi$ is a smooth-function class over a field $\mathfrak{F}$, $m_n$ can be approximated by the covering number of $\mathfrak{F}$ multiplied by appropriate functions of VC-dimension or Dudley's entropy integral for the function class. We find it imperative to point out that such a Gumbel distribution-based simultaneous confidence interval arises in many situations: see~\cite{zhao2007inference, zhou2010simultaneous, ma2012simultaneous, deb2024nonparametric} for some useful applications.

\section{Real-life example: Analyzing air pollution in Delhi}\label{sec:realdata}

We consider a dataset on air pollution in India, extracted from the government website of the Central Pollution Control Board (CPCB: \url{https://cpcb.nic.in/}). The dataset consists of measurements of various pollutants (e.g. $\ppm$, PM10, $\text{NO}_2$, etc.), collected from stations irregularly distributed across different regions of India, with many not actively monitored at the moment. Keeping in view the continuously worsening pollution situation in Delhi \citep{dutta2022air}, we focus on the 38 stations in the national capital region for our analysis. Figure~\ref{fig:station-locations} shows the active status and locations of these measurement stations, most of which are concentrated around the union territory of Delhi and nearby regions.

\begin{figure}[!ht]
    \centering
    \includegraphics[width=0.48\linewidth]{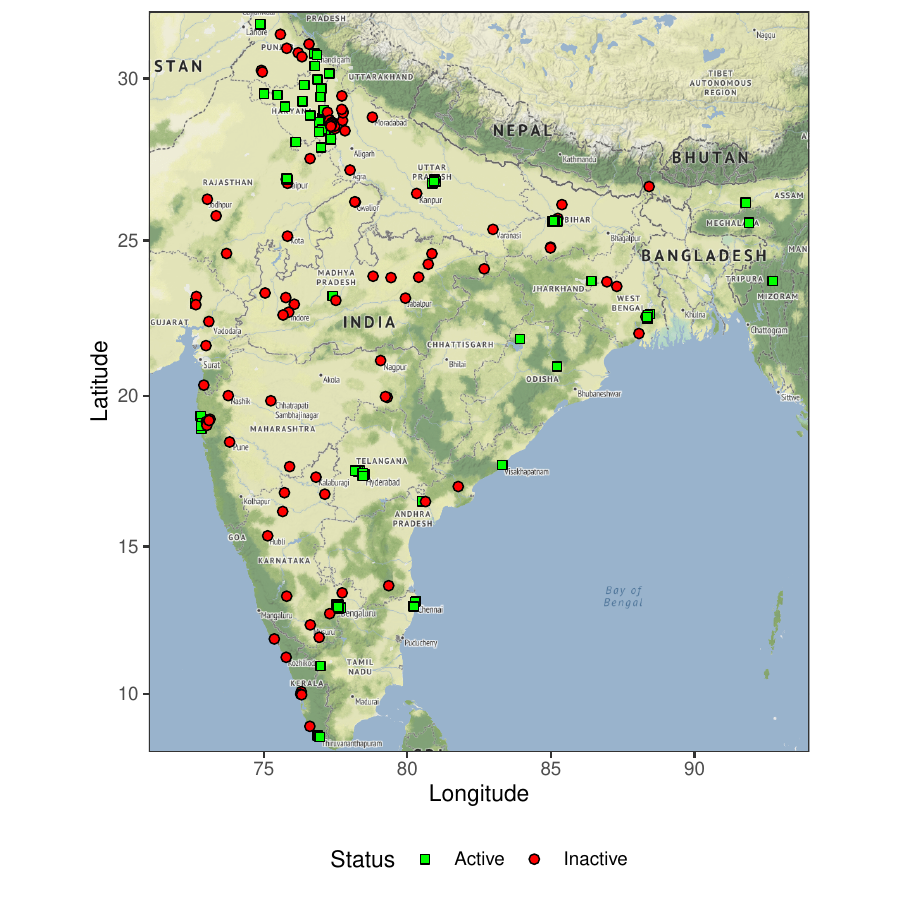}%
    \raisebox{0.9cm}{\includegraphics[width=0.42\linewidth]{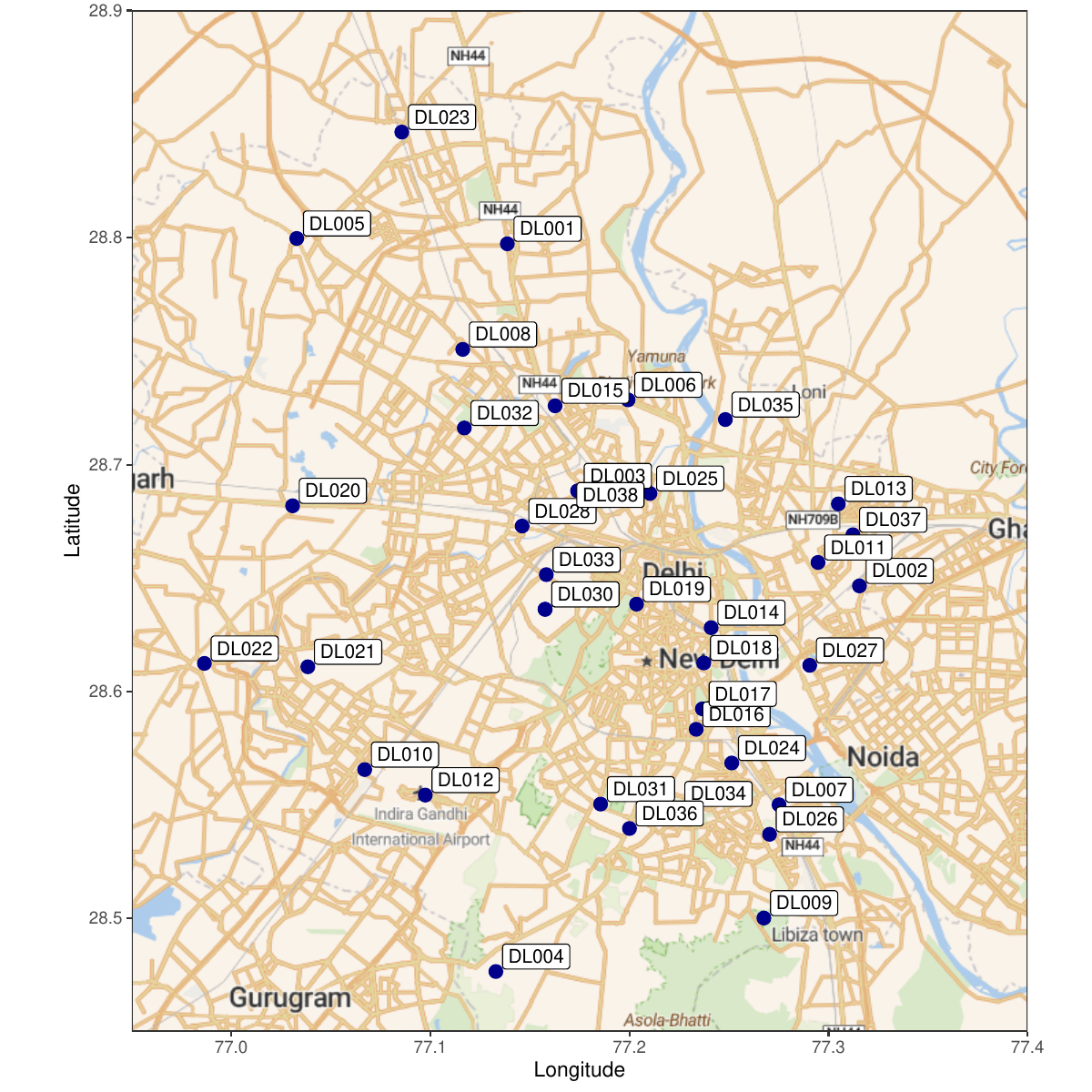}}%
    \caption{The locations of the measurement stations across India and their status (Left), and in the union territory of Delhi region (Right)}
        \label{fig:station-locations}
\end{figure}

The pollutant measurements are collected at an hourly rate between January, 2015 to June, 2020; though these ranges differ from one station to another. Empirically, all pollutant measurements display a positively skewed pattern. Thus, we apply the log-transformation to all of them for the main analysis, and use $\log(\ppm)$ as our main response variable $Y_t(s)$, for station $s$ and timepoint $t$. In general, we shall use $\bb{Z}_t(s)$ to denote the collection of all pollutant measurements (including $\ppm$) available at station $s$ and timepoint $t$. As an illustration, the temporal variation of the log-transformed $\ppm$ measurements are presented in Figure~\ref{fig:stations-pm25}, with a few selected stations highlighted. All $38$ stations demonstrate a consistent pattern, where the $\ppm$ values are usually higher in winter, show a decreasing trend till the months of late summer, and remain low during the monsoon months~\citep{roy2020pollution}.

\begin{figure}[htpb]
    \centering
    \includegraphics[width=0.7\linewidth]{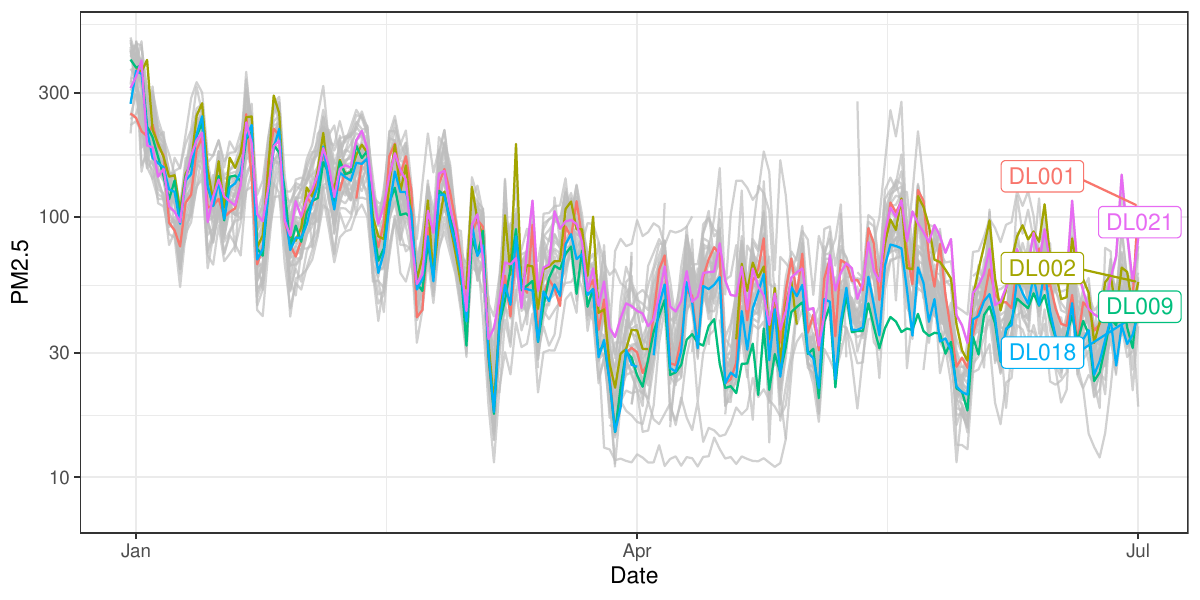}
    \caption{Hourly measurements of PM2.5 in different measurement stations in Delhi for the year 2020; some representative locations are highlighted by colour (y-axis is in log-scale).}
    \label{fig:stations-pm25}
\end{figure}

\subsection{Application of the proposed method}

To demonstrate the applicability of our estimator described in \Cref{sec:mean-estimation}, we model $Y_t(s)$ using $\{\bb{Z}_j(s) : s \in \Scal, \; j < t\}$ (this set of covariates is denoted as $\bb{X}_t$). As it includes measurements of different pollutants in vastly different scales, we also perform a normalization step for each. Although our theoretical analysis considers $\bb{X}_t$ to be infinite-dimensional, the number of historical observations is always finite in practice. Hence, we take these finite-dimensional covariates $\bb{X}_t' =  (X_{t1}, X_{t2}, \dots, X_{tp})$ and augment infinitely many zeros (i.e., $\bb{X}_t' =  (X_{t1}, \dots, X_{tp}, 0, 0, \dots)$) to make it infinite-dimensional and coherent with the theories derived above. As a distance metric, we choose a discounted $\mathcal{L}_2$-metric that discounts past observations with lag $l$ by $\phi^l$ for some $\phi \in (0, 1)$. We choose $\phi = 0.9$ by looking at the partial autocorrelation function of the log-transformed pollutant measurements; however, as discussed below, results remain fairly stable for any $\phi \geqslant 0.6$. Additionally, we use orthogonal Legendre polynomials on $[0, 1]^2$ as basis functions, and make a judicious choice of the number of basis functions via a leave-one-location-out cross-validation method. As illustrated in the introduction, the missingness of the observations brings challenges for standard spatiotemporal techniques with the subset of only complete cases or techniques involving regularly-spaced time series observations. Our approach circumvents these issues by modeling through irregularly-spaced time series data and using infinite-dimensional covariates spanning across multiple locations.

To study the efficacy of our proposed estimator, it may be useful to look at the two different plots given in Figures~\ref{fig:pollution-spatial-inference} and~\ref{fig:pm25_prediction}. The first one primarily illustrates the spatial aspect of the fitted model by depicting the estimated level of the response variable across the entire region. The map shows that the PM2.5 levels are higher in central Delhi, corresponding to the increasing population density of the city, while the levels are lower in the suburbs. Also, the spatial structure is non-isotropic: the PM2.5 levels are higher in the north-west direction compared to the south-east part of Delhi. We also depict this spatial variation across the first day of three different months and for two times of the day (Morning 10 AM, and evening 8 PM). The intensity of the difference in pollution levels between central Delhi and suburbs are most prominent during colder months and decays as the weather gets warmer. 

\begin{figure}[!ht]
    \centering
    \includegraphics[width=0.8\textwidth]{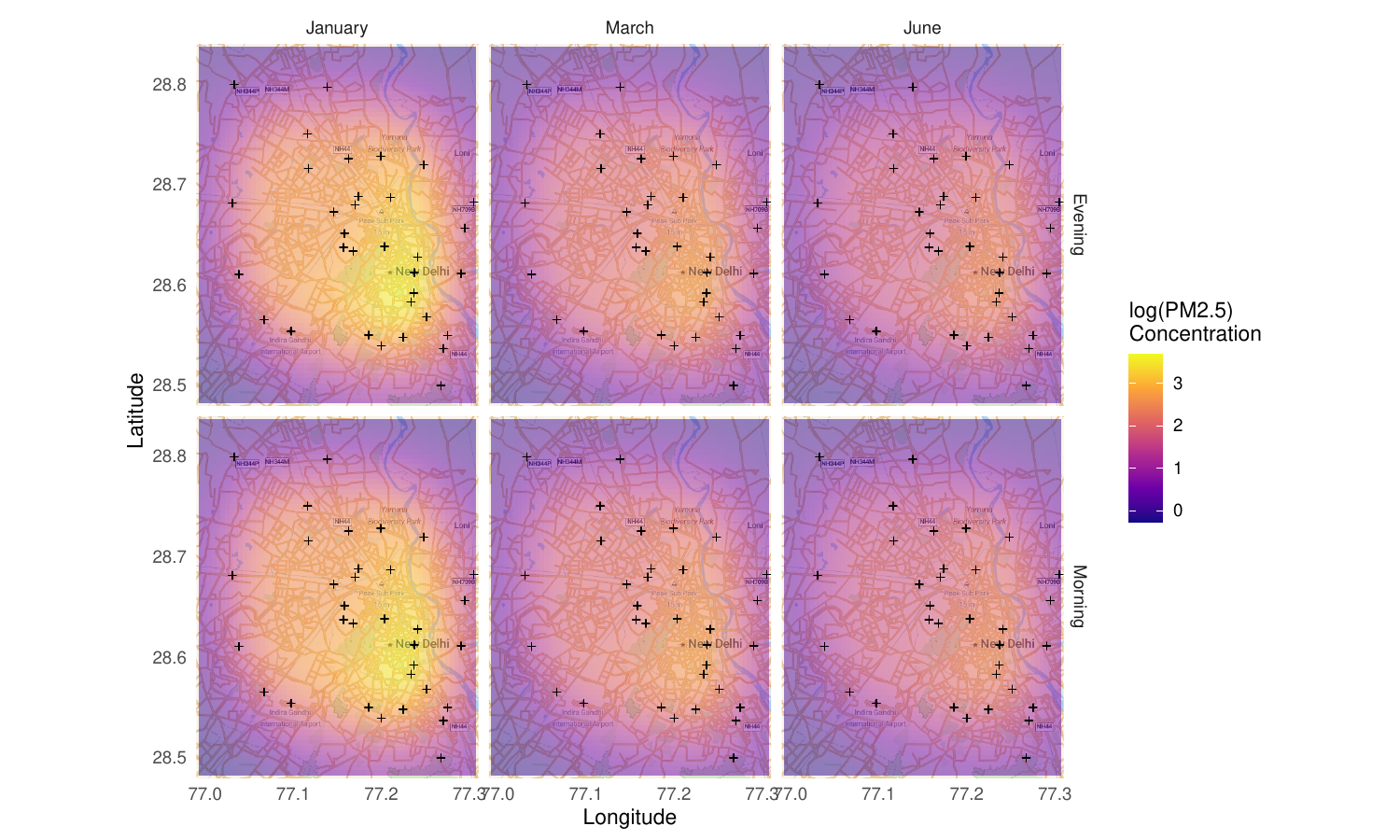}
    \caption{Estimated mean levels of PM$2.5$-concentration across the entire Delhi region for three different months and during two specific times of the day (Morning 10 AM and evening 8 PM). The black plus points indicate the locations of the measurement stations.}
    \label{fig:pollution-spatial-inference}
\end{figure}

Next, in Figure~\ref{fig:pm25_prediction}, we depict the predicted mean response for the month of June, 2020 across four randomly selected stations DL001, DL002, DL018 and DL021 using the pollutant measurements of their closest 3 locations (including self) as the covariates. For ease of visualization, we plot the response and estimates aggregated every 4 hours. This figure presents the temporal aspect captured by the fitted model. Although the predicted response matches the general pattern of the observed datapoints, the performance of the estimation depends on the specific location considered, and possibly each location requires a different number of neighboring locations' data as covariates, depending on several external factors such as the terrain heights, the density of vegetation, and the primary usage of the land area in those regions.

\begin{figure}[!ht]
    \centering
    \includegraphics[width=0.4\textwidth]{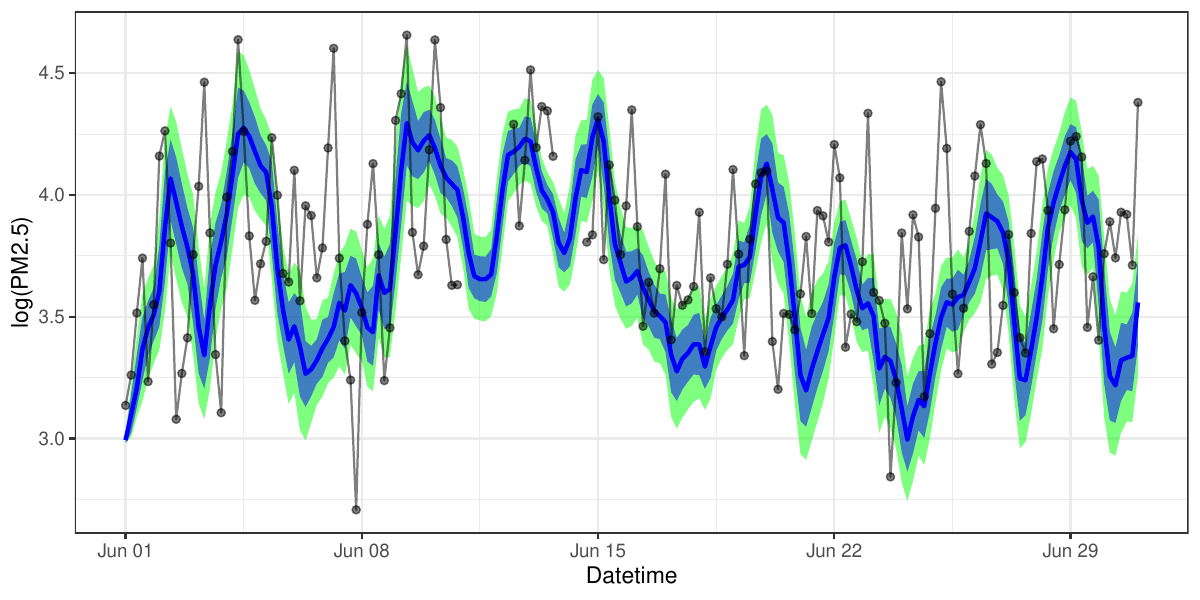}
    \includegraphics[width=0.4\textwidth]{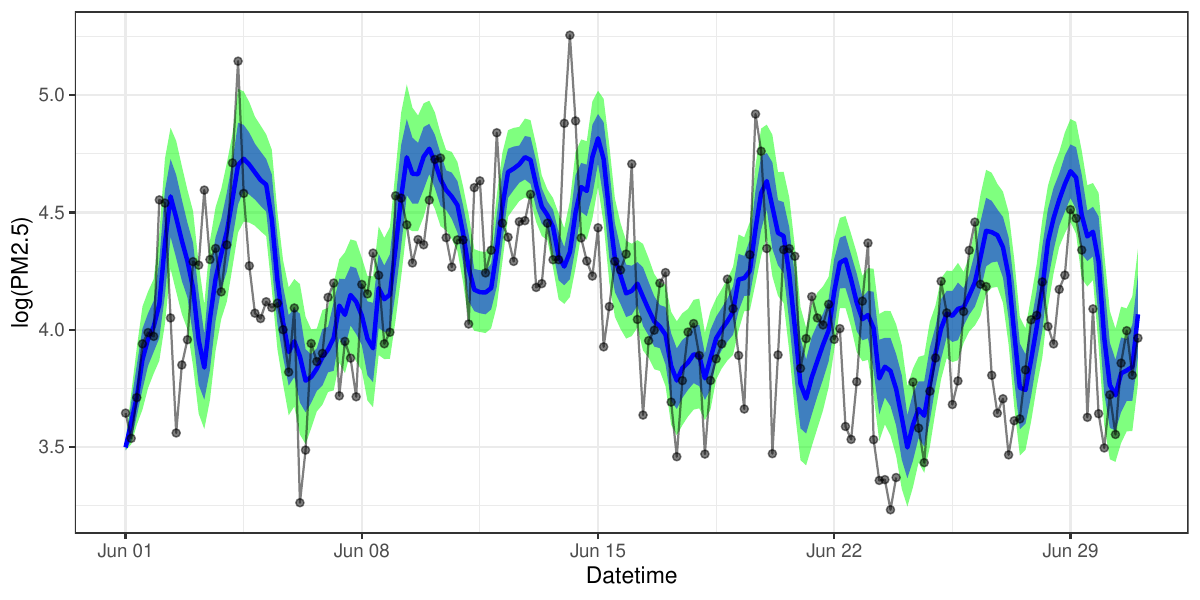}
    \includegraphics[width=0.4\textwidth]{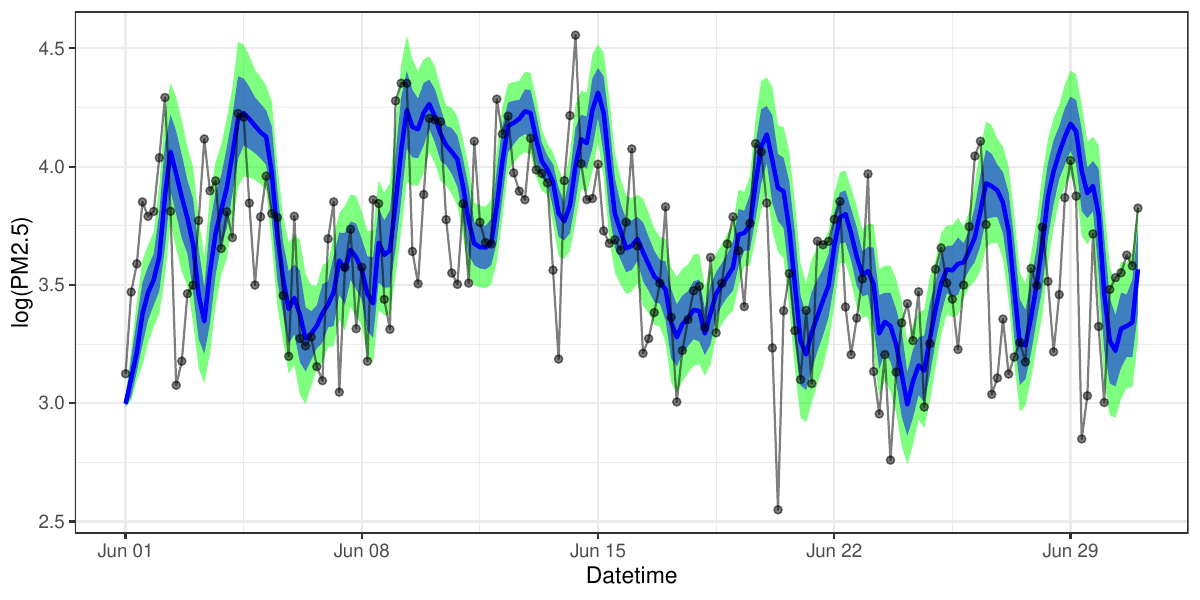}
    \includegraphics[width=0.4\textwidth]{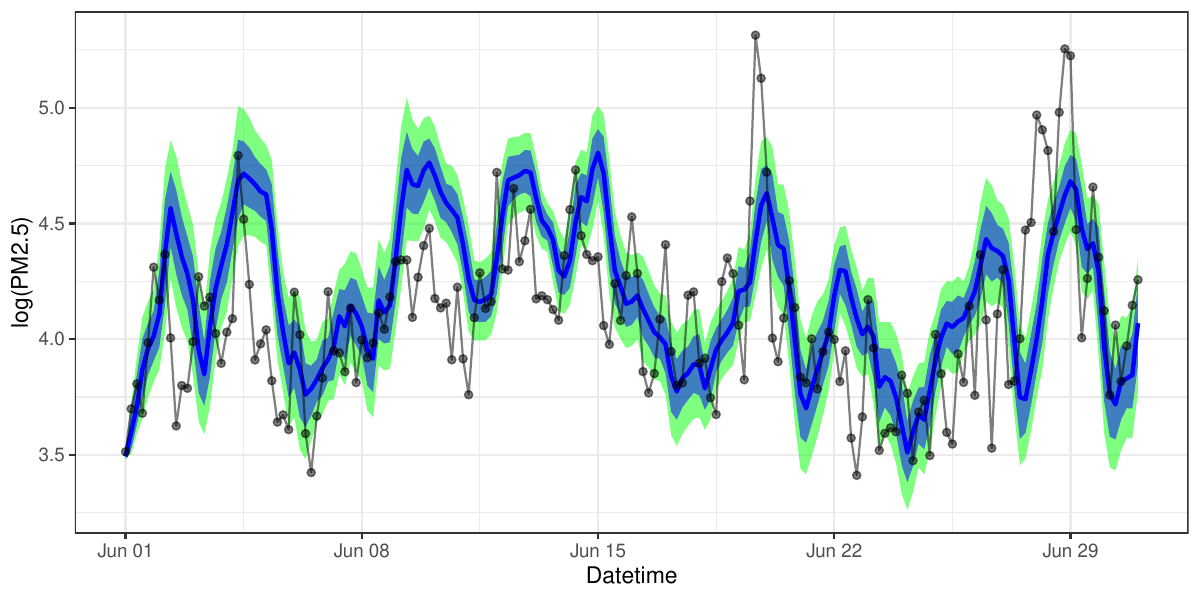}
    \caption{The predicted mean response (log of PM2.5 measurements) for four different measurement stations DL001 (top-left), DL002 (top-right), DL018 (bottom-left), DL021 (bottom-right). The blue band depicts the pointwise confidence interval, and the green band depicts the simultaneous confidence band across all timepoints in June 2020.}
    \label{fig:pm25_prediction}
\end{figure}

To investigate further, we perform estimation for each of the $38$ measurement stations for the month of June, 2020 by considering different numbers of covariates taken from its nearest $r$ locations, where $r$ is varied as $r = 1, \dots, 38$. For example, $r = 1$ indicates that only the past observations of the particular location is considered as covariate, and no spatial effect from other locations is considered in the modeling. With this experiment, for each location, we obtain the optimal number of nearest locations to be taken for covariates by minimizing the RMSE metric between the predicted values and the true log-transformed PM2.5 observations. In Figure~\ref{fig:stations-pm25-result}, we describe this optimal number of locations in the unit of kilometers around each location as a circle, while superimposing the data on the purpose of land usage.

\begin{figure}[!ht]
    \centering
    \includegraphics[width=0.5\textwidth]{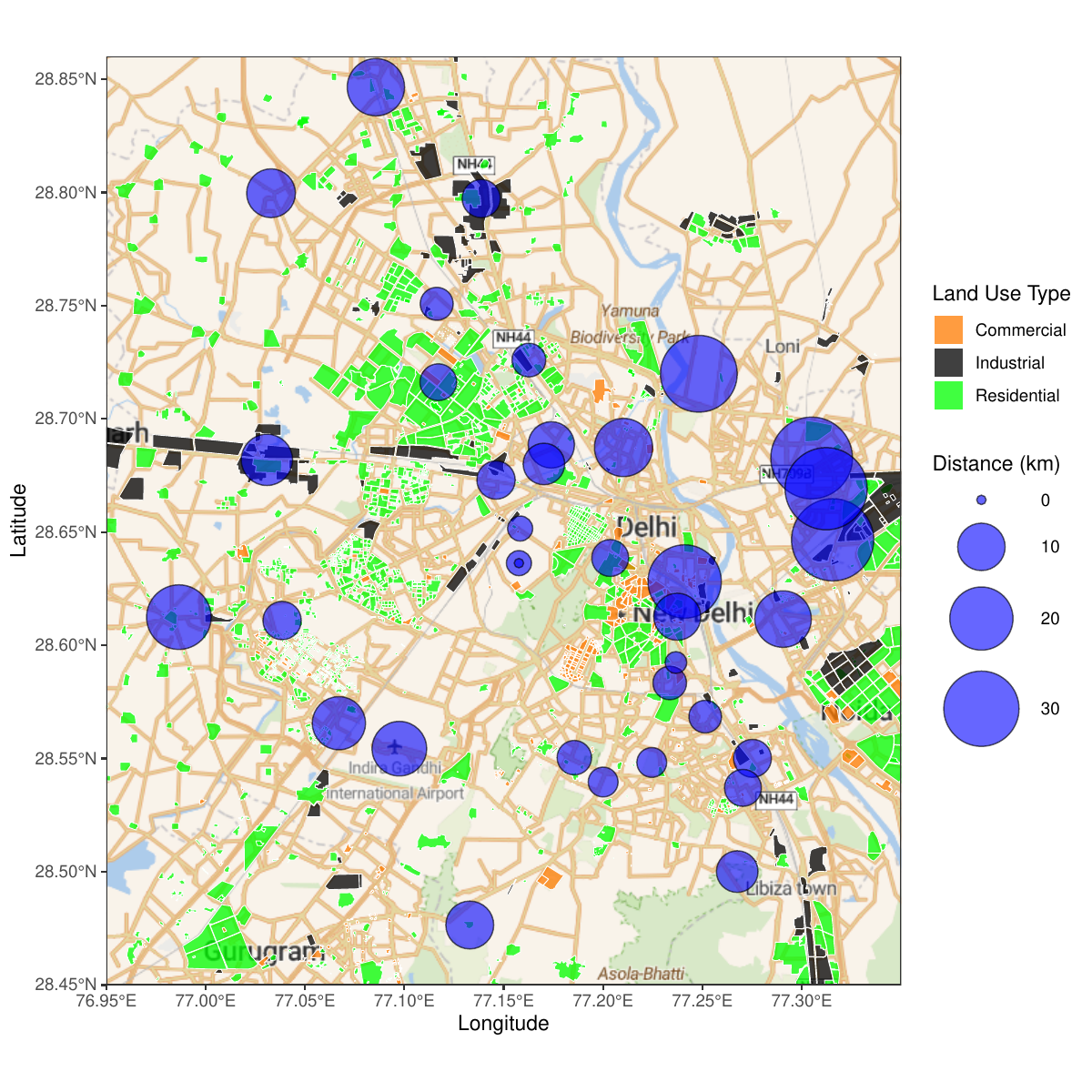}
    \caption{The best choice of the number of spatial covariates (in km) in predicting the PM2.5 measurements for June 2020 in every station at Delhi, along with landuse data.}
    \label{fig:stations-pm25-result}
\end{figure}

As a general pattern from Figure~\ref{fig:stations-pm25-result}, we see that the stations in the central Delhi require much less information for forecasting compared to the measurement stations in the outskirts. This is possibly because the locations of central Delhi are homogeneous in terms of land usage, residential and economic activities, and hence measurements from each additional nearby station in that region have less novel information to account for. On the other hand, the three stations on the east of the Yamuna river require information on many locations as its nearby industrial land usage and related activities are very different from the residential activities on the other side of the river. We believe that a more detailed analysis of the correlation between various factors and the forecasts from our model may provide useful insights about the optimal and strategic placements of these measurement stations in the future.

\subsection{Sensitivity analysis}

There are a few hyperparameters of the proposed nonparametric regression method. It is important to understand to what effect these hyperparameters affect the quality of the final prediction. To this end, we consider two different exercises. For the first one, we aim to analyze the degradation in prediction performance between in-training vs out-of-training locations. Let $s_0 \in \Scal$ be a fixed location. Then, we consider two separate fitted models $\mathcal{M}_{\text{in}}$ and $\mathcal{M}_{\text{out}}$ trained respectively on two separate datasets $\Dcal_{\text{in}}$ and $\Dcal_{\text{out}}$ given by
\begin{equation*}
    \mathcal{D}_{\text{out}}
    = \left\{ (Y_t(s), \bb{X}_t(s)): t \in \Tcal, s \in \Scal \setminus \{ s_0 \} \right\},\
    \mathcal{D}_{\text{in}}
    = \mathcal{D}_{\text{out}} \cup \left\{ (Y_t(s_0), \bb{X}_t(s_0)): t \in \Tcal \right\}.
\end{equation*}

The one-step ahead forecasts for the location $s_0$ from these two models are compared against each other using Diebold-Mariano (DM) test~\citep{diebold2015comparing}. The results of the DM test, along with the prediction error metric MAPE are summarized in Table~\ref{tab:in-vs-out-pred-3up}. It can be seen that for most of the locations ($30$ out of $35$), the null hypothesis (i.e., the two forecasts are same) is not rejected at the nominal level of $\alpha = 0.05$. Also, the error in prediction remains within the range of $10\%-15\%$ at most of the locations. This serves as an additional verification of the fact that our proposed nonparametric method performs equally well to forecast the levels of PM2.5 at an unknown location not present in the training data.

\begin{table}[htbp]
\centering
\caption{MAPE of the two forecast models where the prediction location is included (and excluded) in the training set, along with the results of DM test for their forecasts. Locations ``DL006'' and ``DL011'' are removed due to lack of data.}
\label{tab:in-vs-out-pred-3up}
\renewcommand{\arraystretch}{0.80}
\setlength{\tabcolsep}{2.2pt}
\scriptsize
\begin{tabular}{lccr | lccr | lccr}
\toprule
Location & In & Out & DM ($p$-value) &
Location & In & Out & DM ($p$-value) &
Location & In & Out & DM ($p$-value)\\
\midrule
DL001 & $12.720$ & $12.680$ & $0.235\,(0.407)$ &
DL015 & $11.040$ & $11.080$ & $-7.534\,(1.000)$ &
DL027 & $12.990$ & $13.070$ & $-3.474\,(1.000)$\\

DL002 & $10.370$ & $9.960$  & $4.425\,(0.000)$ &
DL016 & $12.450$ & $12.510$ & $-10.750\,(1.000)$ &
DL028 & $12.420$ & $12.450$ & $-7.601\,(1.000)$\\

DL003 & $12.110$ & $12.120$ & $-1.100\,(0.864)$ &
DL017 & $22.850$ & $22.840$ & $-4.135\,(1.000)$ &
DL029 & $12.000$ & $11.930$ & $1.657\,(0.049)$\\

DL004 & $26.670$ & $29.070$ & $-19.625\,(1.000)$ &
DL018 & $10.690$ & $10.720$ & $-5.654\,(1.000)$ &
DL030 & $20.560$ & $20.210$ & $11.728\,(0.000)$\\

DL005 & $14.610$ & $17.620$ & $-20.867\,(1.000)$ &
DL019 & $12.110$ & $12.060$ & $3.376\,(0.000)$ &
DL031 & $15.650$ & $15.740$ & $-7.595\,(1.000)$\\

DL007 & $13.310$ & $13.400$ & $-4.598\,(1.000)$ &
DL020 & $13.440$ & $13.650$ & $-11.650\,(1.000)$ &
DL032 & $11.130$ & $11.150$ & $-3.238\,(0.999)$\\

DL008 & $11.450$ & $11.440$ & $1.368\,(0.086)$ &
DL021 & $7.410$  & $7.550$  & $-14.008\,(1.000)$ &
DL033 & $15.410$ & $15.740$ & $-19.082\,(1.000)$\\

DL009 & $11.340$ & $12.150$ & $-14.883\,(1.000)$ &
DL022 & $18.070$ & $22.420$ & $-20.891\,(1.000)$ &
DL034 & $10.660$ & $10.660$ & $1.082\,(0.140)$\\

DL010 & $12.060$ & $12.160$ & $-7.580\,(1.000)$ &
DL023 & $23.940$ & $23.360$ & $6.566\,(0.000)$ &
DL035 & $11.020$ & $11.110$ & $-11.765\,(1.000)$\\

DL012 & $15.690$ & $15.740$ & $-3.727\,(1.000)$ &
DL024 & $11.850$ & $11.930$ & $-11.361\,(1.000)$ &
DL036 & $10.990$ & $11.050$ & $-9.657\,(1.000)$\\

DL013 & $16.190$ & $16.980$ & $-18.435\,(1.000)$ &
DL025 & $19.750$ & $19.780$ & $-0.937\,(0.825)$ &
DL037 & $12.500$ & $12.760$ & $-10.086\,(1.000)$\\

DL014 & $9.960$  & $10.070$ & $-10.466\,(1.000)$ &
DL026 & $10.720$ & $10.810$ & $-5.158\,(1.000)$ &
\multicolumn{4}{l}{}\\
\bottomrule
\end{tabular}
\end{table}

\section{Conclusion}\label{sec:conclusion}

The main contribution of this work is a nonparametric regression framework for irregularly sampled spatio-temporal observations with infinite-dimensional covariates. Our work establishes a full asymptotic theory, including simultaneous confidence bands, under a polynomially decaying moment contraction (PMC) condition rather than classical mixing assumptions. We start by establishing the consistency of our estimator of the mean part of the model, which also involves a truncated series approximation of the mean surface. Next, we provide asymptotic normality of the properly scaled and centered mean estimator. We finally end with consistency of variance estimators, which yields the confidence interval and simultaneous confidence band constructions. We address the spatial irregularity through a grid approximation, while the infinite dimensionality of the problem is addressed through small-ball probability assumptions. 

Some future directions one can consider are to strengthen the inferential theory from the mean to the covariance structure, by deriving asymptotic distributions and uniform confidence statements for $\sigma(x, s)$, complementing the consistency results already obtained for covariance estimation. Another possible direction is to relax the stationarity assumptions on $\{ \bb{X}_t\}$
 and develop an analogous theory under local stationarity or structural breaks, while preserving the PMC-based framework.
One could potentially extend the asymptotic results, which are currently obtained pointwise in $x$, and obtain functional central theorems, which could substantially broaden the scope of simultaneous inference. 

\section*{Data Availability Statement}\label{data-availability-statement}

The data used in the paper are sourced from the official website of the Central Pollution Control Board in India (CPCB: \url{https://cpcb.nic.in/}). The cleaned dataset and R implementation codes are available at the following link: \url{https://curated-webrepo.s3.ap-south-1.amazonaws.com/regression-code.zip}.

\begingroup
\setstretch{0.9}
\setlength{\bibsep}{0.2pt}
\bibliography{references.bib}
\endgroup

\renewcommand\thesection{S.\arabic{section}}
\renewcommand\thetable{S.\arabic{table}}
\renewcommand\thefigure{S.\arabic{figure}}
\renewcommand\theequation{S.\arabic{equation}}

\setcounter{table}{0}
\setcounter{figure}{0}
\setcounter{section}{0}
\setcounter{equation}{0}

\allowdisplaybreaks

\newpage
\begin{center}
    {\LARGE \textbf{Supplementary material}}
\end{center}

\section{Inference of the Variance Component}

In the main paper, we have tackled the problem of estimating the mean function $\mu(\bb{x}, s)$, and the component-wise covariance operator $\sigma_{kl}(\bb{x})$, but a detailed exposition on the estimation of the scale component $\sigma(\bb{x}, s)$ in the model \eqref{eqn:model} mentioned in the main manuscript was not included. This section establishes the estimation procedure and the asymptotic properties of the estimator for the variance component. Note that all cross-references in this document may pertain to the contents in the supplement (whenever it has the prefix S.) or the main manuscript (if there is no additional prefix).

\subsection{Estimation Algorithm}\label{sec:estimate-variance}

Akin to the case with the mean function, we again start with the reduced model established in equation \eqref{eqn:model-reduced} of the main paper and use the decomposition given in \eqref{eqn:mean-sigma-basis} to estimate the component functions $\sigma_u(\bb{x})$ for each $u \in \Z^+$. Note that, by definition of $\sigma_{kl}(\bb{x})$, we know that 
\begin{equation}
    \sigma_{kl}(\bb{x}) = \sum_{u=1}^\infty \sum_{v = 1}^\infty c_{u,v}(k,l) \sigma_{u}(\bb{x}) \sigma_v(\bb{x}), \ \text{ for any } k, l \in \Z^+.
    \label{supp-eqn:sigma-k-l-defn}
\end{equation}
If we substitute $\sigma_{kl}(\bb{x})$ by its estimate $\widehat{\sigma}_{k,l}(\bb{x})$ as in \eqref{eqn:mode-reduced-variance-estimate} for each $k, l \in \Z^+$, we have a countable set of equations involving the components $\sigma_u(\bb{x})$. If such a system is solvable, then it would let us recover an estimate of the variance component as $\widehat{\sigma}(\bb{x}, s) = \sum_{u=1}^{\infty} \widehat{\sigma}_u(\bb{x}) b_u(s)$ where $\widehat{\sigma}_u(\bb{x})$ is a solution of these equations.

One approach to solve this system is as follows: Following the decomposition of $\sigma(\bb{x},s)$ as in \eqref{eqn:mean-sigma-basis}, we must have $\sum_{u=1}^\infty \sigma_u(\bb{x})$ to be finitely summable, and hence for any given $\epsilon > 0$, there exists a $U_\epsilon \in \Z^+$ such that $\sum_{u=U_\epsilon+1}^\infty \sigma_u(\bb{x}) < \epsilon$ for each $\bb{x} \in \chi$. Additionally, as $\vert \rho(s,s')\vert \leqslant 1$, by an application of Cauchy-Schwartz inequality, it follows that $c_{u,v}(k,l) \leqslant 1$ for all $u,v,k,l \in \Z^+$. Therefore, we have
\begin{equation*}
    \sum_{\max(u,v) > U_\epsilon} \vert c_{u,v}(k,l) \sigma_u(x)\sigma_v(x) \vert 
    \leqslant \sum_{u=U_\epsilon + 1}^{\infty}\sum_{v=1}^\infty \sigma_u(x)\sigma_v(x) + \sum_{u=1}^{\infty}\sum_{v=U_\epsilon + 1}^\infty \sigma_u(x)\sigma_v(x) = \Ocal(\epsilon).
\end{equation*}

As a result, we can truncate the infinite sums present in the right-hand side of \eqref{supp-eqn:sigma-k-l-defn} with an error of at most $\Ocal(\epsilon)$. At this point, let us denote the $U_\epsilon$-length vector comprising of $\{ \sigma_u(\bb{x}) \}_{u=1}^{U_\epsilon}$ as $\bb{\sigma}_{1:U_\epsilon}(\bb{x})$. Let us denote the corresponding vector of the estimates as $\widehat{\bb{\sigma}}_{1:U_\epsilon}(\bb{x})$. Let us also denote $\bb{C}_{k,l}$ as a $U_\epsilon\times U_\epsilon$-matrix comprising of the elements $\{ c_{u,v}(k,l) \}_{u,v=1}^{U_\epsilon}$. Using these notations and the truncation argument, we can then rewrite the system of quadratic equations to be solved as a system of linear equations in the unknown coordinates of $\text{vec}(\bb{\sigma}_{1:U_\epsilon})$, which yields
\begin{equation}
    \widehat{\bb{\sigma}}_{1:U_\epsilon}(\bb{x})\widehat{\bb{\sigma}}_{1:U_\epsilon}(\bb{x})\tr = \text{vec}^{-1}\left( \left[ \text{vec}(\bb{C}_{11}) : \dots : \text{vec}(\bb{C}_{U_\epsilon, U_\epsilon}) \right]^{-1} \text{vec}(\widehat{\bb{\Sigma}}(\bb{x})) \right)=  \bb{A},
    \label{eqn:sigma-estimate}
\end{equation}
where $\widehat{\bb{\Sigma}}(\bb{x})$ is the $U_\epsilon \times U_\epsilon$-size matrix with entries $\sigma_{kl}(\bb{x})$ for $k, l = 1, 2, \dots, U_\epsilon$, and $\text{vec}(\cdot)$ denotes the vectorization operation of a matrix. However, the solution of the right-hand side of \eqref{eqn:sigma-estimate} may not result in a rank-one matrix, hence we obtain the estimate $\widehat{\bb{\sigma}}_{1:U_\epsilon}(\bb{x})$ by considering the partial eigendecomposition of the output matrix. In summary, we obtain the estimate of the variance component as expressed in Algorithm~\ref{alg:variance-estimation}.

\begin{algorithm}[!ht]
    \DontPrintSemicolon
    \KwInput{$t_i \in \Tcal$, $\{ s_{t_i j} \}_{j=1}^{n_i} \subseteq \Scal$, Response observations $\{ Y_{t_i}(s_{t_i j})\}_{j=1}^{n_i}$}
    \For{$k = 1, 2, \dots, U_\epsilon$}{
        Compute $\widehat{\sigma}_{k l}(\bb{x})$ using \eqref{eqn:mode-reduced-variance-estimate}\;
    }
    Calculate $\bb{A}$ as the RHS of \eqref{eqn:sigma-estimate}\;
    $\bb{A}^\ast \leftarrow \argmin_{\bb{B} \text{ is p.d.}} \Vert \bb{A} - \bb{B}\Vert_F^2$\;
    $\lambda_1(\bb{A}^\ast) \leftarrow \text{maximum eigenvalue of } \bb{A^\ast}$ and $e_1(\bb{A}^\ast)$ be corresponding eigenvector\;
    \For{$u = 1, 2, \dots, U_\epsilon$}{
        $\widehat{\sigma}_u(\bb{x}) \leftarrow \sqrt{\lambda_1(\bb{A}^\ast)}\vert e_{1u}(\bb{A}^\ast)\vert$\;
    }
    \caption{Algorithm for estimation of $\sigma(\bb{x},s)$.}
    \label{alg:variance-estimation}
\end{algorithm}

One caveat of the above approach is that to compute the matrix $\bb{C}_{kl}$ and its elements $c_{u,v}(k,l)$, one requires the knowledge of the spatial correlation function $\rho(s, s')$. If it is known, one can directly use the above algorithm. Often in different contexts, the functional form of $\rho(s,s')$ may be known, e.g., $\rho(s, s') = e^{-\rho\Vert s - s\Vert^2}$ for some unknown parameter $\rho > 0$. In this case, we can estimate $\rho$ by the choice that minimizes the discrepancy $\Vert A - \widehat{\bb{\sigma}}_{1:U_\epsilon}(\bb{x})\widehat{\bb{\sigma}}_{1:U_\epsilon}(\bb{x})\tr \Vert_2$ in the above estimation process. 

In a completely nonparametric setup, the form of $\rho(s, s')$ is unknown. In this case, one can employ an iterative scheme as follows. Given an estimate of $\widehat{\rho}^{(m)}(s, s')$ at $m^{th}$ step of the iteration, we can employ the above procedure to obtain an estimate of $\widehat{\sigma}^{(m)}(\bb{x}, s)$ for any $\bb{x} \in \chi$ and $s \in \Scal$. However, using \eqref{eqn:cov-y-rho}, we can obtain a revised estimate of the $\rho(\cdot, \cdot)$ function as
\begin{equation}
    \widehat{\rho}^{(m+1)}(s, s') = \dfrac{\sum_{i=1}^n K(\Vert{\bb{H}_n^{-1} (\bb{X}_{t_i} - \bb{x}) }\Vert) Y_{t_i}(s)Y_{t_i}(s') }{\sum_{i=1}^n K(\Vert{\bb{H}_n^{-1} (\bb{X}_{t_i} - \bb{x}) }\Vert) \widehat{\sigma}^{(m)}(\bb{x},s)\widehat{\sigma}^{(m)}(\bb{x},s') }.
    \label{eqn:rho-revised-estimate}
\end{equation}
The final estimates of $\rho(\cdot, \cdot)$ and $\sigma(\cdot, \cdot)$ are given by their corresponding converged values.

\subsection{Asymptotic Results}\label{sec:variance-asymptotics}

From Theorem~\ref{thm:sigma-pointwise}, we see that for any fixed $k, l \in \Z^+$, the quantity $\sigma_{k l}(\bb{x})$ can be estimated consistently. Let us denote $\bb{\Sigma}_{1:K,1:K}(\bb{x})$ as the matrix with entries $\sigma_{k l}(\bb{x})$ for $k, l \in \{1, 2, \dots, K\}$, and denote the estimated matrix as $\hat{\bb{\Sigma}}_{1:K,1:K}(\bb{x})$ with corresponding entries $\hat{\sigma}_{k l}(\bb{x})$. Although Theorem~\ref{thm:sigma-pointwise} shows that each element of the matrix $\bb{\Sigma}_{1:U_\epsilon, 1:U_\epsilon}(\bb{x}) - \hat{\bb{\Sigma}}_{1:U_\epsilon, 1:U_\epsilon}(\bb{x})$ is small, the errors can accumulate based on the choice of $\epsilon$ through a considerable increase in $U_\epsilon$. Let us denote a sequence of $U_\epsilon\times U_\epsilon$-matrices $\bb{C}_{k,l}$ for each $k, l \in \{ 1, 2, \dots, U_\epsilon\}$ comprising of the elements $c_{k,l}(u,v)$ as defined in \eqref{eqn:c-kl-defn}. Then, based on the algorithm illustrated in Section~\ref{sec:estimate-variance} and an application of Weyl's inequality, it is easy to see that the estimation error of the variance component is bounded by $o_{\prob}(\sqrt{2U_\epsilon} \sup_{k,l}\text{cond}(\bb{C}_{k,l}))$. Here, $\text{cond}(\bb{A})$ denotes the condition number of a matrix $\bb{A}$. Based on the above discussion, we now state a consistency theorem for the estimation of $\sigma(\bb{x}, s)$ without a formal proof. 

\begin{prop}
    Assume the same conditions as in Theorem~\ref{thm:sigma-pointwise}. Given a fixed $x \in \chi$, for any $\epsilon > 0$, find $U_\epsilon \in \Z^+$ such that $\sum_{u=U_\epsilon}^\infty \sigma_u(\bb{x}) < \epsilon$. Then, for any $u \in \{ 1, 2, \dots, U_\epsilon \}$, 
    \begin{equation*}
        \vert \hat{\sigma}_u(\bb{x}) -\sigma_u(\bb{x}) \vert = \Ocal(\epsilon) + o_{\prob}\left( \sqrt{2 U_\epsilon} \sup_{k,l} \text{cond}(\bb{C}_{k,l}) \right).
    \end{equation*}
\end{prop}

\section{Proofs of the Results}

\subsection{Proof of Proposition~\ref{lem:f-Xk-bound}}\label{proof-lem:f-Xk-bound}

Let $\bb{X}_{t_i}'$ be another independent copy of $\bb{X}_{t_i}$. Then, consider the chain of equality,
\begin{align*}
    \left( \E(f(\bb{X}_{t_i})) - \E(f(\bb{X}_{t_i}) \mid \Fcal_{t_i - m}) \right)^2
    & = \left( \E\left( f(\bb{X}_{t_i}) - \E(f(\bb{X}_{t_i})) \mid \Fcal_{t_i-m} \right) \right)^2\\
    & \leqslant \E\left( (f(\bb{X}_{t_i}) - \E(f(\bb{X}_{t_i})))^2 \mid \Fcal_{t_i-m} \right)\\
    & = \dfrac{1}{2} \E\left( (f(\bb{X}_{t_i}) - f(\bb{X}'_{t_i}))^2 \mid \Fcal_{t_i-m} \right)\\
    & \leqslant \dfrac{C^2}{2} \E\left( \Vert \bb{D}^{-1}(\bb{X}_{t_i} - \bb{X}'_{t_i})\Vert^2 \mid \Fcal_{t_i-m} \right), \text{as } f \text{ is Lipschitz},\\
    & = \dfrac{C^2}{2} \E\left( \Vert \bb{D}^{-1}(\bb{X}_{t_i} - \bb{X}'_{t_i})\Vert^2 \mid \dots, \xi_{t_i-m-1}, \xi_{t_i-m} \right)\\
    & \leqslant \dfrac{C^2}{2}\Delta_2^2(m),
\end{align*}
where the last line follows from the definition of $\Delta_2(m)$. Taking square root to both sides now yields \eqref{eqn:f-Xk-bound}. 

\subsection{Proof of Corollary~\ref{corr:kernel-covariance-bound}}\label{proof-corr:kernel-covariance-bound}

Let us fix any $\bb{x}_0 \in \chi$ and $\delta > 0$, and consider any two distinct points $\bb{x}, \bb{y} \in \chi$ such that $\Vert \bb{D}^{-1}(\bb{x}-\bb{y})\Vert > \delta$. Since $h_n \rightarrow 0$ as $n \rightarrow \infty$ as in Assumption~\ref{assum:bandwidth-decay}, for sufficiently large $n$ we must have $\delta > 2h_n \lambda$. As a result, by applying triangle inequality, we obtain that either $\Vert \bb{H}_n^{-1}(\bb{x} - \bb{x}_0)\Vert$ or $\Vert \bb{H}_n^{-1}(\bb{y} - \bb{x}_0)\Vert$ must be greater than $\lambda$, where $\bb{H}_n = h_n \bb{D}$. Since the kernel has a bounded support $[0, \lambda]$, either $K(\Vert \bb{H}_n^{-1}(\bb{x} - \bb{x}_0)\Vert) = 0$ or $K(\Vert \bb{H}_n^{-1}(\bb{y} - \bb{x}_0)\Vert) = 0$. Therefore, the following Lipschitz property holds for any $\bb{x}_0 \in \chi$:
\begin{equation*}
    \vert K(\Vert \bb{H}_n^{-1}(\bb{x} - \bb{x}_0)\Vert ) - K(\Vert \bb{H}_n^{-1}(\bb{y} - \bb{x}_0)\Vert ) \vert
    \leqslant C_2 \leqslant (C_2/\delta) \Vert \bb{D}^{-1}(\bb{x} - \bb{y})\Vert.
\end{equation*}
Now, we can decompose the covariance into two terms: When $\Vert\bb{D}^{-1}(\bb{X}_t - \bb{X}_{t'})\Vert > \delta$, we make use of the Lipschitz property along with Proposition~\ref{lem:f-Xk-bound}, and when $\Vert\bb{D}^{-1}(\bb{X}_t - \bb{X}_{t'})\Vert \leqslant \delta$, we make use of the bound given in \eqref{eqn:kernel-cross-bound}. Since $\E\left[ K(\Vert \bb{H}_n^{-1}(\bb{X}_t - \bb{x}) \Vert) \right]$ remains same for any $t$ due to the stationarity of $\bb{X}_t$, let us choose $c$ to denote this expectation. For any $t' < t$, this now yields
\begin{align*}
    & \cov\left( \Kcal(\bb{H}_n^{-1}(\bb{X}_t - \bb{x} )), \Kcal(\bb{H}_n^{-1}(\bb{X}_{t'} - \bb{x} )) \right)\\
    ={} & \E\left[ \left( K(\Vert \bb{H}_n^{-1}(\bb{X}_t - \bb{x}) \Vert) - c \right)\left(  K(\Vert \bb{H}_n^{-1}(\bb{X}_{t'} - \bb{x}) \Vert) - c\right) \right]\\
    ={} & \E\left[ \left( K(\Vert \bb{H}_n^{-1}(\bb{X}_t - \bb{x}) \Vert) - c \right)\left(  K(\Vert \bb{H}_n^{-1}(\bb{X}_{t'} - \bb{x}) \Vert) - c\right) \ind{\Vert\bb{D}^{-1}(\bb{X}_t - \bb{X}_{t'})\Vert > \delta} \right] + \\
    & \qquad \E\left[ \left( K(\Vert \bb{H}_n^{-1}(\bb{X}_t - \bb{x}) \Vert) - c \right)\left(  K(\Vert \bb{H}_n^{-1}(\bb{X}_{t'} - \bb{x}) \Vert) - c\right) \ind{\Vert\bb{D}^{-1}(\bb{X}_t - \bb{X}_{t'})\Vert \leqslant \delta} \right]\\
    \leqslant & \E\left[ \left(  K(\Vert \bb{H}_n^{-1}(\bb{X}_{t'} - \bb{x}) \Vert) - c\right) \E\left[ \left( K(\Vert \bb{H}_n^{-1}(\bb{X}_t - \bb{x}) \Vert) - c \right) \ind{\Vert\bb{D}^{-1}(\bb{X}_t - \bb{X}_{t'})\Vert > \delta} \mid \Fcal_{t'} \right] \right] + \\
    & \qquad C_2^2 C \phi_{\bb{x}}(h_n\lambda) \prob(\Vert\bb{D}^{-1}(\bb{X}_t - \bb{X}_{t'})\Vert \leqslant \delta), \ \text{by \eqref{eqn:kernel-cross-bound}}\\
    \leqslant & \dfrac{\sqrt{2} C_2^2 C}{\delta} \Delta_2(\vert t - t'\vert) \prob(\Vert \bb{D}^{-1}(\bb{X}_t - \bb{X}_{t'})\Vert > \delta) +  C_2^2 C \phi_{\bb{x}}(h_n\lambda) \prob(\Vert\bb{D}^{-1}(\bb{X}_t - \bb{X}_{t'})\Vert \leqslant \delta)
\end{align*}
The above inequality holds for any $\delta > 0$ and for all sufficiently large $n$. By rewriting $\prob\left( \Vert \bb{D}^{-1}(\bb{X}_t - \bb{X}_{t'})\Vert \leqslant \delta \right)$ as a conditional expectation and using the stationarity of $\bb{X}_t$, we obtain
\begin{equation*}
    \prob\left( \Vert \bb{D}^{-1}(\bb{X}_t - \bb{X}_{t'})\Vert \leqslant \delta \right)
    = \E\left( \prob\left( \Vert \bb{D}^{-1}(\bb{X}_t - \bb{X}_s)\Vert \leqslant \delta \right) \mid \bb{X}_t \right)
    = \E\left( \phi_{\bb{X}_t}(\delta) \right)
    = \E\left( \phi_{\bb{X}_0}(\delta) \right).
\end{equation*}
Since we simply require $\delta > 2h_n\lambda$ for sufficiently large $n$, we can choose $\delta = 3h_n\lambda$. Therefore,
\begin{multline*}
    \cov\left( \Kcal(\bb{H}_n^{-1}(\bb{X}_t - \bb{x})), \Kcal(\bb{H}_n^{-1}(\bb{X}_{t'} - \bb{x})) \right)
    \leqslant \dfrac{\sqrt{2}C C_2^2}{9h_n^2\lambda^2} \Delta_2(\vert t - t'\vert) (1 - \E\left( \phi_{\bb{X}_0}(3h_n\lambda) \right))\\ 
    + C_2^2 (1 + C) \phi_{\bb{x}}^2(h_n \lambda) \E\left( \phi_{\bb{X}_0}(3h_n\lambda) \right).
\end{multline*}
This is same as \eqref{eqn:Kt-cov-bound} that we wanted to show.

\subsection{Proof of Proposition~\ref{lem:Yk-convergence}}\label{proof-lem:Yk-convergence}

Let us fix $t_i \in \Tcal$ and $k \in \Z^+$. With $\Scal = [0,1]^d$, let us denote $H_l$ as the $l^{th}$ hypercube with diameter $\epsilon^\ast_{t_i}$, with its center at $s_l^c$, as described in Algorithm~\ref{alg:modified-montecarlo}. Also note that since $\Scal$ is compact and $b_k(s)$ is continuous, $\sup_{s\in\Scal} \vert b_k(s)\vert$ exists and is finite. Similarly, by Assumption~\ref{assum:mean-sigma-bound}, we know that the same conclusion holds for $\widetilde{\mu}(s) = \E(\mu(\bb{X}_0, s))$ and $\widetilde{\sigma}(s) = \E(\sigma(\bb{X}_0, s))$. Let us take $M$ to a uniform upper bound of all these three functions.

Now, the expected error due to the modified Monte-Carlo procedure is
\begin{align*}
    \E\left( \widehat{Y}^\ast_{t_i k} - Y^\ast_{t_i k}\right)
    & = \E\left( \dfrac{1}{r_i} \sum_{l=1}^{r_i} Y_{t_i}(s_{t_i j_l})b_k(s_{t_i j_l}) - \int_{[0,1]^d} Y_{t_i}(s)b_k(s)ds \right)\\
    & = \dfrac{1}{r_i}\sum_{l=1}^{r_i} \E\left( Y_{t_i}(s_{t_i j_l})b_k(s_{t_i j_l}) - \int_{H_l}Y_{t_i}(s)b_k(s)ds \right)\\
    & = \dfrac{1}{r_i}\sum_{l=1}^{r_i} \E\left( Y_{t_i}(s_{t_i j_l})b_k(s_{t_i j_l}) - Y_{t_i}(s_l^\ast)b_k(s_l^\ast) \right), 
\end{align*}
for some $s_l^\ast \in H_l$ by an application of Mean Value Theorem. Continuing, 
\begin{align*}
    \E\left( \widehat{Y}^\ast_{t_i k} - Y^\ast_{t_i k}\right)
    & = \dfrac{1}{r_i}\sum_{l=1}^{r_i} \E\left( \mu(\bb{X}_{t_i}, s_{t_i j_l})b_k(s_{t_i j_l}) - \mu(\bb{X}_{t_i}, s_l^\ast)b_k(s_l^\ast) \right)\\
    & = \dfrac{1}{r_i}\sum_{l=1}^{r_i} \E\left( \mu(\bb{X}_{0}, s_{t_i j_l})b_k(s_{t_i j_l}) - \mu(\bb{X}_{0}, s_l^\ast)b_k(s_l^\ast) \right), \text{by stationarity}\\
    & \leqslant M \sup_{\Vert s - s'\Vert \leqslant \epsilon^\ast_{t_i}} \vert b_k(s) - b_k(s')\vert + M \sup_{\Vert s - s'\Vert \leqslant \epsilon^\ast_{t_i}} \vert \E(\mu(\bb{X}_0, s) - \mu(\bb{X}_0, s')) \vert,
\end{align*}
which takes care of the bias term of the error. Here, the constant $M$ is the generic uniform bound on the basis function $b_k(s)$ and the expected mean function $\tilde{\mu}(s)$; see the discussions following Assumption~\ref{assum:bounded-errors}. Now, turning our attention to the variance, we obtain
\begin{align*}
    \var\left( \widehat{Y}^\ast_{t_i k} \right)
    & = \var\left( \dfrac{1}{r_i} \sum_{l=1}^{r_i} Y_{t_i}(s_{t_i j_l})b_k(s_{t_i j_l}) \right)\\
    & = \dfrac{1}{r_i^2} \left[ \sum_{l=1}^{r_i} b_k^2(s_{t_i j_l}) \var\left( Y_{t_i}(s_{t_i  j_l}) \right) +  \sum_{l_1,l_2} b_k(s_{t_i j_{l_1}})b_k(s_{t_i j_{l_2}}) \cov\left( Y_{t_i}(s_{t_i  j_{l_1}}), Y_{t_i}(s_{t_i j_{l_2}}) \right) \right]\\
    & \leqslant \dfrac{M}{r_i^2} \sum_{l=1}^{r_i} \E(\sigma^2(\bb{X}_{t_i}, s_{t_i j_l})) + \dfrac{M'}{r_i^2} \sum_{l_1 \neq l_2}^{r_i} \E(\sigma(\bb{X}_{t_i}, s_{t_i j_{l_1}}) \sigma(\bb{X}_{t_i}, s_{t_i j_{l_2}})) \rho(s_{t_i j_{l_1}}, s_{t_i j_{l_2}})\\
    & = \dfrac{M}{r_i} \E(\sigma^2(\bb{X}_0, s_{t_i j_l})) + \dfrac{M'}{r_i^2} \sum_{l_1 \neq l_2}^{r_i} \E(\sigma(\bb{X}_{0}, s_{t_i j_{l_1}}) \sigma(\bb{X}_{0}, s_{t_i j_{l_2}})) \rho(s_{t_i j_{l_1}}, s_{t_i j_{l_2}}), \\
    & \leqslant \dfrac{M''}{r_i} + \dfrac{M''}{r_i^2}\sum_{l_1 \neq l_2}^{r_i} \rho(s_{t_i j_{l_1}}, s_{t_i j_{l_2}})
\end{align*}
Here, $M, M'$, and $M''$ are some generic constants which serve as the uniform bound for the basis function $b_k(s)$ and the expected mean $\tilde{\mu}(s)$ and variance $\tilde{\sigma}^2(s)$ functions.

We have the first term $\Ocal( (\epsilon_{t_i}^\ast)^d)$. For the second quantity, let us look at the distances between two points $s_{t_i j_{l_1}}$ and $s_{t_i j_{l_2}}$ for each pair of choices. For a fixed $j_{l_1}$, there are $(2k+1)^d-(2k-1)^d$ hypercubes which are $k$ units distant from $s_{t_i j_{l_1}}$ in $L^1$ distance. Correspondingly, there are $(2k+1)^d-(2k-1)^d = \Ocal(k^{d-1})$ many choices of $l_2$ such that the Euclidean distance between $s_{t_i,j_{l_1}}$ and $s_{t_i,j_{l_2}}$ is at least $k\epsilon_{t_i}^\ast/\sqrt{d}$. It now follows that due to the spatial dependence structure given in~\eqref{eqn:cov-y-rho}, we obtain an upper bound of the second term
\begin{align*}
    \dfrac{M^4}{r_i^2}\sum_{l_1 \neq l_2}^{r_i} \rho(s_{t_i j_{l_1}}, s_{t_i j_{l_2}})
    & \leqslant \dfrac{M^4}{r_i^2} \sum_{l_1 = 1}^{r_i} \sum_{k} \left( (2k+1)^d-(2k-1)^d \right) (k\epsilon_{t_i}^\ast/\sqrt{d})^{-(d+\delta)}\\
    & = \dfrac{M^4}{r_i} \sum_{k} \left( (2k+1)^d-(2k-1)^d \right) (k\epsilon_{t_i}^\ast/\sqrt{d})^{-(d+\delta)}\\
    & \leqslant \dfrac{M^4}{r_i} \sum_{k=1}^{\infty} \Ocal(k^{-(1+\delta)}) = \Ocal(1/r_i),
\end{align*}
since $\sum_{k=1}^{\infty} \Ocal(k^{-(1+\delta)}) < \infty$ for $\delta > 0$. 

Combining all of the above and using the fact that $r_i (\epsilon_{t_i}^\ast/\sqrt{d})^d = 1$, we get
\begin{multline*}
    \E^2\vert \widehat{Y}_{t_i k}^\ast - Y^\ast_{t_i k}\vert 
    \leqslant \E\left[\left( \widehat{Y}_{t_i k}^\ast - Y^\ast_{t_i k} \right)^2\right] \\
    = \Ocal\left( \max\left\{ (\epsilon^\ast_{t_i})^d, \sup_{\Vert s - s'\Vert \leqslant \epsilon^\ast_{t_i}} \vert b_k(s) -b_k(s')\vert,  \sup_{\Vert s - s'\Vert \leqslant \epsilon^\ast_{t_i}} \vert \widetilde{\mu}(s) - \widetilde{\mu}(s')) \vert \right\}^2  \right).
\end{multline*}
\noindent Taking square root on both sides now completes the proof.

\subsection{Proof of Theorem~\ref{thm:pointwise-component}}

We start by fixing any $k \in \Z^+$ and any $\bb{x} \in \chi$. Denoting $K_t = K( \Vert{ \bb{H}_n^{-1}(\bb{X}_t - \bb{x})}\Vert )$, we rewrite the estimate of $k$-th component of the mean function as
\begin{equation}
    \widehat{\mu}_k(\bb{x})
   = \dfrac{\frac{1}{n}\sum_{i=1}^n K_{t_i} \widehat{Y}^\ast_{t_i k} / \E(K_0) }{ \frac{1}{n} \sum_{i=1}^n K_{t_i} / \E(K_0) }
   =  \dfrac{\widehat{\mu}_{k,2}(\bb{x})}{\widehat{\mu}_{k,1}(\bb{x})},
   \label{eqn:mu-k-decomp}
\end{equation}
where $\widehat{\mu}_{k,1}(\bb{x})$ and $\widehat{\mu}_{k,2}(\bb{x})$ are the quantities it is replacing. Then, the error in estimation can be expressed as 
\begin{align}
    \widehat{\mu}_k(\bb{x}) - \mu_k(\bb{x})
    & = \dfrac{\widehat{\mu}_{k,2}(\bb{x})}{\widehat{\mu}_{k,1}(\bb{x})} - \mu_k(\bb{x}) \nonumber \\
    & = \dfrac{\widehat{\mu}_{k,2}(\bb{x}) - \mu_k(x) \widehat{\mu}_{k,1}(\bb{x}) }{\widehat{\mu}_{k,1}(\bb{x})} \nonumber \\
    & = \dfrac{\E\widehat{\mu}_{k,2}(\bb{x}) -\mu_k(\bb{x}) }{\widehat{\mu}_{k,1}(\bb{x})} + \dfrac{ \widehat{\mu}_{k,2}(\bb{x}) - \E\widehat{\mu}_{k,2}(\bb{x}) }{\widehat{\mu}_{k,1}(\bb{x})} - \mu_k(\bb{x}) \dfrac{ \widehat{\mu}_{k,1}(\bb{x}) - \E\widehat{\mu}_{k,1}(\bb{x}) }{\widehat{\mu}_{k,1}(\bb{x})} \nonumber \\
    & = (A_1 + A_2 + A_3) / \widehat{\mu}_{k,1}(\bb{x}), \label{eqn:mu-conv-decomposition}
\end{align}
where we use the mean stationarity of $\bb{X}_t$ to obtain that $\E\widehat{\mu}_{k,1}(\bb{x}) = 1$. 

At this point, the proof follows by showing three things:
\begin{enumerate}
    \item $\E(\widehat{\mu}_{k,1}(\bb{x})) = 1$ and  $(\widehat{\mu}_{k,1}(\bb{x}) - 1) = o_{\prob}(1)$, hence the denominator converges to $1$ and the quantity $A_3$ converges to $0$ in probability. This is demonstrated in Lemma~\ref{lem:muk1-conv}.
    \item The numerator of the term $A_1$ converges to $0$ in probability, uniformly over all $\bb{x} \in \chi$. This is demonstrated in Lemma~\ref{lem:muk2-muk-conv}.
    \item The numerator of the term $A_2$ converges to $0$ in $L_1$ norm, hence in probability as well. This is illustrated in Lemma~\ref{lem:muk2-conv}.
\end{enumerate}
These, together with an application of Slutsky's theorem, will imply that the pointwise convergence  $(\widehat{\mu}_k(\bb{x}) - \mu_k(\bb{x})) \xrightarrow{P} 0$ holds for any fixed $k$ and $\bb{x} \in \chi$. More details about this specific argument can be found in~\cite{hong2020nonparametric}, where the authors use the exact same decomposition.

Now, we will proceed to verify each of these claims through a series of Lemmas.

\begin{lem}\label{lem:muk1-conv}
    Suppose the Assumptions~\ref{assum:mean-sigma-bound}-\ref{assum:gmc-Xt} hold. Then for any fixed $k \in \Z^+, \bb{x} \in \chi$, as $n \rightarrow \infty$,
    \begin{equation*}
        \widehat{\mu}_{k, 1}(x) - 1 = o_{\mathbb{P}}(1).
    \end{equation*}
\end{lem}
\begin{proof}
    Since $\bb{X}_t$ is second-order stationary, we have
    \begin{equation*}
        \E(\widehat{\mu}_{k,1}(\bb{x})) = \dfrac{1}{n \E(K_0)} \sum_{i=1}^n \E(K_{t_i}) = 1.
    \end{equation*}
    Considering the variance, we get
    \begin{align*}
        & \var(\widehat{\mu}_{k,1}(\bb{x}))\\
        = \quad & \dfrac{1}{n^2 \E^2(K_0)} \left[ \sum_{i=1}^n \text{var}(K_{t_i}^2) + \sum_{\vert i - j\vert > 0} \cov(K_{t_i}, K_{t_j}) \right]\\
        = \quad &  \dfrac{1}{n\E^2(K_0)} (\E(K_0^2) - \E^2(K_0)) + \dfrac{1}{n^2\E^2(K_0)}\sum_{\vert i - j\vert > 0} \cov(K_{t_i}, K_{t_j})  \\
        \leqslant \quad & \dfrac{(1+o(1))\xi_2 }{(1-o(1))n\phi_{\bb{x}}(h_n\lambda)\xi_1^2 } + \dfrac{1}{(1-o(1)) \xi_1^2} \sum_{\vert i - j\vert > 0} \left[ \dfrac{\sqrt{2}CC_2^2 \Delta_2(\vert t_i - t_j\vert) }{9n^2 h_n^2\lambda^2 \phi_{\bb{x}}^2(h_n\lambda)} p_n + \dfrac{(1+C)C_2^2}{n^2} \bar{p_n} \right]
    \end{align*}
    where the last line follows from \eqref{eqn:Kt-cov-bound}, where $p_n = 1 - \E\left( \phi_{\bb{X}_0}(3h_n\lambda) \right)$ and $\bar{p_n} = 1 - p_n$. Now, for the first term, we apply the fact that $n\phi_{\bb{x}}(h_n\lambda) \rightarrow \infty$ as $n \rightarrow \infty$ as shown in \eqref{eqn:bandwidth-density-inf}. For the second term, we note that due to the PMC condition~\ref{assum:gmc-Xt}, 
    we have $\sum_{l = 1}^{n-1} \Delta_2(l) = \Ocal(n^{-\tau+2}) < \infty$, and hence the numerator is $\Ocal(n)$. Therefore, the quantity goes to $0$ as $nh_n^2\phi_{\bb{x}}^2(h_n\lambda) \rightarrow \infty$ by the choice of the bandwidth given in Assumption~\ref{assum:bandwidth-decay}. For the third and the final term, as shown in \eqref{eqn:Kt-cov-bound}, it is of $\Ocal(\E(\phi_{\bb{X}_0}(3h_n\lambda)))$. Since, $\phi_{\bb{x}}(h_n u) \rightarrow 0$ for any fixed $\bb{x} \in \chi$ and any $u > 0$, the third term also becomes asymptotically negligible. Therefore, we have
    \begin{equation*}
        \var(\widehat{\mu}_{k,1}(x)) \rightarrow 0,
    \end{equation*}
    as $n \rightarrow \infty$. An application of Chebyshev's inequality now completes the proof.
\end{proof}

\begin{lem}\label{lem:muk2-muk-conv}
    Suppose that the Assumptions~\ref{assum:mean-sigma-bound}-\ref{assum:kernel-mu-lipschitz} hold. Additionally, the effective spatial resolution $\epsilon_{t_i}^\ast$ decays to $0$ uniformly over $t_i \in \Tcal$. Then, for any fixed $k\in\Z^+$, we have $\sup_{\bb{x} \in \chi} \left\vert \E\widehat{\mu}_{k,2}(\bb{x}) -\mu_k(\bb{x}) \right\vert = o_{\prob}(1)$ as $n \rightarrow \infty$.    
\end{lem}

\begin{proof}
    Fix any $k\in\Z^+$. By using second-order stationarity of $\bb{X}_t$, we decompose the given quantity as
    \begin{align*}
    & \sup_{\bb{x} \in \chi} \left\vert \E\widehat{\mu}_{k,2}(\bb{x}) - \mu_k(\bb{x}) \right\vert\\ 
    = {} & \sup_{\bb{x} \in \chi} \left\vert \dfrac{1}{n} \sum_{i =1}^n \dfrac{\E K_{t_i}\widehat{Y}^\ast_{tk} }{\E K_0} - \mu_k(\bb{x}) \right\vert\\
    = {} & \sup_{\bb{x} \in \chi} \left\vert \dfrac{1}{n}\sum_{i=1}^n \dfrac{\E K_{t_i} (\widehat{Y}^\ast_{t_i k} - {Y}^\ast_{t_i k}) ) }{\E K_0} \right\vert  + \sup_{\bb{x} \in \chi} \left\vert \dfrac{1}{n} \sum_{i =1}^n \dfrac{\E K_{t_i} {Y}^\ast_{t_i k} }{\E K_0} - \dfrac{1}{n} \sum_{i=1}^n \dfrac{\E K_{t_i} \mu_k(\bb{x})}{\E K_0} \right\vert \\
    \leqslant {} & \sup_{\bb{x} \in \chi} \dfrac{1}{n} \sum_{i=1}^n \dfrac{\left\vert \E K_{t_i} (\widehat{Y}^\ast_{t_i k} - {Y}^\ast_{t_i k})\right\vert }{\left\vert \E K_0 \right\vert} 
    + \sup_{\bb{x} \in \chi} \dfrac{1}{n}\sum_{i=1}^n \left\vert \dfrac{ \E( K_{t_i}\E( \mu_k(\bb{X}_{t_i}) + \eta_{t_i k} - \mu_k(\bb{x})  \mid \mathcal{F}_{t_i} )) }{\E K_0} \right\vert\\
    \leqslant {} & \dfrac{C_2}{C_1}\sup_{t_i} \E\left\vert \widehat{Y}^\ast_{t_i k} - {Y}^\ast_{t_i k} \right\vert + \dfrac{C_2}{C_1}
    \sup_{x \in \chi} \sup_{y \in B(\bb{x}, h_n\lambda)} \left\vert \mu_k(\bb{y}) - \mu_k(\bb{x}) \right\vert,
    \end{align*}
    where the last line follows from the fact that $\E(\eta_{t_i k}\mid \Fcal_{t_i}) = 0$, and the kernel function has bounded support on $[0,\lambda]$ (cf.~\eqref{eqn:kernel-bounds}). Here, $B(\bb{x}, h_n\lambda)$ denotes the infinite-dimensional hyper-ellipsoid centered at $x$ with axis lengths $2h_n\eta_j\lambda$ in the $j^{th}$ direction. Note that, for the second term, the supremum can also be taken over the choice of any $\bb{x}$ and $\bb{y} \in \chi$ such that $\Vert \bb{H}_n^{-1}(\bb{x} - \bb{y}) \Vert \leqslant \lambda$.

    The first term of the above bound is $\Ocal(\delta_k(\epsilon_{t_i}^\ast))$ where $\delta_k(\epsilon_{t_i}^\ast)$ is as given in Proposition~\ref{lem:Yk-convergence}. This is asymptotically negligible since $\sup_{t_i} \epsilon_{t_i}^\ast \rightarrow 0$ and that the basis functions $b_k(s)$ is continuous over the spatial horizon $\Scal$.

    For the second term, let us define the notation $K_{\bb{x}}(\bb{y}) =  \Kcal(\bb{H}_n^{-1}(\bb{x}-\bb{y}))$. Then using the triangle inequality and the symmetry of the kernel function, we get that
    \begin{align*}
        & \vert \mu_k(\bb{x}) - \mu_k(\bb{y})\vert\\
        ={} & \dfrac{1}{K(0)+K_{\bb{x}}(\bb{y})} \vert (K(0) + K_{\bb{x}}(\bb{y}))(\mu_k(\bb{x})  - \mu_k(\bb{y}))\vert \\
        \leqslant{} & \dfrac{1}{K(0) + K_{\bb{x}}(\bb{y})} \left( \vert K_{\bb{x}}(\bb{x})\mu_k(\bb{x}) - K_{\bb{x}}(\bb{y})\mu_k(\bb{y}) \vert + \vert K_{\bb{y}}(\bb{x})\mu_k(\bb{x}) - K_{\bb{y}}(\bb{y})\mu_k(\bb{y}) \vert \right)\\
        \leqslant{} & \dfrac{1}{K(0) + C_1} \Ocal(\Vert \bb{D}^{-1}(\bb{x} - \bb{y}) \Vert),
    \end{align*}
    where in the last line, we make use of the Lipschitz property given in Assumption~\ref{assum:kernel-mu-lipschitz} and the boundedness of the type-I kernel. Since we are considering the supremum when $\Vert \bb{H}_n^{-1}(\bb{x}-\bb{y})\Vert < \lambda$, it follows that the above quantity is $\Ocal(h_n)$, which is asymptotically negligible as $n \rightarrow 0$, due to the choice of bandwidth as in Assumption~\ref{assum:bandwidth-decay}.    
\end{proof}

\begin{lem}\label{lem:muk2-conv}
    Suppose that the Assumptions~\ref{assum:mean-sigma-bound}-\ref{assum:kernel-mu-lipschitz} hold. Additionally, the effective spatial resolution $\epsilon_{t_i}^\ast$ decays to $0$ uniformly over $t_i \in \Tcal$. Fix any $k \in \Z^+$ and $\bb{x} \in \chi$. Then as $n \rightarrow \infty$,
    \begin{equation*}
        \E\vert \widehat{\mu}_{k,2}(\bb{x}) - \E(\widehat{\mu}_{k,2}(\bb{x})) \vert \rightarrow 0.
    \end{equation*}
\end{lem}
\begin{proof}
    We start by noting that, 
    \begin{equation}
        \left\vert \widehat{\mu}_{k,2}(\bb{x}) - \E(\widehat{\mu}_{k,2}(\bb{x})) \right\vert
        \leqslant \dfrac{1}{\E K_0 } \left\vert \dfrac{1}{n}\sum_{i=1}^n K_{t_i}\widehat{Y}_{t_i k}^\ast - \E(K_{t_i}\widehat{Y}_{t_i k}^\ast) \right\vert = \dfrac{1}{\E K_0} \left\vert \dfrac{1}{n}\sum_{i=1}^n U_i \right\vert.
        \label{eqn:proof-lem-muk2-conv-1}
    \end{equation}
    where
    \begin{equation*}
        U_i = K_{t_i}\widehat{Y}_{t_i k}^\ast - \E\left( K_{t_i}\widehat{Y}_{t_i k}^\ast \right).
    \end{equation*}
    Note that, $\E(U_i) = 0$ for each $i = 1, 2, \dots, n$. Although $\{ U_i \}_{i=1}^n$ are dependent, if we can show a result similar to the law of large numbers for the average of these $U_i$s, then the result will follow. We will use Theorem 1 by~\cite{andrews1988lln} to show this. In view of this, it suffices to show that there exists an appropriate sequence of real numbers $\{ \psi(m) \}_{m=1}^{\infty}$ and a large absolute constant $C_0$ such that
    \begin{enumerate}
        \item $\E\vert \E(U_i \mid \Fcal_{t_i-m}) \vert < C_0 \psi(m)$ with $\psi(m) \rightarrow 0$ as $m \rightarrow \infty$.
        \item $\E\vert U_i - \E(U_i \mid \Fcal_{t_i+m}) \vert < C_0 \psi(m+1)$ for $m \geq 0$.
        \item $U_i$s are uniformly integrable.
    \end{enumerate}
    Beginning with the first condition, we have
    \begin{align*}
        \vert \E(U_i \mid \Fcal_{t_i-m}) \vert
        & = \left\vert \E\left( K_{t_i}\widehat{Y}_{t_i k}^\ast - \E(K_{t_i}\widehat{Y}_{t_i k}^\ast) \mid \Fcal_{t_i-m} \right) \right\vert\\
        & = o_{\prob}(1) + \left\vert  \E\left( K_{t_i}Y_{t_i k}^\ast - \E(K_{t_i}Y_{t_i k}^\ast) \mid \Fcal_{i-m} \right) \right\vert, \text{as } \sup_{t_i} \E \vert \widehat{Y}_{t_i k}^\ast - Y_{t_i k}^\ast\vert = o_{\mathbb{P}}(1)\\
        & \leqslant o_{\prob}(1) + \left\vert \E\left( K_{t_i}\mu_k(\bb{X}_{t_i}) - \E(K_{t_i}\mu_k(\bb{X}_{t_i})) \mid \Fcal_{i-m} \right) \right\vert, \text{ as } \E(\eta_{t_i,k}\mid \Fcal_{t_i-m}) = 0,\\
        & \leqslant o_{\prob}(1) + C m^{-\tau},
    \end{align*}
    where the last line follows from the fact that the function $f(\bb{y}) =  \Kcal(\bb{H}_n^{-1}(\bb{y} - \bb{x}))\mu_k(\bb{y})$ is Lipschitz due to Assumption~\ref{assum:kernel-mu-lipschitz}, and a direct application of Proposition~\ref{lem:f-Xk-bound} in addition of the PMC condition~\eqref{assum:gmc-Xt}. Since $\tau > 0$, $\E\vert\E(Z_i \mid \Fcal_{t_i-m})\vert < c\psi(m)$ holds with $\psi(m) = m^{-\tau}$ and $C_0 = 2C$. 

    For the second condition, note that $\bb{X}_{t_i}$ is a causal process with respect to the filtration $\{ \Fcal_{t_i} \}_{i=1}^\infty$. Since $U_i$ is a linear combination of only $\{ Y_{t_i}(s_{t_i j}) \}_{j=1}^{n_i}$ with the weights being a function of $\bb{X}_{t_i}$, it is known almost surely given the history $\Fcal_{t_i + m}$ for any $m \geq 0$. Hence, $\E(U_i \mid \Fcal_{t_i+m}) = U_i$, and the second condition holds.

    For the final condition regarding the uniform integrability of $U_i$, we consider the decomposition
    \begin{align}
        \sup_i \E\vert U_i \vert
        & \leqslant \sup_{t_i} 2 \E\vert K(t_i) \widehat{Y}^\ast_{t_i k} \vert, \nonumber \\
        & = o_{\prob}(1) + \sup_{t_i} 2C_2 \E\left\vert \mu_k(\bb{X}_{t_i}) + \int_{\Scal} \sigma(\bb{X}_{t_i}, s)\epsilon_{t_i}(s)b_k(s) ds \right\vert \label{eqn:uniform-ui-bound-1} \\
        & \leqslant o_{\prob}(1) + 2C_2 \sup_{t_i} \left( \int_{\Scal} \E(\vert \mu(\bb{X}_{t_i}, s)\vert)b_k(s)ds + \int_{\Scal} \E\left(\sigma(\bb{X}_{t_i}, s) \E(\vert \epsilon_{t_i}(s)\vert) \right) b_k(s) ds \right) \nonumber \\
        & = o_{\prob}(1) + 2C_2 \left( \int_{\Scal} \E(\vert \mu(\bb{X}_0, s)\vert)b_k(s)ds + M \int_{\Scal} \E\left(\sigma(\bb{X}_{0}, s) \right) b_k(s) ds \right), \label{eqn:uniform-ui-bound-2}\\
        & \leqslant o_{\prob}(1) + 2C_2 \left( M' \int_{\Scal} \widetilde{\mu}(s) ds + M \left( \int_{\Scal} \widetilde{\sigma}^2(s)ds \right)^{1/2} \left( \int_{\Scal} b_k^2(s)ds \right)^{1/2} \right) \label{eqn:uniform-ui-bound-3}
    \end{align}
    where $M$ and $M'$ are some generic constants. The equation~\eqref{eqn:uniform-ui-bound-1} follows from the boundedness of the type-I kernel and Proposition~\ref{lem:Yk-convergence}, equation~\eqref{eqn:uniform-ui-bound-2} is a consequence of the stationarity of $\bb{X}_t$ and boundedness of expected error given in Assumption~\ref{assum:bounded-errors}. The last inequality~\eqref{eqn:uniform-ui-bound-3} is an application of the Cauchy-Schwarz inequality. Here, we use the notation $\widetilde{\mu}(s) = \E(\mu(\bb{X}_0, s))$ and $\widetilde{\sigma}^2(s) = \E(\sigma^2(\bb{X}_0, s))$ as in Assumption~\ref{assum:mean-sigma-bound}. Because of Assumption~\ref{assum:mean-sigma-bound}, both the corresponding integrals are finite, and $\int_{\Scal} b_k^2(s)ds = 1$ due to the orthonormality of the basis function. Hence $\{ U_i: 1 \leqslant i \leqslant n \}$ is uniformly integrable.

    Now an application of Theorem 1 by~\cite{andrews1988lln} implies that $\vert n^{-1}\sum_{i=1}^n U_i \vert$ converges to $0$ in $L_1$. We take expectation on both sides of the inequality \eqref{eqn:proof-lem-muk2-conv-1}.
\end{proof}

\subsection{Proof of Theorem~\ref{thm:pointwise-full}}\label{proof-thm:pointwise-full}

Fix $\bb{x} \in \chi$ and any $\epsilon > 0$. As $\mu(\bb{x}, s)$ has a decomposition as in \eqref{eqn:mean-sigma-basis}, we have $\sum_{k=1}^\infty \vert \mu_k(\bb{x})\vert < \infty$, where $\mu(\bb{x}, s) = \sum_{k=1}^\infty \mu_k(\bb{x}) b_k(s)$. 
    
\noindent\textbf{Consistency:} For the first part of the proof, let us fix $s \in \Scal$ as well. It is, therefore, possible to find $K_\epsilon$ such that the tail series $\sum_{k=K_{\epsilon}+1}^\infty \vert \mu_k(\bb{x})\vert < \epsilon/2b_\infty(s)$, as $b_\infty(s) > 0$ by definition. Hence,
\begin{equation*}
    \left\vert \sum_{k=K_{\epsilon}+1}^\infty \mu_k(\bb{x})b_k(s) \right\vert
    \leqslant \sum_{k=K_{\epsilon}+1}^\infty \vert \mu_k(\bb{x})\vert \left\vert b_\infty(s) \right\vert < \dfrac{\epsilon}{2}.
\end{equation*}
For the first $K_\epsilon$ terms, due to Theorem~\ref{thm:pointwise-component}, it is possible to choose sufficiently large $n$ such that for each $k = 1, 2, \dots, K_\epsilon$, the difference satisfies $(\widehat{\mu}_{k}(\bb{x}) - \mu_k(\bb{x})) = o_{\prob}(\epsilon/2K_\epsilon b_\infty(s))$. Therefore,
\begin{align*}
    \left\vert \sum_{k=1}^{K_{\epsilon}} \widehat{\mu}_k(\bb{x})b_k(s) - \mu(\bb{x}, s) \right\vert
    & = \left\vert \sum_{k=1}^{K_{\epsilon}} \widehat{\mu}_k(\bb{x})b_k(s) - \sum_{k=1}^\infty \mu_k(\bb{x})b_k(s) \right\vert\\
    & \leqslant \sum_{k=1}^{K_{\epsilon}} \left\vert \widehat{\mu}_k(\bb{x}) - \mu_k(\bb{x}) \right\vert \vert b_\infty(s)\vert + \sum_{k=1}^{K_\epsilon + 1} \vert \mu_k(\bb{x}) \vert \vert b_\infty(s)\vert\\
    & \leqslant \sum_{k=1}^{K_{\epsilon}} o_{\prob}\left( \dfrac{\epsilon}{2K_\epsilon b_\infty(s)} \right) \vert b_\infty(s)\vert + \dfrac{\epsilon}{2b_\infty(s)} \vert b_\infty(s)\vert = o_{\prob}(\epsilon).
\end{align*}

\textbf{Uniform Consistency:} Now we move over to the second part of the proof. As $b_\infty(s)$ is continuous on the compact set $\Scal$, there exists a constant $B_\infty \in (0, \infty)$ such that the supremum $\sup_{s\in\Scal} \vert b_\infty(s)\vert < B_\infty$. As a result, we have the decomposition
\begin{equation*}
    \left\vert \sum_{k=1}^{K_{\epsilon}} \widehat{\mu}_k(\bb{x})b_k(s) - \mu(\bb{x}, s) \right\vert
    \leqslant B_\infty \sum_{k=1}^{K_\epsilon} \left\vert \widehat{\mu}_k(\bb{x}) - \mu_k(\bb{x}) \right\vert + B_\infty \sum_{k=K_\epsilon+1}^\infty \vert \mu_k(\bb{x})\vert. 
\end{equation*}
Since the right-hand side is free of $s$, the inequality still holds true when the left-hand side of the inequality is replaced by its supremum over $s \in \Scal$. Finally, we choose $K_\epsilon$ such that $\sum_{k=1}^{K_\epsilon+1}\vert \mu_k(\bb{x})\vert < \epsilon/2B_\infty$ and choose sufficiently large $n$ such that $(\widehat{\mu}_k(\bb{x}) - \mu_k(\bb{x})) = o_{\prob}(\epsilon/2K_\epsilon B_\infty)$ for all $k = 1, 2, \dots, K_\epsilon$. The proof for uniform consistency now follows the same steps as in the previous part.

\subsection{Proof of Theorem~\ref{thm:asymptotic-normality}}\label{proof-thm:asymptotic-normality}

Following the same approach as in~\cite{masry2005nonparametric} and~\cite{hong2020nonparametric}, we decompose the centered estimate of the mean components into 
\begin{equation*}
    \widehat{\mu}_k(\bb{x}) - \mu_k(\bb{x}) - \tilde{b}_{nk}(\bb{x}) = \dfrac{ Q_{n k}(\bb{x}) - \tilde{b}_{nk}(\bb{x}) (\widehat{\mu}_{k,1}(\bb{x}) - \E(\widehat{\mu}_{k,1}(\bb{x})) ) }{ \widehat{\mu}_{k,1}(\bb{x}) }    
\end{equation*}
where 
\begin{align*}
    \tilde{b}_{nk}(\bb{x}) & = \dfrac{ \E(\widehat{\mu}_{k,2}(\bb{x})) - \mu_k(\bb{x}) \E(\widehat{\mu}_{k,1}(\bb{x})) }{ \E(\widehat{\mu}_{k,1}(\bb{x})) } = \E(\widehat{\mu}_{k,2}(\bb{x})) - \mu_k(\bb{x}) \\
    Q_{nk}(\bb{x}) & = ( \widehat{\mu}_{k,2}(\bb{x}) - \E(\widehat{\mu}_{k,2}(\bb{x})) ) - \mu_k(\bb{x}) ( \widehat{\mu}_{k,1}(\bb{x}) - \E(\widehat{\mu}_{k,1}(\bb{x})) ) 
\end{align*}
First, we consider the bias term $\tilde{b}_{nk}$. In Lemma~\ref{lem:muk2-muk-conv}, we have already established that under suitable assumptions it is asymptotically negligible as $n \rightarrow \infty$, hence it is $o_{\prob}(1)$. Also, we already know by Lemma~\ref{lem:muk1-conv} that $\widehat{\mu}_{k,1}(\bb{x})$ converges in probability to $\E(\widehat{\mu}_{k,1}(\bb{x})) =  1$. Therefore, we can rewrite the above decomposition as 
\begin{equation*}
    \widehat{\mu}_k(\bb{x}) - \mu_k(\bb{x}) - \tilde{b}_{nk}(\bb{x}) = \dfrac{Q_{n k}(\bb{x})}{ \widehat{\mu}_{k,1}(\bb{x}) }(1 + o_{\prob}(1)).
\end{equation*}
Hence, the interesting quantity is $Q_{n,k}(\bb{x})$, whose asymptotic distribution governs the asymptotic limit of the mean estimator. Therefore, we consider the normalized sum
\begin{align*}
    \sqrt{n\phi_{\bb{x}}(h_n\lambda)} Q_{nk}(\bb{x})
    & = \dfrac{ \sqrt{\phi_{\bb{x}}(h_n \lambda)} }{\sqrt{n}\E(K_0)} \sum_{i=1}^n \left[ K_{t_i}\widehat{Y}_{t_i, k}^\ast - \mu_k(\bb{x}) K_{t_i} - \E\left( K_{t_i}\widehat{Y}_{t_i, k}^\ast - \mu_k(\bb{x}) K_{t_i} \right) \right] \\
    & = \tilde{v}_n \sum_{i=1}^n U_i + \tilde{v}_n \sum_{i=1}^n V_i
\end{align*}
where 
\begin{align*}
    U_i & = K_{t_i}\left( Y_{t_i, k}^\ast - \mu_k(\bb{x}) \right) - \E\left(K_{t_i}\left( Y_{t_i k} - \mu_k(\bb{x}) \right) \right) \\
    V_i & = K_{t_i}\left( \widehat{Y}_{t_i k}^\ast - Y_{t_i k}^\ast \right) - \E\left( K_{t_i}\left( \widehat{Y}_{t_i k}^\ast - Y_{t_i k}^\ast \right) \right),\\
    \text{and, } \tilde{v}_n & = n^{-1/2}\phi^{1/2}_{\bb{x}}(h_n \lambda) / \E(K_0).
\end{align*}
Note that, since $(\phi_{\bb{x}}^{-1}(h_n\lambda))\E(K_0) \rightarrow \xi_1$ where $\xi_1 \in (0, \infty)$ as $n \rightarrow \infty$, so $\tilde{v}_n = (n\phi_{\bb{x}}(h_n\lambda))^{-1/2} \xi_1^{-1} + o_{\prob}(1)$. As shown in \eqref{eqn:bandwidth-density-inf}, we have $n\phi_{\bb{x}}(h_n\lambda) \rightarrow \infty$, hence $\tilde{v}_n \rightarrow 0$.

Now, by Lemma~\ref{lem:Yk-convergence} and the boundedness of the kernel as in~\eqref{eqn:kernel-bounds}, it follows that each $V_i = \Ocal(\delta_k(\epsilon^\ast_{t_i}))$. Therefore,
\begin{equation*}
    \tilde{v}_n \sum_{i=1}^n V_i = \Ocal\left( \dfrac{n \sup_{t_i} \delta_k(\epsilon^\ast_{t_i}) }{\xi_1 \sqrt{n \phi_{\bb{x}}(h_n\lambda)}}  \right) = \Ocal\left( \sup_{t_i} \delta_k(\epsilon^\ast_{t_i}) \sqrt{\dfrac{n}{\phi_{\bb{x}}(h_n\lambda)}} \right) = o_{\prob}(1)
\end{equation*}
where the last line follows from Assumption~\ref{assum:bandwidth-zero}. Also, since $\widehat{\mu}_{k, 1} \xrightarrow{P} 1$ due to Lemma~\ref{lem:muk1-conv}, we have 
\begin{equation}
    \sqrt{n\phi_{\bb{x}}(h_n\lambda)} \left( \widehat{\mu}_k(x) - \mu_k(x) - \tilde{b}_{nk}(x)\right) = \tilde{v}_n \sum_{i=1}^n U_i \left( \dfrac{1 + o(1)}{1 - o(1)} \right).
    \label{eqn:mu-asymp-decomp}
\end{equation}
Therefore, we simply have to restrict our attention to the scaled sum $\tilde{v}_n \sum_{i=1}^n U_i$ alone.

Our proof for establishing the asymptotic distribution of the scaled sum $\tilde{v}_n \sum_{i=1}^n U_i$ follows similar to the big-block small-block approach as in~\cite{masry2005nonparametric}, but with some modifications to incorporate the PMC-type condition instead of the strong mixing conditions. We begin by picking integers $a_n, b_n$ such that $a_n = \Ocal(n^{\beta})$ and $b_n = \Ocal(n^\alpha)$ where $\alpha$ and $\beta$ are as asserted by Assumption~\ref{assum:bandwidth-zero}. Let $c_n = (a_n + b_n)$ and $g_n = [n/c_n]$, where $[x]$ denote the largest integer less than or equal to $x$. Then the sum $\tilde{v}_n \sum_{i=1}^n U_i$ can be further decomposed as
\begin{equation}
    \tilde{v}_n\sum_{i=1}^n U_i = \tilde{v}_n\sum_{j=1}^{g_n} \sum_{i=(j-1)c_n+1}^{(j-1)c_n + a_n} U_i + \tilde{v}_n\sum_{j=1}^{g_n} \sum_{i=(j-1)c_n + (a_n + 1)}^{jc_n} U_i + \tilde{v}_n\sum_{i=g_n c_n}^{n} U_i =: S_n +  B_n + R_n,
    \label{eqn:small-block-big-block-defn}
\end{equation}
where $S_n, B_n$ and $R_n$ are the quantities it is replacing. Note that, since $a_n/b_n \rightarrow 0$ as $n \rightarrow \infty$, it follows that $b_n / c_n \rightarrow 1$ and hence $g_n$ is in the asymptotic order of $n^{1-\alpha}$. Now, to show the proof of the asymptotic normality, the idea is to show that the contribution from the small block sum $S_n$ and the remainder $R_n$ are asymptotically negligible, and the summands of the big-block sum are asymptotically independent, which we will show through the multiplicative decomposition of the characteristic functions of the appropriate random variables. These ideas are presented in precise mathematical terms through the following series of lemmas.

\begin{lem}\label{lem:small-block-zero}
    Suppose the Assumptions~\ref{assum:mean-sigma-bound}-\ref{assum:bandwidth-zero} hold. Then as $n \rightarrow \infty$, $S_n \xrightarrow{P} 0$, where $S_n$ is as defined in \eqref{eqn:small-block-big-block-defn}.
\end{lem}

\begin{proof}
Clearly, $\E(S_n) = 0$. In view of Chebyshev's inequality, it is, therefore, enough to show that the variance of the small block sum goes to $0$ as $n \rightarrow \infty$. To this direction, we have
\begin{align*}
    \var(S_n) 
    & = \tilde{v}_n^2 \var\left( \sum_{j=1}^{g_n} \sum_{i=(j-1)c_n + 1}^{(j-1)c_n + a_n} U_i \right)\\
    & = \tilde{v}_n^2 \sum_{j=1}^{g_n} \var\left( \sum_{i=(j-1)c_n + 1}^{(j-1)c_n + a_n} U_i \right) + \tilde{v}_n^2 \sum_{j \neq l}^{g_n} \cov\left( \sum_{i=(j-1)c_n + 1}^{(j-1)c_n + a_n} U_i, \sum_{i=(l-1)c_n + 1}^{(l-1)c_n + a_n} U_i \right)\\
    & = \tilde{v}_n^2 \sum_{j=1}^{g_n} \left[ \sum_{i=(j-1)c_n + 1}^{(j-1)c_n + a_n} \var(U_i) + \sum_{i \neq i'} \cov(U_i, U_{i'})  \right] + \tilde{v}_n^2 \sum_{j\neq l}^{g_n} \sum_{i \neq i'} \cov(U_i, U_i')
\end{align*}
where the sums over $i \neq i'$ are taken over the appropriate groups. For the variance term, we have
\begin{align*}
    \var(U_i)
    & = \var\left( K_{t_i}(Y_{t_i k}^\ast - \mu_k(\bb{x})) \right)\\
    & = \E\left[ \var\left(  K_{t_i}(Y_{t_i k}^\ast - \mu_k(\bb{x})) \mid \Fcal_{t_i} \right) \right] + \var\left( \E\left[ K_{t_i}(Y_{t_i k}^\ast - \mu_k(\bb{x})) \mid \Fcal_{t_i} \right] \right)\\
    & = \E\left[ K_{t_i}^2 \int_{\Scal^2} \sigma(\bb{X}_{t_i}, s)\sigma(\bb{X}_{t_i}, s')\rho(s, s') b_k(s)b_k(s')ds ds'\right] + \E\left[K^2_{t_i}(\mu_k(\bb{X}_{t_i}) - \mu_k(\bb{x}))^2 \right]\\
    & = \E\left[ K_{0}^2 \int_{\Scal^2} \sigma(\bb{X}_{0}, s)\sigma(\bb{X}_{0}, s')\rho(s, s') b_k(s)b_k(s')ds ds'\right] + \E\left[K^2_{0}(\mu_k(\bb{X}_{0}) - \mu_k(\bb{x}))^2 \right]
\end{align*}
The last line follows from stationarity of $\bb{X}_t$. We bound the second term $\E\left[K^2_{0}(\mu_k(\bb{X}_{0}) - \mu_k(\bb{x}))^2 \right]$ by noting that, the random variable inside the expectation is nonzero only when $\bb{X}_0$ lies in the infinite-dimensional ball $B(\bb{x}, h_n\lambda)$. Additionally, when $\bb{X}_0$ is indeed inside the ball $B(\bb{x}, h_n\lambda)$, we can apply the Lipschitz condition presented in Assumption~\ref{assum:kernel-mu-lipschitz} to bound the term. In summary, we get
\begin{align*}
    \E\left[K^2_{0}(\mu_k(\bb{X}_{0}) - \mu_k(\bb{x}))^2 \right] & 
    \leqslant \E\left[ \indd{\bb{X}_0 \in B(\bb{x},h_n\lambda)} \left( \Ocal\left( \Vert \bb{D}^{-1}(\bb{X}_0 - \bb{x})\Vert  \right) + (K_0-K(0))^2\mu_k^2(x)  \right) \right]\\
    & \leqslant \Ocal(h_n \lambda) + \Ocal(\phi_{\bb{x}}(h_n\lambda)) \rightarrow 0,
\end{align*}
as $n \rightarrow \infty$. To deal with the first term, note that again it is enough to consider the situation when $\bb{X}_0$ is within the ball $B(\bb{x}, h_n\lambda)$. By a similar logic as above and using the Lipschitz continuity of $\Kcal(\Vert \bb{H}_n^{-1}(\bb{X}_0 - \bb{x})\Vert) \sigma_k(\bb{X}_0)$ as in Assumption~\ref{assum:kernel-mu-lipschitz}, we obtain that $K_0\sigma(\bb{X}_0, s) = K_0\sigma(\bb{x}, s) + o_{\prob}(1)$ for any fixed $s \in \Scal$. Since $\Scal$ is compact, this convergence may be strengthened to a uniform convergence as $\sigma(\bb{x}, s)$ is continuous over $s \in \Scal$ due to continuity of the basis functions. Therefore,
\begin{align}
    & \int_{\Scal^2} K_0^2 \left\vert \sigma(\bb{X}_0, s)\sigma(\bb{X}_0,s') - \sigma(\bb{x},s) \sigma(\bb{x},s') \right\vert \vert \rho(s,s')\vert b_k(s)b_k(s')ds ds' \nonumber\\
    \leqslant{} & \int_{\Scal^2} K_0^2\left( \vert \sigma(\bb{X}_0, s) \vert \vert \sigma(\bb{X}_0,s') - \sigma(\bb{x},s') \vert + \vert \sigma(\bb{x}, s') \vert \vert \sigma(\bb{X}_0,s) - \sigma(\bb{x},s) \vert \right) \vert \rho(s,s')\vert b_k(s)b_k(s')ds ds'\nonumber\\
    \leqslant{} & o_{\prob}(1) \left[ \int_{\Scal^2} \vert \sigma(\bb{X}_0, s) \vert b_k(s)b_k(s')dsds' + \int_{\Scal^2} \vert \sigma(\bb{x},s')\vert b_k(s)b_k(s')ds ds' \right] \label{eqn:sigma-cont-proof-1}\\
    ={} & o_{\prob}(1)\left[ \left( \int_{\Scal}\sigma^2(\bb{X}_0, s)ds \right)^{1/2} \int_{\Scal}b_k(s)ds + \left( \int_{\Scal}\sigma^2(\bb{x}, s')ds'\right)^{1/2} \int_{\Scal}b_k(s)ds  \right]\label{eqn:sigma-cont-proof-2}
\end{align}
The inequality \eqref{eqn:sigma-cont-proof-1} in the third step follows from noting that $\vert \rho(s,s')\vert \leqslant 1$. The inequality \eqref{eqn:sigma-cont-proof-2} in the fourth step is a simple application of Cauchy-Schwartz inequality and the fact that $\int_{\Scal} b_k^2(s)ds = 1$. Finally, by taking expectation on both sides, and using Assumption~\ref{assum:mean-sigma-bound} on the integrability of $\E(\sigma_k(\bb{X}_0, s))$ as a function of $s$, we note that the above quantity is also asymptotically negligible. Now, we use the relationship given in \eqref{eqn:kernel-convergence} to get
\begin{equation*}
    \E\left[ K_{0}^2 \int_{\Scal^2} \sigma(\bb{X}_{0}, s)\sigma(\bb{X}_{0}, s')\rho(s, s') b_k(s)b_k(s')ds ds'\right]
    = \phi_x(h_n\lambda) \xi_2 \sigma_{kk}(\bb{x}) + o_{\prob}(1),
\end{equation*}
where $\sigma_{k,k}(\bb{x})$ is as defined in \eqref{eqn:mode-reduced-variance-estimate}. Together, we have $\var(U_i) = \phi_x(h_n\lambda) \xi_2 \sigma_{k k}(x) + o_{\prob}(1)$. 

For the covariance term, with $i < i'$, analogus to the above, we obtain
\begin{align*}
    \cov(U_i, U_{i'})
    & = \cov\left( K_{t_i}(Y_{t_i k}^\ast - \mu_k(\bb{x})), K_{t_{i'}}(Y_{t_{i'} k}^\ast - \mu_k(\bb{x})) \right)\\
    & = \E\left[ K_{t_i}K_{t_{i'}} \int_{\Scal^2} \sigma(\bb{X}_{t_i},s)\sigma(\bb{X}_{t_{i'}}, s')\rho(s, s')b_k(s)b_k(s')ds ds' \right] \\
    & \qquad \qquad + \E\left[ K_{t_i}K_{t_{i'}}(\mu_k(\bb{X}_{t_i}) - \mu_k(\bb{x}))(\mu_k(\bb{X}_{t_{i'}}) - \mu_k(\bb{x})) \right]
\end{align*}
Note that, both of the terms are nonzero only when both $\bb{X}_{t_i}$ and $\bb{X}_{t_{i'}}$ lie in the infinite-dimensional ball $B(\bb{x}, h_n\lambda)$. By applying the Lipschitz conditions as before and using the bounds on the joint small ball probabilities described in Section~\ref{sec:kernel-choice}, it is easy to obtain that
\begin{align*}
    \E\left[ K_{t_i}K_{t_{i'}}(\mu_k(\bb{X}_{t_i}) - \mu_k(\bb{x}))(\mu_k(\bb{X}_{t_{i'}}) - \mu_k(\bb{x})) \right] 
    & = \Ocal(\phi_{\bb{x}}^2(h_n\lambda)) \rightarrow 0,\\
    \E\left[ K_{t_i}K_{t_{i'}} \int_{\Scal^2} \sigma(\bb{X}_{t_i},s) \sigma(\bb{X}_{t_{i'}}, s')\rho(s, s')b_k(s)b_k(s')ds ds' \right] 
    & = \E(K_{t_i}K_{t_{i'}}) \sigma_{k,k}(\bb{x}) + o_{\prob}(1)\\
    & = \Ocal(\phi_{\bb{x}}^2(h_n\lambda)) \rightarrow 0,
\end{align*}
under the asymptotic regime considered in the paper. As a result, we have $\cov(U_i, U_{i'}) \rightarrow 0$ as $n \rightarrow \infty$ for any $i \neq i'$. Putting everything back together, we obtain that 
\begin{align*}
    \var(S_n) 
    & = \tilde{v}_n^2 \phi_{\bb{x}}(h_n \lambda)\xi_2 g_n a_n \sigma_{k k}(\bb{x}) + o( \tilde{v}_n^2 g_n a_n^2 \phi_{\bb{x}}^2(h_n \lambda) ) + \Ocal\left( \tilde{v}_n^2 g_n^2 a_n^2 \phi_{\bb{x}}^2(h_n\lambda) \right)\\
    & = \dfrac{n^\beta \xi_2}{n^\alpha \xi_1^2} \sigma_{k k}(\bb{x}) + o\left( \dfrac{n^{2\beta}\phi_{\bb{x}}(h_n\lambda) }{n^\alpha \xi_1^2} \right) + \Ocal\left( \dfrac{n^{2\beta} n^{1-2\alpha} \phi_{\bb{x}}(h_n\lambda) }{ \xi_1^2 } \right)\\
    & = \dfrac{n^{\beta - \alpha} \xi_2}{\xi_1^2} \sigma_{k k}(\bb{x}) + o\left( \dfrac{n^{2\beta - \alpha}\phi_{\bb{x}}(h_n\lambda) }{\xi_1^2} \right) + \Ocal\left( \dfrac{n^{1 - 2\alpha + 2\beta} \phi_{\bb{x}}(h_n\lambda) }{ \xi_1^2 } \right)
\end{align*}
It is now immediate that all the terms are asymptotically negligible due to Assumption~\ref{assum:bandwidth-zero} and since $\alpha < 1$.
\end{proof}

Since the number of terms contained in $R_n$ is even smaller than the number of terms in $S_n$, the contribution of the remainder term $R_n$ is also asymptotically negligible as $n \rightarrow \infty$. Now, for the big-block sum $B_n$, we wish to show that the summands of $B_n$ are asymptotically independent. Then, an application of the central limit theorem as in~\cite{masry2005nonparametric} can be applied to establish the asymptotic normality of the $B_n$. To show this asymptotic independence result, our objective will be to demonstrate that the characteristic function of $B_n$ asymptotically splits into a product of characteristic functions of individual summands $\sum_{i=(j-1)c_n + a_n+1}^{jc_n} U_i$ (to be denoted as $B_{n,j}$ for notational convenience), adjusting for the proper normalization factor $\tilde{v}_n$. In view of this, we obtain the relationship $B_n = \tilde{v}_n \sum_{j=1}^{g_n} B_{n,j}$. 

\begin{lem}\label{lem:big-block-factor}
    Suppose that Assumptions~\ref{assum:mean-sigma-bound}-\ref{assum:bandwidth-zero} hold. Then, as $n \rightarrow \infty$, 
    \begin{equation*}
        \left\vert \E( e^{i t B_n} ) - \prod_{j=1}^{g_n} \E( e^{i t \tilde{v}_n B_{n,j} } )  \right\vert \rightarrow 0,
    \end{equation*}
    where $B_n$ is the big-block as in \eqref{eqn:small-block-big-block-defn} and $B_{n,j} = \sum_{i=(j-1)c_n + a_n+1}^{jc_n} U_i$ for each $j = 1, \dots, g_n$.
\end{lem}

\begin{proof}
It follows that
\begin{align*}
    & \left\vert \E( e^{i t B_n} ) - \prod_{j=1}^{g_n} \E( e^{i t \tilde{v}_n B_{n,j} } )  \right\vert\\
    = {} & \left\vert \E( e^{i t B_{g_n c_n}} ) - \prod_{j=1}^{g_n} \E( e^{i t \tilde{v}_n B_{n,j} } ) \right\vert\\
    \leqslant {} & \left\vert \E( e^{i t B_{g_n c_n}} ) - \E( e^{i t B_{(g_n-1) c_n}} )\E\left( e^{it \tilde{v}_n B_{n,g_n} }\right) \right\vert + \left\vert \E( e^{i t B_{(g_n-1) c_n}} )\E\left( e^{it \tilde{v}_n B_{n,g_n} }\right) - \prod_{j=1}^{g_n} \E( e^{i t \tilde{v}_n B_{n,j} } ) \right\vert\\
    = {} & \left\vert \E( e^{i t B_{(g_n-1) c_n}} )\E\left( e^{it \tilde{v}_n B_{n,g_n} } \mid \Fcal_{t_{(g_n - 1)c_n}} \right) - \E( e^{i t B_{(g_n-1) c_n}} )\E\left( e^{it \tilde{v}_n B_{n,g_n} }\right) \right\vert + \\
    & \qquad \left\vert \E\left( e^{it \tilde{v}_n B_{n,g_n} }\right) \right\vert \left\vert \E( e^{i t B_{(g_n-1) c_n}} ) - \prod_{j=1}^{g_n - 1} \E( e^{i t \tilde{v}_n B_{n,j} } ) \right\vert\\
    \leqslant{} & \left\vert \E\left( e^{it \tilde{v}_n B_{n,g_n} } \mid \Fcal_{t_{(g_n - 1)c_n}} \right) - \E\left( e^{it \tilde{v}_n B_{n,g_n} }\right) \right\vert + \left\vert \E( e^{i t B_{(g_n-1) c_n}} ) - \prod_{j=1}^{g_n - 1} \E( e^{i t \tilde{v}_n B_{n,j} } ) \right\vert
\end{align*}
where in the last inequality we use the fact that $\vert \E(e^{itX})\vert \leqslant \E(\vert e^{itX}\vert) = 1$ for any complex-valued random variable $X$. Note that, this creates an inductive inequality on the factorization provided, namely we provide a bound on the factorization error up to the term $g_nc_n$ using the factorization error till the term $(g_n - 1)c_n$, plus an additional term which takes care of the last term $B_{n, g_n}$ based on a conditional expectation.

To deal with the first term of the right-hand side of the above inequality, consider the difference
\begin{multline}
    \left\vert \E\left( \cos\left(t \tilde{v}_n B_{n, g_n} \mid \Fcal_{t_{(g_n-1)c_n}} \right) \right) - \E\left( \cos\left( t\tilde{v}_n B_{n, g_n} \right) \right) \right\vert\\
    = b_n \sup_{(g_n-1)c_n + a_n + 1 \leqslant i \leqslant g_n c_n} \left\vert \E(\cos(t \tilde{v}_n U_i) \mid \Fcal_{t_{(g_n-1)c_n}} ) - \E(\cos(t \tilde{v}_n U_i)) \right\vert.
    \label{eqn:cosine-diff-proof-1}
\end{multline}
Now, $\cos(x)$ is Lipschitz with constant $1$ and the function $f(y) = K(\Vert \bb{H}_n^{-1}(y-x)\Vert) \mu_k(y)$ is also Lipschitz due to the Assumption~\ref{assum:kernel-mu-lipschitz}. Also, $\E(\sigma(X_{t_i},s)\epsilon_{t_i}(s)) = 0$ for all $s \in \Scal$ due to independence of $\epsilon_{t_i}$ and $\Fcal_{t_i}$ and $\tilde{v}_n = o(1)$. By combining all of these along with Proposition~\ref{lem:f-Xk-bound}, it follows that the difference in the cosine term as in \eqref{eqn:cosine-diff-proof-1} is $\Ocal(\Delta_2(a_n))$ for all $t \in [-M', M']$ for some arbitrarily large but fixed constant $M' > 0$. 
A similar derivation holds for bounding the discrepancy through the sine function as well, i.e.,
\begin{equation}
    \left\vert \E\left( \sin\left(t \tilde{v}_n B_{n, g_n} \mid \Fcal_{t_{(g_n-1)c_n}} \right) \right) - \E\left( \sin\left( t\tilde{v}_n B_{n, g_n} \right) \right) \right\vert = \Ocal(\Delta_2(a_n)).
    \label{eqn:sine-diff-proof-1}
\end{equation}
In particular, by combining the inequalities from \eqref{eqn:cosine-diff-proof-1} and \eqref{eqn:sine-diff-proof-1}, we obtain that
\begin{equation*}
    \left\vert \E\left( e^{it \tilde{v}_n B_{n,g_n} } \mid \Fcal_{t_{(g_n - 1)c_n}} \right) - \E\left( e^{it \tilde{v}_n B_{n,g_n} }\right) \right\vert
    = \Ocal(b_n \Delta_2(a_n)),
\end{equation*}
for any $t \in [-M, M]$, some interval around $0$. As a result, by repeatedly using this bound, we obtain
\begin{equation*}
    \left\vert \E( e^{i t B_n} ) - \prod_{j=1}^{r} \E( e^{i t \tilde{v}_n B_{n, j}} )  \right\vert
    \leqslant \Ocal(b_n\Delta_2(a_n)) + \left\vert \E( e^{it B_{(g_n-1)c_n }} ) - \prod_{j=1}^{g_n-1} \E( e^{i t \tilde{v}_n B_{n,j} } ) \right\vert = \Ocal(g_n b_n \Delta_2(a_n)).
\end{equation*}
Since $g_n = [n/c_n]$ and $b_n / c_n \rightarrow 1$, we have $\Ocal(g_n b_n \Delta_2(a_n)) = \Ocal(n a_n^{-\tau+1}) = \Ocal(n^{1 - \beta(\tau - 1)})$. Again, due to Assumption~\ref{assum:bandwidth-zero}, this difference tends to $0$ asymptotically as $n \to \infty$. 
\end{proof}

Lemma~\ref{lem:big-block-factor} now establishes that the summands of the big blocks are asymptotically independent as $n \rightarrow \infty$. As demonstrated before in the proof of Lemma~\ref{lem:small-block-zero}, we have $\E(B_n) = 0$ and 
\begin{align*}
    \var(B_n) 
    & = \tilde{v}_n^{2} g_n b_n \var(U_i) +  \tilde{v}_n^{2} g_n b_n^2 \cov(U_i, U_{i'}) + o(1)\\
    & = \dfrac{\xi_2}{\xi_1^2}\sigma_{k k}(x) + \Ocal\left( \dfrac{n^{\alpha} \phi_x(h_n\lambda) }{\xi_1^2} \right)
\end{align*}
Note that, here we use the asymptotic independence between $U_i$ and $U_{i'}$ for $i \neq i'$ established in the proof of Lemma~\ref{lem:small-block-zero} to calculate the covariance term. The remainder part can be made small using Assumption~\ref{assum:bandwidth-zero}. Therefore, an application of the Central Limit Theorem as in~\cite{masry2005nonparametric} shows that the normalized sum of the big blocks, i.e., $B_n$ is asymptotically normal with mean zero and variance $\xi_2\xi_1^{-2} \sigma_{k k}(\bb{x})$. 

Finally, combining this result with Lemma~\ref{lem:small-block-zero}, we obtain the asymptotic distribution of the normalized estimate of the mean component, completing the proof.

\subsection{Proof of Theorem~\ref{thm:sigma-pointwise}}\label{proof-thm:sigma-pointwise}

Let us fix $\bb{x} \in \chi$ and denote $K_{t_i} = K(\Vert \bb{H}_n^{-1}(\bb{X}_{t_i} - \bb{x})\Vert)$ and $w_i = K_{t_i}/\sum_{i} K_{t_i}$. Then we can decompose the estimate $\widehat{\sigma}_{k l}(\bb{x})$ as 
\begin{align*}
    \widehat{\sigma}_{k l}(\bb{x}) 
    & = \sum_{i} w_i (\widehat{Y}^\ast_{t_i k} - \widehat{\mu}_k(\bb{x}))(\widehat{Y}^\ast_{t_i l} - \widehat{\mu}_l(\bb{x}))\\
    & = \sum_i w_i \left( \widehat{Y}^\ast_{t_i k} - Y_{t_i k}^\ast - \widehat{\mu}_k(\bb{x}) + \mu_k(\bb{x}) + \eta_{t_i k} \right)\left( \widehat{Y}^\ast_{t_i l} - Y_{t_i l}^\ast - \widehat{\mu}_l(\bb{x}) + \mu_l(\bb{x}) + \eta_{t_i l} \right).
\end{align*}
First, we apply Proposition~\ref{lem:Yk-convergence} and an application of Cauchy-Schwartz inequality to conclude that for any $k, l \in \Z^+$, we have
\begin{equation*}
    \E\left(\sum_i w_i \left\vert \widehat{Y}^\ast_{t_i k} - Y_{t_i k}^\ast \right\vert \left\vert \widehat{Y}^\ast_{t_i, l} - Y_{t_i l}^\ast \right\vert \right)
    \leqslant \left( \sup_{t_i} \E\left\vert \widehat{Y}^\ast_{t_i k} - Y_{t_i k}^\ast \right\vert^2 \right)^{1/2} \left( \sup_{t_i} \E\left\vert \widehat{Y}^\ast_{t_i l} - Y_{t_i l}^\ast \right\vert^2 \right)^{1/2} = o_{\prob}(1).
\end{equation*}
Due to the asymptotic normality of the estimator of the mean as in Theorem~\ref{thm:asymptotic-normality}, it follows that for any fixed $k \in \Z^+$ and fixed $\bb{x} \in \chi$,
\begin{equation*}
    \E\vert \mu_k(x) - \widehat{\mu}_k(x)\vert^2 = \E^2\vert \mu_k(x) - \widehat{\mu}_k(x)\vert + \var(\widehat{\mu}_k(x))
    = o_{\prob}(1) + \Ocal((n\phi_x(h_n\lambda))^{-1/2}) = o_{\prob}(1),
\end{equation*}
as $n \rightarrow \infty$, where the last equality follows from the Assumption~\ref{assum:bandwidth-decay} on the choice of bandwidth. Additionally, Assumption~\ref{assum:bounded-errors} implies that $\eta_{t_i k}$ is uniformly bounded over $t_i \in [0, 1]$ by a deterministic constant free of $n$. Therefore, as $n \rightarrow \infty$, we obtain
\begin{align*}
    & \E\vert \widehat{\sigma}_{k l}(\bb{x}) - \sigma_{k,l}(\bb{x})\vert\\
    = {} & o_{\prob}(1) + \E\left\vert \sum_i w_i \eta_{t_i k}\eta_{t_i l} - \sigma_{k l}(\bb{x}) \right\vert\\
    = {} & o_{\prob}(1) + \dfrac{1}{\E(K_0)}\E\left\vert \dfrac{1}{n}\sum_i K_{t_i} (\eta_{t_i k}\eta_{t_i l} - \sigma_{k l}(\bb{x})) \right\vert \\
    \leqslant {} & o_{\prob}(1) + \dfrac{1}{\E(K_0)} \E\left\vert \sum_i \dfrac{K_{t_i}}{n}(\eta_{t_i k}\eta_{t_i l} - \E(\eta_{t_i k}\eta_{t_i l}\mid \Fcal_{t_i})) \right\vert + \dfrac{1}{\E(K_0)}\E\left\vert \sum_i \dfrac{K_{t_i}}{n} (\E(\eta_{t_i k}\eta_{t_i l}\mid \Fcal_{t_i}) - \sigma_{k,l}(\bb{x}) ) \right\vert\\
    \leqslant {} & o_{\prob}(1) + \dfrac{1}{\E(K_0)} \E\left\vert \dfrac{1}{n}\sum_i V_i \right\vert + \dfrac{1}{\E(K_0)}\E\left\vert \dfrac{1}{n}\sum_i K_{t_i} (\sigma_{k l}(\bb{X}_{t_i}) - \sigma_{k l}(\bb{x})) \right\vert,
\end{align*}
where 
\begin{equation*}
    V_i = K_{t_i}\eta_{t_i k}\eta_{t_i,l} - \E(K_{t_i}\eta_{t_i k}\eta_{t_i l} \mid \Fcal_{t_i}),
\end{equation*}
and we make use of the stationarity of $\bb{X}_t$ and the form of covariance given in \eqref{eqn:model-reduced-variance-cond}. We deal with both of these terms separately.

Since $\E(V_i \mid \Fcal_{t_i}) = 0$ for all $i = 1, \dots, n$, we can apply the same technique shown in the proof of Lemma~\ref{lem:muk2-conv} based on the $L^1$-type convergence result by~\cite{andrews1988lln}, to show that the first term of the bound is $o_{\prob}(1)$ as $n \rightarrow \infty$.

For the second term of the bound, choose any $\epsilon > 0$. By the existence of the basis decomposition as in \eqref{eqn:mean-sigma-basis}, we note that there exists $U_{\epsilon,1}$ such that for the fixed $\bb{x} \in \chi$, $\sum_{v=1}^{U_{\epsilon, 1}} \vert \sigma_v(\bb{x})\vert < \epsilon$. Additionally, because of the decomposition in \eqref{eqn:sigma-k-l-defn}, we must have some $U_{\epsilon,2}$ such that the tail series of the variance decomposition satisfy
\begin{equation*}
    \sum_{u,v: \max(u, v) > U_{\epsilon, 2}} c_{u,v}(k,l) \E(\sigma_u(\bb{X}_{0})\sigma_v(\bb{X}_0)) < \epsilon.
\end{equation*}
Therefore, by choosing $U_\epsilon = \max(U_{\epsilon, 1}, U_{\epsilon, 2})$ and using the boundedness of the kernel function as in~\eqref{eqn:kernel-bounds}, we can write the second term of the bound as,
\begin{align*}
    & \E\left\vert \dfrac{1}{n}\sum_{i} K_{t_i} (\sigma_{k l}(\bb{X}_{t_i}) - \sigma_{k l}(\bb{x})) \right\vert \\
    ={} & \E\left\vert \dfrac{1}{n}\sum_i K_{t_i} \sum_{u,v}c_{u,v}(k,l) (\sigma_{u}(\bb{X}_{t_i}) \sigma_v(\bb{X}_{t_i}) - \sigma_u(\bb{x})\sigma_v(\bb{x})) \right\vert \\
    ={} & \Ocal(\epsilon) + \E\left\vert \dfrac{1}{n}\sum_{i}K_{t_i} \sum_{u=1}^{U_\epsilon}\sum_{v=1}^{U_\epsilon} c_{u,v}(k,l) (\sigma_{u}(\bb{X}_{t_i}) \sigma_v(\bb{X}_{t_i}) - \sigma_u(\bb{x})\sigma_v(\bb{x})) \right\vert\\
    \leqslant {} & \Ocal(\epsilon) + \E\left\vert \dfrac{1}{n}\sum_{u=1}^{U_\epsilon}\sum_{v=1}^{U_\epsilon} c_{u,v}(k,l) \sum_{i=1}^n K_{t_i}\sigma_{u}(\bb{X}_{t_i}) (\sigma_v(\bb{X}_{t_i}) - \sigma_v(\bb{x})) \right\vert \\
    & \qquad + \E\left\vert \dfrac{1}{n} \sum_{u=1}^{U_\epsilon}\sum_{v=1}^{U_\epsilon} c_{u,v}(k,l) \sum_{i=1}^n K_{t_i}\sigma_{v}(\bb{x}) (\sigma_u(\bb{X}_{t_i}) - \sigma_u(\bb{x})) \right\vert \\
    ={} & \Ocal(\epsilon) + 2U_\epsilon^2 C (h_n \lambda),
\end{align*}
for some constant $C \in (0, \infty)$. Here, we use the fact that $c_{u,v}(k,l) \leqslant 1$ and the function $K_{t_i}\sigma_u(\bb{x})$ is Lipschitz as assured by Assumption~\ref{assum:kernel-mu-lipschitz}. By choosing sufficiently large $n$ (depending on $\epsilon$), as the bandwidth $h_n \rightarrow \infty$, we can ensure that the second term $2U_\epsilon^2 C (h_n \lambda) \leqslant \epsilon$. This makes the entire different $\Ocal(\epsilon)$. Since $\epsilon$ is arbitrary, this completes the proof.

\subsection{Proof of Theorem~\ref{thm:simult-conf-1}}\label{proof-thm:simult-conf-1}

Due to the specific choice of the $\chi_n$, for any $i = 1, 2, \dots, n$, we have either $\Vert{\bb{D}^{-1}(\bb{X}_{t_i} - \bb{x})}\Vert > h_n\lambda$ or $\Vert{\bb{D}^{-1}(\bb{X}_{t_i} - \bb{x}')}\Vert > h_n\lambda$. As the type-I kernel function has bounded support, it follows that 
\begin{equation*}
    \mathcal{K}(\bb{H}_n^{-1}(\bb{X}_{t_i} - \bb{x}) )\mathcal{K}(\bb{H}_n^{-1}(\bb{X}_{t_i} - \bb{x}') ) = 0, 
\end{equation*}
\noindent for all $i = 1, 2, \dots, n$. By the form of the estimates of the mean components as in \eqref{eqn:mean-estimator}, it follows that the covariance between $\widehat{\mu}_k(\bb{x})$ and $\widehat{\mu}_l(\bb{x}')$ is zero for any $k, l \in \Z^+$ and $\bb{x} \neq \bb{x}' \in \chi_n$. The same conclusion holds for the variance estimates $\widehat{\sigma}_{k l}(\bb{x})$ and $\widehat{\sigma}_{k' l'}(\bb{x}')$ as well for any $k, k', l, l' \in \Z^+$. Therefore, by using the same principle of extreme value theory as in~\cite{zhao2008confidence}, we can obtain the asymptotic probability
\begin{multline}
    \lim_{n \rightarrow \infty}\mathbb{P}\left( \sup_{\bb{x} \in \chi_n} \dfrac{\xi_1 \sqrt{n\phi_{\bb{x}}(h_n\lambda)}}{\sqrt{\xi_2 Q_{K}(\bb{x},s) }} \left\vert \widehat{\mu}_{1:K}(\bb{x},s) - \mu_{1:K}(\bb{x},s) - b^\ast_{n, 1:K}(\bb{x}, s) \right\vert < B_{m_n}(z) \right) = e^{-2e^{-z}},
    \label{eqn:conf-simult-proof-1}
\end{multline}
for any fixed $z > 0$, $K \in \Z^+$ and $s \in \Scal$. This produces a probabilistic bound for the sum of first $K$ components of the mean function. Now, for the given $\delta > 0$, consider an $\delta$-net $\mathcal{N}(\delta)$ of the set $\chi_n$, i.e., for every $\bb{x} \in \chi_n$, there exists a point $\bb{x}_0 \in \mathcal{N}(\delta)$ such that $\Vert{\bb{D}^{-1}(\bb{x} - \bb{x}_0)}\Vert < \delta$. Since $\chi_n$ is bounded and $\delta$ is free of $n$, it follows that $\vert \mathcal{N}(\delta)\vert$ is finite and free of $n$. Also, 
\begin{align*}
    \left\vert \sum_{k=(K+1)}^\infty \mu_k(\bb{x})b_k(s)\right\vert 
    & \leqslant \left\vert \sum_{k=K+1}^\infty (\mu_k(\bb{x}) - \mu_k(\bb{x}_0))b_k(s) \right\vert + \left\vert \sum_{k=K+1}^\infty \mu_k(\bb{x}_0)b_k(s) \right\vert \\
    & \leqslant b_\infty(s) \sup_{\bb{x} \in N_\delta(\bb{x}_0) } \left\vert \sum_{k=K+1}^\infty (\mu_k(\bb{x}) - \mu_k(\bb{x}_0)) \right\vert + b_\infty(s) \sum_{k=K+1}^\infty \vert \mu_k(\bb{x}_0)\vert.
\end{align*}
Since $\sum_{k=1}^\infty \mu_k(\bb{x})$ is uniformly convergent for all $\bb{x} \in N_\delta(\bb{x}_0)$, it follows that for any given $\epsilon > 0$, there exists $K_\epsilon(\bb{x}_0)$ such that the above bound is less than $\epsilon b_\infty(s)$. Taking supremum over both sides yields,
\begin{equation*} 
    \sup_{\bb{x} \in \chi_n} \left\vert \sum_{k=(K+1)}^\infty \mu_k(\bb{x})b_k(s)\right\vert \leqslant \sup_{\bb{x}_0 \in \mathcal{N}(\delta)} b_\infty(s) \left[ \sup_{\bb{x} \in N_\delta(\bb{x}_0) } \left\vert \sum_{k=K+1}^\infty (\mu_k(\bb{x}) - \mu_k(\bb{x}_0)) \right\vert + \sum_{k=K+1}^\infty \vert \mu_k(\bb{x}_0)\vert \right].
\end{equation*}
Since $\mathcal{N}(\delta)$ has finitely many elements, taking $K_\epsilon =  \max_{\bb{x}_0 \in \mathcal{N}(\delta)} K_{\epsilon}(\bb{x}_0)$ we can ensure that the right-hand side of the above quantity is less than $\epsilon b_\infty(s)$.

For ease of notation, let us denote 
\begin{align*}
    \alpha_{n,K} =  & \dfrac{\xi_1 \sqrt{n\phi_{\bb{x}}(h_n\lambda)}}{\sqrt{\xi_2 Q_{K}(\bb{x},s) }}, \\
    \beta_{n,K} =  & \widehat{\mu}_{1:K}(\bb{x},s) - \mu_{1:K}(\bb{x},s) - b^\ast_{n, 1:K}(\bb{x}, s), \\
    \gamma_{n,K} =  & \widehat{\mu}_{1:K}(\bb{x},s) - \mu(\bb{x},s) - b^\ast_{n, 1:K}(\bb{x}, s) = \beta_{n,K} - \sum_{k=(K+1)}^\infty \mu_k(\bb{x})b_k(s).
\end{align*}
By an application of triangle inequality, it follows that
\begin{equation*}
    \alpha_{n,K_\epsilon}(\vert \beta_{n,K_\epsilon}\vert - \epsilon b_\infty(s) ) \leqslant \alpha_{n,K_\epsilon}\vert \gamma_{n,K_\epsilon}\vert \leqslant  \alpha_{n,K_\epsilon}(\vert \beta_{n,K_\epsilon}\vert + \epsilon b_\infty(s) ).
\end{equation*}
Combining this with \eqref{eqn:conf-simult-proof-1} now completes the proof.

\section{Real-life example: Analyzing Soccer data}

As another interesting illustration, we apply our proposed methodology to analyze the shot accuracy patterns in the English Premier League (EPL) across multiple seasons. The dataset is obtained from the \texttt{worldfootballR} package \citep{worldfootballR}. In a soccer match, various players from both teams take shots from different positions on the field, and as mentioned in \Cref{sec:intro}, xG quantifies the quality of each attempt based on the shot characteristics. The reader is referred to \cite{anzer2021goal} for a detailed exposition on this metric. However, it is important to recognize that xG does not explicitly account for team-specific or player-specific effect. We undertake a small-scaled analysis to uncover how this measure varies spatially across the field and to identify systematic effects associated with different teams and seasons. Specifically, we analyze the attacking patterns of Arsenal, Liverpool and Manchester City for three different seasons -- 2017/18, 2021/22, and 2024/25. These three teams are chosen for specific reasons: since 2017/18, only Manchester City and Liverpool have managed to win the EPL whereas during the same period Arsenal have been on a path of steady growth and are at the top of the league table in the ongoing season. In the three seasons we consider, Manchester City were the winners in the first two and finished third in the last one. Liverpool won the league in 2024/25 and finished fourth and second in the other two seasons respectively. Arsenal, on the other hand, finished in the sixth, fifth and second position in the league during these three seasons.

For modeling purposes, the response variable $Y_t(s)$ is taken to be the logit-transformed xG for a shot taken from location $s$ in the $t^{th}$ match. The covariate vector $\bb{X}_t$ is defined as team-specific profile at the $t^{th}$ match ($t$ can be 1 to 38 during a season), including the proportion of matches already played in the season, win–loss ratios, average goals scored and conceded, mean number of shots, and proportion of shots on target, among others. Figure~\ref{fig:football} displays the estimated mean xG surface -- the sigmoid function applied on $\hat{\mu}(\bb{X}_t, s)$ -- for the three representative teams in the three different seasons as mentioned. For illustration, these estimates are obtained under the setting where the covariates are fixed at an average level. 

\begin{figure}[!ht]
    \centering
    \includegraphics[width=\linewidth]{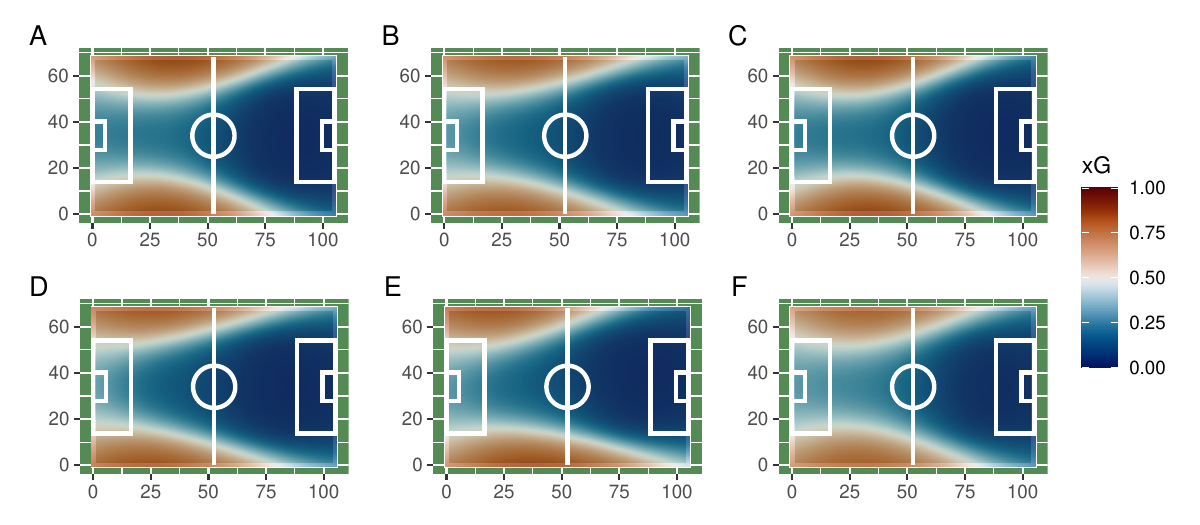}
    \caption{Estimated Spatial Surface for expected goals (xG) for Arsenal, Liverpool, and Manchester City in different seasons, when playing against an equally strong team.}
    \label{fig:football}
\end{figure}

Typically, we observe that the shots are most effective from central locations close to goal, particularly within the penalty area, and this canonical spatial structure is present across all estimated surfaces, serving as a basic validation that the fitted mean function aligns with domain expectations. The key insight is in the team and season-specific deviations from this baseline geometry. For Arsenal, the last season's estimated surface exhibits a comparatively broader region of moderate expected xG outside the penalty box, which is consistent with a shift toward taking (or creating) shots from longer range and wider areas. Although this should not be interpreted as claiming that long shots become high-probability outcomes, we find it imperative to highlight that Arsenal's ranking in the league table improved in 2025. Statistically, our model suggests a relative increase in the conditional expected shot quality associated with those locations, once the match-level context $\bb{X}_t$ is taken into account. 

For Liverpool, we notice a more concentrated pattern of xG in the central locations inside the box. It suggests that the chances created by Liverpool were more impactful in the last season, potentially becoming critical in helping them win the league. One must further observe that it is qualitatively similar to Manchester City's surfaces from the first two seasons under consideration, where they emerged victorious in the end. Generally, it suggests a stable attacking mechanism that consistently produces high-quality chances from the most dangerous zones. In contrast, Manchester City's estimated xG during 2024/25 was less pronounced and can directly be connected to them failing to finish in top two. 

Overall, the analysis summarizes how the spatial distribution of shot quality varies across teams and seasons: while all teams exhibit the expected concentration of high xG near goal, the extent and spread of moderate-to-high xG regions differ in ways that are consistent with changes in performance across seasons. These patterns provide a compact, interpretable representation of team-specific attacking profiles. With this context in mind, a further advantage of our proposed framework is that these comparisons are not purely descriptive heatmaps: they are conditional spatial surfaces $\hat{\mu}(\bb{X}_t,s)$, so differences across teams or seasons can be interpreted as systematic effects beyond what is explained by observable team-level covariates. Finally, the results are obtained by fixing equally strong opponent for both the teams. This setting should be read as a controlled scenario in which opponent-related components of $\bb{X}_t$ are held at comparable levels. Consequently, observed differences in the surfaces are most naturally attributed to team-specific spatial shot-generation profiles rather than imbalances in opponent strength.

\section{Simulation Studies}\label{sec:simulation}

To empirically verify the results on the behaviour of the proposed estimator, we consider a simple numerical exercise as follows. Let $n$ be the number of timepoints and $p^2$ be the number of spatial locations at each timepoint. We pick regularly spaced time and location intervals, namely $\Tcal_n = \{ 0, 1/(n-1), \dots, 1 \}$ and a regularly spaced grid $\Scal$ on $[0, 1]^2$ given by the points $\{ (i/(p-1), j/(p-1)) : i, j \in 0, 1, \dots, (p-1) \}$. Since, $\Scal \subset \R^2$, we pick a set of polynomial basis functions for the Hilbert space given by
\begin{align*}
    b_1(s) & = 1, &
    b_2(s) & = \sqrt{12}\left( s_1 - \frac{1}{2} \right), &
    b_3(s) & = \sqrt{12}\left( s_2 - \frac{1}{2} \right)\\
    b_4(s) & = 6\sqrt{5}\left( s_1^2 - s_1 + \frac{1}{6} \right), &
    b_5(s) & = 6\sqrt{5}\left( s_2^2 - s_2 + \frac{1}{6} \right), &
    b_6(s) & = 12\left( s_1 - \frac{1}{2} \right)\left( s_2 - \frac{1}{2} \right),
\end{align*}
where $s = (s_1, s_2) \in \Scal$. The above set of basis functions is derived from a shifted version of the Legendre polynomials upto degree $2$ on $[0, 1]^2$. 

First, we generate the covariate $X_t$ based on an autoregressive model
\begin{equation*}
    X_t = 1 + 0.5X_{t-1} + e_t, \ e_t \sim N(0, 0.1^2).
\end{equation*}
Then, we generate the response variable as
\begin{equation*}
    Y_t(s) = X_t(s_1 + s_2) + \epsilon_{t}(s), (\epsilon_t(s_1), \dots, \epsilon_t(s_{p^2})) \sim N\left(0, 0.1^2 ((e^{-\Vert s_i - s_j\Vert^2 }))_{i,j=1}^{p^2}  \right),
\end{equation*}
so that $\mu(X_t, s) = X_t(s_1 + s_2)$ and $\sigma(X_t, s) = 1$. It is clear that, the provided truncated basis functions can exactly express the mean and the covariance terms as a linear combination without requiring the higher order terms. Figure~\ref{fig:sim-sample} showcases one such sample realization.

\begin{figure}[htbp]
    \centering
    \includegraphics[width=0.8\textwidth]{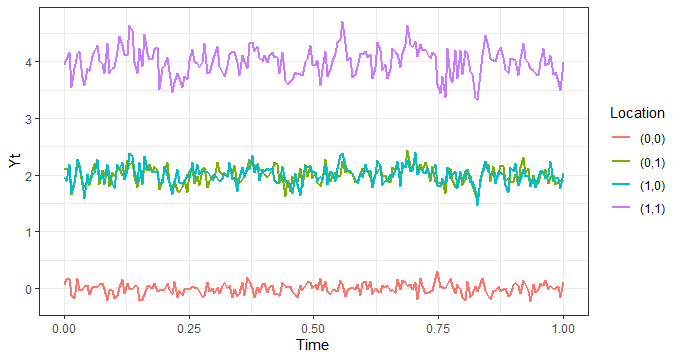}
    \caption{A sample realization of the simulated data for $n = 200$ and $p = 2$ (i.e., $4$ locations)}
    \label{fig:sim-sample}
\end{figure}

After we simulate the data from the above model with $n = 100$ and $p = 15$, we use the proposed methodologies to predict the mean function $\hat{\mu}$, and as a result, obtain prediction for $\widehat{Y}_{t}(s)$ for the last $10$ timepoints (among those $n = 100$) for all locations using the first $90$ timepoints as observations. About the choice of the covariate, we make three different situations as follows:
\begin{enumerate}
    \item[(S1)] At timepoint $t$, the original covariate $X_t$ is known, i.e., we can estimate $\widehat{Y}_t(s)$ by $\widehat{\mu}(X_t, s)$.
    \item[(S2)] At timepoint $t$, the original covariate $X_t$ is not known, but only its lagged value $X_{t-1}$ is known. Since, $X_t = 1 + 0.5X_{t-1} + e_t$, it follows that
    \begin{equation*}
        Y_t(s) = (s_1 + s_2) + 0.5X_{t-1}(s_1 + s_2) + \epsilon_t(s) + (s_1+s_2)e_t,
    \end{equation*}
    which is also covered by the proposed model in \eqref{eqn:model}.
    \item[(S3)] At timepoint $t$, only the lagged value of response $Y_{t-1}$ is known. Note that, the average of $Y_{t}(s)$ over the spatial grid is equal to $\sum_{s \in \Scal_p}Y_t(s) = X_t \times C + \sum_{s \in \Scal_p}\epsilon_t(s)$, where $C$ is a fixed deterministic constant. Therefore, the knowledge of $Y_{t-1}$ can be incorporated into model \eqref{eqn:model} by taking a proxy covariate $X'_t$ to be the average of lagged responses across all locations, which reduces it to the former case (S2).
\end{enumerate}

We consider $B = 100$ replications of this dataset, and for each replication we compute the estimates. For demonstration purposes, we tabulate the metrics such as Bias, Mean Absolute Error (MAE), Root Mean Squared Error (RMSE) and Mean Absolute Percentage Error (MAPE) for each timepoints across all replications and all locations in the Table~\ref{tab:simulations-transposed}.

\begin{table}[htbp]
    \centering
    \resizebox{\textwidth}{!}{
    \begin{tabular}{llrrrrrrrrrr}
        \toprule
        & & \multicolumn{10}{c}{\textbf{Time Index ($t$)}} \\ 
        \cmidrule(l){3-12}
        \textbf{Scenario} & \textbf{Metric} 
        & \textbf{0.909} & \textbf{0.919} & \textbf{0.929} & \textbf{0.939} & \textbf{0.949} 
        & \textbf{0.960} & \textbf{0.970} & \textbf{0.980} & \textbf{0.990} & \textbf{1.000} \\
        \midrule
        \multirow{4}{*}{(S1)} 
            & Bias       & -0.100 & -0.106 & -0.104 & -0.103 & -0.0980 & -0.0950 & -0.107 & -0.0981 & -0.108 & -0.102 \\
            & MAE        &  0.439 &  0.439 &  0.440 &  0.440 &  0.439  &  0.438  &  0.440 &  0.437  &  0.441 &  0.439 \\
            & RMSE       &  0.543 &  0.543 &  0.544 &  0.544 &  0.542  &  0.540  &  0.544 &  0.540  &  0.545 &  0.543 \\
            & MAPE (\%)  & 25.5   & 25.2   & 25.4   & 25.3   & 25.5    & 25.5    & 25.3   & 25.4    & 25.2   & 25.4 \\
        \midrule
        \multirow{4}{*}{(S2)} 
            & Bias       & -0.0935 & -0.0999 & -0.0923 & -0.0971 & -0.0919 & -0.0952 & -0.0993 & -0.0985 & -0.0920 & -0.0961 \\
            & MAE        &  0.444  &  0.445  &  0.446  &  0.444  &  0.445  &  0.445  &  0.447  &  0.443  &  0.449  &  0.447 \\
            & RMSE       &  0.550  &  0.551  &  0.553  &  0.550  &  0.550  &  0.551  &  0.554  &  0.548  &  0.558  &  0.555 \\
            & MAPE (\%)  & 25.9    & 25.7    & 26.0    & 25.7    & 25.9    & 25.8    & 25.8    & 25.6    & 26.1    & 25.9 \\
        \midrule 
        \multirow{4}{*}{(S3)} 
            & Bias       & -0.111 & -0.154 & -0.148 & -0.152 & -0.153 & -0.167 & -0.157 & -0.170 & -0.173 & -0.176 \\
            & MAE        &  0.449 &  0.566 &  0.567 &  0.570 &  0.588 &  0.595 &  0.596 &  0.606 &  0.606 &  0.621 \\
            & RMSE       &  0.558 &  0.704 &  0.704 &  0.708 &  0.732 &  0.743 &  0.744 &  0.759 &  0.759 &  0.784 \\
            & MAPE (\%)  & 25.6   & 32.8   & 33.1   & 33.1   & 34.7   & 34.9   & 35.6   & 36.2   & 35.5   & 37.6 \\
        \bottomrule
    \end{tabular}
    }
    \caption{Simulation results for scenarios (S1)-(S3) with $n = 100, p = 15$.}
    \label{tab:simulations-transposed}
\end{table}

It is also interesting to see how the prediction error would decrease as a function of the $p$, the number of gridpoints on the spatial horizon. Figure~\ref{fig:simdata-p-vary} demonstrates this behaviour which shows that the MAPE decreases as the number of spatial locations increases, as assured by the consistency result in Theorem~\ref{thm:pointwise-full}.

\begin{figure}[htbp]
    \centering
    \includegraphics[width=0.8\textwidth]{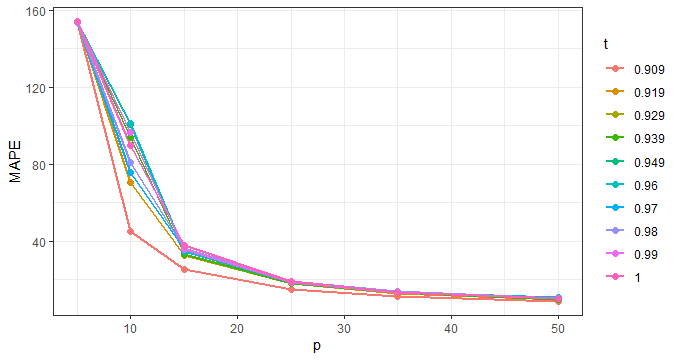}
    \caption{Resulting MAPE in predicting different timepoints as a function of $p$, the number of grid points}
    \label{fig:simdata-p-vary}
\end{figure}

\end{document}